\documentclass[reqno,12pt]{article}

\usepackage{subcaption}  % For subfigures with captions
\usepackage{caption}     % Optional, for customizing captions
\usepackage{bm}
\usepackage{mathtools}
\mathtoolsset{showonlyrefs}
\usepackage{algorithm}
\usepackage{algpseudocode}

\usepackage{array}
\usepackage{mdframed}
\usepackage{amsfonts, amsmath, amssymb, amsthm}
\usepackage{color}
\usepackage{dsfont}
\usepackage[T1]{fontenc}
\usepackage{authblk}
\usepackage{graphicx,tikz,xcolor}
\usepackage{xcolor,eucal,enumerate,mathrsfs}
\usepackage[normalem]{ulem}
\usepackage[latin1]{inputenc}
\usepackage{cancel}
\usepackage{setspace}
\usepackage{stmaryrd}
\usepackage{stackrel}
\usepackage{longtable}

\usepackage{amsmath,amssymb}

\usepackage{hyperref}

\usepackage{mathtools}
\newcommand{\eps}{\varepsilon}

\newcommand{\tand}{\quad \text{and} \quad}

\newcommand{\imag}{\mathbf{i}}

\newcommand{\bra}[1]{\left \langle #1 \right |}
\newcommand{\ket}[1]{\left | #1 \right \rangle}

\DeclareMathOperator{\argmax}{argmax}

\DeclareMathOperator{\ep}{\varepsilon}

\DeclareMathOperator{\Tr}{\mathrm{Tr}}

\DeclareMathOperator{\qAdm}{\rm{Adm}(\mathbf{Q},\mathbf{q})}
    \DeclareMathOperator{\uAdm}{\underline{\rm{Adm}}(\rm{\mathbf{Q}}, \rm{\mathbf{q}})}

\newcommand{\cH}{\mathcal{H}}

\newcommand{\cF}{\mathcal{F}}

\newcommand{\fr}{\penalty-20\null\hfill$\blacksquare$} 

\def\R{\mathbb R}
\newcommand{\rmH}{\mathrm{H}}

\usepackage{appendix}
\usepackage{enumitem}

\newcommand{\suchthat}{\ensuremath{\ : \ }} % such that inside, for example, the sets definition              

\newcommand{\restr}[1]{\lower3pt\hbox{$|_{#1}$}}

\newcommand{\Id}{\mathbb I}

\newcommand{\prim}{\mathfrak{F}_{\eps}(\rm{\mathbf{Q}}, \rm{\mathbf{q}}) }
    \newcommand{\primal}{\tF_{\eps}}
\newcommand{\dual}{\mathfrak{D}_{\eps}(\rm{\mathbf{Q}}, \rm{\mathbf{q}})} 
    \newcommand{\dualf}{{\rm D}_\eps}

\newcommand{\primz}{\mathfrak{F}({\rm \mathbf{Q}}, {\rm \mathbf{q}}) }
    \newcommand{\primalz}{\tF}
\newcommand{\dualz}{\mathfrak{D}({\rm \mathbf{Q}}, {\rm \mathbf{q}})} 
    \newcommand{\dualfz}{{\rm D}}

\theoremstyle{plain}
\newtheorem{theo}{Theorem}[section]
\theoremstyle{definition}
\newtheorem{defin}[theo]{Definition}
\newtheorem{rem}[theo]{Remark}

\newtheorem{oss}[theo]{Remark}
\newtheorem{exam}[theo]{Example}
\numberwithin{equation}{section}
\theoremstyle{plain}

\newtheorem{prop}[theo]{Proposition}
\theoremstyle{remark}

\newcommand{\tF}{\operatorname{F}}

\newcommand{\tH}{\operatorname{H}}

\usepackage{wrapfig}

\DeclareFontFamily{U}{mathx}{\hyphenchar\font45}
\DeclareFontShape{U}{mathx}{m}{n}{
	<5> <6> <7> <8> <9> <10>
	<10.95> <12> <14.4> <17.28> <20.74> <24.88>
	mathx10
}{}
\DeclareSymbolFont{mathx}{U}{mathx}{m}{n}
\DeclareMathSymbol{\bigtimes}{1}{mathx}{"91}

\usepackage{scalerel}

\usepackage{color}
\usepackage{pstricks}
\usepackage{ifthen}
\newrgbcolor{kascolor}{.7 .1 .1}
\newrgbcolor{olicolor}{.1 .7 .1}
\newrgbcolor{augucolor}{.1 .1 .7}
\newrgbcolor{lorcolor}{.5 .5 .5}
\usepackage{csquotes}

\usepackage{booktabs}   % for \toprule, \midrule, \bottomrule
\usepackage{multirow}   % for \multirow

\newboolean{DEBUG}
\setboolean{DEBUG}{true}

\ifthenelse {\boolean{DEBUG}}
{}

\ifthenelse {\boolean{DEBUG}}
{}

\ifthenelse {\boolean{DEBUG}}
{}

\ifthenelse {\boolean{DEBUG}}
{}

\ifthenelse {\boolean{DEBUG}}
{}

\definecolor{pavcolor}{RGB}{230, 40, 150} 
\ifthenelse {\boolean{DEBUG}}
{}

\usepackage{subcaption}
\usepackage{array,booktabs}
\newcolumntype{L}{>{\centering\arraybackslash}p{0.1\linewidth}} % Reg.
\newcolumntype{T}{>{\centering\arraybackslash}p{0.20\linewidth}} % Tol.
\newcolumntype{Y}{>{\centering\arraybackslash}p{0.12\linewidth}} % Type
\newcolumntype{N}{>{\centering\arraybackslash}p{0.20\linewidth}} % Time
\newcolumntype{I}{>{\centering\arraybackslash}p{0.15\linewidth}} % Iters
\newcolumntype{A}{>{\centering\arraybackslash}p{0.15\linewidth}} % Ach.

\usepackage{graphicx}
\usepackage{wrapfig}
\usepackage{needspace}   % to keep enough space on the page
\usepackage{placeins}    % optional: for \FloatBarrier

\title{Quantum thermodynamics and semidefinite programming: regularization and algorithms}
\author[1]{Emanuele Caputo}
\author[2]{Augusto Gerolin}
\author[3]{Nataliia Monina} 
\author[3]{\\ Pavlo Pelikh}
\author[4]{Lorenzo Portinale}

\affil[1]{University of Warwick}
\affil[2]{Instituto de Matem\'atica Pura e Aplicada}
\affil[3]{Department of Mathematics and Statistics, University of Ottawa}
\affil[4]{Universit\`a statale di Milano}

\date{}
\begin{document}
\maketitle
\begin{abstract}

\noindent
We investigate variational problems in quantum thermodynamics at positive temperature, in which admissible states are constrained by prescribed outcomes of a finite set of measurements. We solve a problem raised by the recent work \cite[Section C]{MarkWilde} and develop a general mathematical setup which allows a broad class of possible regularizations.
Employing methods inspired by non-commutative optimal transport, 
we analyze the dual formulation of the problem, study the existence and characterization of maximizers, and investigate the qualitative behavior of the model in the zero-temperature limit.
In the second part, we develop computational algorithms tailored to this class of variational problems and demonstrate the effectiveness of our approach through applications to quantum information tasks, including quantum state tomography and quantum optimal transport.
\end{abstract}

\setcounter{tocdepth}{1}
\tableofcontents

\section{Introduction}
\label{sec:intro}

Let $M\in\mathbb N$ be a natural number, $\varepsilon > 0$ be a positive real number, $\cH$ be a finite-dimensional Hilbert space, $H, Q_0,Q_1,\dots,Q_M \in \rmH(\cH)$ be Hermitian matrices and $(q_0,q_1,\dots,q_M)\in\mathbb{R}^{M+1}$ be real numbers. In this paper we are interested in the following class of variational problems
\begin{align}
\label{intro:main}
\cF_\eps(\textbf{Q},\textbf{q}) 
    &:= 
\inf 
\big\{
    \primal(\pi)
        \suchthat 
    \pi \in \qAdm
\big\},
\end{align}
where the target functional and the admissible operators are given by 
\begin{align} 
    \primal(\pi)
        :=
    \Tr[H\pi] + \varepsilon S_{\varphi}(\pi) 
        \, , \quad  
    \qAdm
        :=
    \big\{
        \pi\in \rmH_\geq(\cH)
            \suchthat 
        \Tr[Q_i\pi] = q_i, \, 0\leq i\leq M
    \big\}
        . 
\end{align}
 Here $\rmH_\geq(\cH)$ denotes the set of Hermitian and semi-definite positive operators. The function $\varphi:[0,+\infty) \to \mathbb{R}$ is a proper, convex, superlinear at infinity, bounded from below function with $\varphi(0) \in \R$, and $S_{\varphi}:\rmH_\geq(\cH)\to\mathbb{R}$ is the quantum entropy induced by $\varphi$, i.e. $S_{\varphi}(\pi) = \Tr[\varphi(\pi)]$.  
 
In words, the goal of the variational problem defined in~\eqref{intro:main} is to find, within the class of all operators which \textit{measured} via $Q_i$ give the result $q_i$, the ones that minimize the total energy $\primal$, namely a noised version of the ground state energy associated with the Hamiltonian $H$.

In~\cite{ChhLas-arxiv-2025, Lin-arxiv-2023, MarkWilde}, the authors consider the case when $\varphi(z)=z\ln z$ and the variational problem~\eqref{intro:main} corresponds to the von Neumman entropy regularization of the ground state energy associated with the Hamiltonian $H$. In this case, the solution of that problem at a giving temperature $\varepsilon>0$ is giving by a Gibbs state of an effective Hamiltonian\vspace{-2mm}
\begin{equation}\label{intro:Gibbsstate}
\pi_{\varepsilon} = \exp\left(-\frac{1}{\varepsilon} H_{\mathrm{eff}}\right), \quad \text{where } H_{\mathrm{eff}} = H-\sum_{i=0}^M\alpha^{\varepsilon}_iQ_i,\vspace{-2mm}
\end{equation}
where $(\alpha_i^{\varepsilon})_{i=0}^M$ are real numbers and can be interpreted as Lagrange multipliers associated with the constraints. In this sense, $H_{\mathrm{eff}}$ plays the role of an effective Hamiltonian, modified by the constraints, such that the corresponding Gibbs state $\pi_{\varepsilon} \in \qAdm$ belongs to the admissible class of states.

A motivation for our work is provided by~\cite[Section C]{MarkWilde}, where the authors emphasize the importance of developing a rigorous mathematical theory and computational algorithms capable of handling regularization schemes beyond the von Neumann entropy. In particular, they propose hybrid quantum--classical algorithms to solve the quantum optimization problem~\eqref{intro:main} under von Neumann entropy regularization. While substantial progress has been achieved for Gibbs state-based methods -- largely due to their variational characterization and the relative maturity of thermal state preparation techniques -- the extension beyond the Gibbs energy/von Neumann entropy regularization framework poses significant conceptual and technological challenges.

In this paper, we address the question raised in the aforementioned work by developing a general mathematical framework for convex regularizations, encompassing, for example, quantum quadratic regularizations and Tsallis-type entropies. %More precisely, we provide a characterization of the solutions to~\eqref{intro:main}, design corresponding (classical) computational algorithms, and establish their convergence in the zero-temperature limit $\varepsilon \to 0$.

 Our approach is grounded in duality theory and builds on mathematical tools that have recently been developed and successfully applied in the study of quantum optimal transport (QOT) problems, a special case of the variational problem~\eqref{intro:main}~\cite{CapGerMonPor24, FelGerPor23,GerMon2023, GerPel25} (with $\varepsilon>0)$ and \cite{CalGolPau18, DPaTre19,DPaTre2023,Golse-Mouhot-Paul:2016} (with $\varepsilon=0)$. We refer to Section~\ref{subsec:qot} for more details and also \cite{portinale2024entropic} for a review on the subject.

Within this framework, we establish several results for the variational problem~\eqref{intro:main}, including:
(i) the derivation of an equivalent dual formulation, together with a detailed analysis of both primal and dual optimizers and a clarification of the relations between them; 
(ii) the investigation of the asymptotic behavior of the problem and its optimizers as $\varepsilon \to 0$, recovering in the limit a corresponding quantum formulation at zero temperature, whose structural properties and minimizers are analyzed.

Finally, we illustrate our methodology by proposing a novel algorithm computing numerical realizations of the dual problem of~\eqref{intro:main} via L-BFGS for both the von Neumman entropy and quadratic quantum entropy regularizations. From our knowledge, (classical) algorithms developed previously focus on the particular case of the von Neumman entropy and are based on first-order methods for dual problem~\cite{CaiLin-arxiv-2025, Lin-arxiv-2023, GerPel25, RanRen-arxiv-2025}. 

A promising direction for future development lies in extending hybrid quantum--classical algorithms beyond the von Neumann entropy setting~\cite{MarkWilde} to encompass more general convex regularizations. Achieving this goal would considerably broaden the scope of quantum optimization methods, enabling the treatment of models with quadratic penalties, Tsallis-type entropies, or other non-Gibbsian regularizers that arise naturally in quantum information theory, statistical mechanics, and machine learning.

\subsection{Contributions and organization of the paper}
For the reader's convenience, Section~\ref{sec:setting_mr} presents the complete mathematical framework and the precise statements of the main results of the paper. The remainder of the paper is devoted to their proofs. Below, we provide an overview of our contributions and the organization of this work. 

\smallskip 
\noindent
\textbf{Well-posedness of the problem (Section~\ref{sec:preliminaries})}. \ 
Naturally, conditions may need to be imposed on the observables $Q_i$ and the outcomes $q_j$ in order to have a nonempty set of admissible operators. For example, assuming that $Q_0, Q_1, \dots, Q_{j_0}$ are a basis for the set of all observables $\{Q_i \}_{i=0}^M \subset \cH(\tH)$, we show that a necessary condition for $\qAdm \neq \emptyset$ is given by 
\begin{align}
\label{eq:necessary_intro}
    Q_j
        =
    \sum_{i=0}^{j_0}
        t_j^i Q_i 
            \quad \Longrightarrow \quad 
      q_j 
            = 
        \sum_{i=0}^{j_0}
            t_j^i q_i    
                \, ,
\end{align}
we refer to Proposition~\ref{prop:nonmptiness} for more details.

\smallskip
\noindent 
\textbf{Analysis of the regularized problems (Section~\ref{sec:main_duality})}. \ 
Our first result describes the variational problem~\eqref{intro:main} via a suitably (unconstrained when $\eps>0$) dual formulation -- see Theorems~\ref{thm:max_exist} and \ref{thm:characterization_maximizers} -- and states that
\begin{gather}\label{intro:maindual}
\cF_\eps(\textbf{Q},\textbf{q}) 
    = 
\sup \big\lbrace \dualf(\alpha_0,\dots,\alpha_M) \, : \, \alpha_i\in\R, \, 0\leq i\leq M \big\rbrace
    \, , \quad \text{where} \\
\dualf(\alpha_0,\dots,\alpha_M) 
    :=
\sum^M_{i=0}\alpha_iq_i - \varepsilon \Tr\left[\psi\left(\frac{1}{\varepsilon}\left(\sum^M_{i=0}\alpha_iQ_i - H\right)\right)\right],
\end{gather}
where $\psi=\varphi^*$ denotes the Legendre transform of $\varphi$. Furthemore, we characterize optimality for the dual and primal problems, and show that primal and dual optimizers are related by the \textit{complementary slackness} condition 
\begin{align}
     \pi_{\varepsilon}=\psi' \left( \frac{\sum_{i=0}^M\alpha_iQ_i - H}{\varepsilon} \right)
     .
\end{align}
Notice that, by choosing the von Neumann entropy regularization, we recover a Gibbs-state form described in equation~\eqref{intro:Gibbsstate}.

\smallskip
\noindent 
\textbf{Convergence results at zero temperature (Section \ref{sec:finitetemplimit})}. In a subsequent section, we investigate the asymptotic regime of the above variational problems as the temperature vanishes, that is, in the limit $\varepsilon \to 0^+$. This corresponds to the physically relevant zero-temperature limit and leads to a qualitatively different (primal and dual) optimization problems, where entropic effects disappear and only energetic constraints remain.

To rigorously analyze this singular limit, we employ the framework of \emph{$\Gamma$-convergence} (see Section~\ref{subsec:convergence_dual_problem}). For readers less familiar with this notion, we stress that $\Gamma$-convergence is a variational convergence theory specifically designed to study the stability of minimization problems. Its strength lies in the fact that it guarantees not only the convergence of the functionals themselves, but the convergence of minimizers. In other words, it provides a mathematically robust way to pass to the limit in optimization problems.

Within this framework, we prove that, as $\varepsilon \to 0^+$, both the primal and dual formulations converge to the zero-temperature problem
\begin{align}
     \primz 
        &:=
     \inf \left\lbrace \Tr[H\pi]  \suchthat  \pi \in {\rmH}_{\ge}(\cH), \, {\rm Tr}[Q_i\pi] = q_i, \, 0\leq i\leq M \right\rbrace 
\\
        &=
    \sup \left\lbrace \sum_{i=0}^M\alpha_i q_i \, : \,  \bm\alpha \in \R^{M+1} ~ \text{so that} ~ \sum_{i=0}^M \alpha_iQ_i \leq H, \right\rbrace
        \, , 
\end{align}
where with the latter inequality constraint between operators is intended in the sense of quadratic forms. This zero-temperature limit problem can be interpreted as a semi-definite programming problem.

We also deduce the convergence (up to subsequences) of minimizers $\Gamma_\varepsilon$ toward a minimizer $\Gamma_0$ of the zero-temperature problem. This highlights the practical relevance of the theory: the zero-temperature model is not merely a formal limit, but genuinely describes the asymptotic behavior of optimal states of the regularized problem.

Finally, analogously to the positive temperature case, we show that, at zero temperature, optimizers are characterized by a complementary slackness conditions, which in this case assumes the form of 
\begin{align}
    \sqrt \pi \left( H - \sum_{i = 0}^M \alpha_i Q_i \right) \sqrt \pi= 0
        \, .
\end{align}

% {\color{blue}START}

\smallskip
\noindent
\textbf{Numerical results (Section~\ref{sec:numerics})}. \ 
The last section contains numerical simulations for the computation of the maximizer of the dual problem at positive temperature, by means of L-BFGS~\cite{LiuNoc-MathPr-1989}. We specialize our experiments to quantum optimal transport and quantum tomography.
We qualitatively analyze convergence behavior for several values of the temperature parameter~$\varepsilon$. 
The numerical results illustrate both the practical applicability of the proposed methods and the computational challenges that may arise.

% {\color{blue}END.}

\section{Setting and statement of the main results}
\label{sec:setting_mr}

Given a Hilbert space $\cH$ of dimension $d = \dim(\cH) \in \mathbb{N}$, we denote by $\rmH(\cH)$ the vector space of Hermitian operators over $\cH$.
Let us denote by $\rmH_\geq(\cH)$ the set of positive semi-definite Hermitian, and by 
$\rmH_>(\cH)$ the elements $A$ of $\rmH_\geq(\cH)$ such that ${\rm ker}(A)=\{ 0\}$.

For convenience, we introduce the following notation: for each $i \in \{1, \dots, d\}$, we denote by $\lambda_i(A) \in \R$ the $i$-th smallest eigenvalue of the Hermitian matrix $A$. In particular, $\lambda_1(A)$ represents the smallest eigenvalue of $A$ and $\lambda_d(A)$ the largest. 
Given a continuous function $g \colon \mathbb{R} \to \mathbb{R}$, we define its \textit{lifting} to the space of Hermitian matrices as the operator\footnote{With a slight abuse of notation, we use the same symbol $g$ for both the scalar function and its lifting.} 
$g \colon \rmH(\cH) \to \rmH(\cH)$ 
given by 
$
g(A) := \sum_{i=1}^d g(\lambda_i(A)) \, |\xi_i\rangle \langle \xi_i|
    , 
$
where $A = \sum_{i=1}^d \lambda_i(A) \, \ket{\xi_i}\bra{\xi_i}$ is a spectral decomposition of $A$. 
When the context is clear, we may omit the dependence on $A$ and simply write $\lambda_i$ instead of $\lambda_i(A)$.

We shall fix  $\varphi \colon [0,+\infty) \to \R$ a convex, superlinear at infinity, and bounded from below function, namely 
\begin{align}
    \label{eq:assumptions_varphi_intro}
    \varphi \colon [0,+\infty) \to \R
        \, , \quad 
    \text{convex} 
        \, , \quad 
    \lim_{t \to +\infty}
        \frac{\varphi(t)}{t} =+\infty
        \, ,  \tand 
    \inf \varphi \geq l > - \infty 
        , \qquad 
\end{align}
for some $l \in \R$ (in fact, the third condition follows from the first two). In particular, we assume that $\varphi(0)\in \R$.

The following operators are given and fixed throughout the paper: 
\begin{enumerate}
    \item A cost operator/Hamiltonian $H \in \rmH(\cH)$. 
    \item A set of $M+1$ observables, $M \in \mathbb{N}$, that we denote by $Q_0, Q_1, \dots, Q_M \in \rmH(\cH)$. The first one, $Q_0$, has a special role in our variational problem, and it is assumed to be positive definite (hence invertible).
    \item A given sequence of numbers $q_0, q_1, \dots, q_M$ that represents the outcomes of the measurements of a state through $Q_i$. Here $q_0 >0$.
\end{enumerate}
 Our analysis is then constrained to \textit{admissible states}, denoted by $\qAdm$ and given by 
\begin{align}
    \qAdm
        :=
    \big\{
        \pi\in \rmH_\geq(\cH)
            \suchthat 
        \Tr[Q_i\pi] = q_i, \, 0\leq i\leq M
    \big\}
        \, .
\end{align}

We also fix $\eps >0$ as a regularization parameter. 
The main objects of study of our work are primal and dual functionals. 
\begin{defin}[Primal problem]
We define the functional $\primal \colon \rmH_\geq(\cH) \to \mathbb{R}$ as
\begin{align}
\label{eq:def_primal}
    \primal(\pi): = \Tr\left[H \pi\right] + \eps \Tr \left[ \varphi (\pi) \right]
        , \qquad  \pi \in \rmH_\geq(\cH) ,
\end{align}
and define the primal problem as 
\begin{align}
\label{eq:def_primal_problem}
    \prim:= \inf \left\{ \primal(\pi) \suchthat  \pi \in \qAdm  \right\}
    \, .
\end{align}
\end{defin}

\begin{rem}[Boundedness of $\qAdm$]
\label{rem:boundedness}
    The choice of assuming $Q_0 \in \tH_>(\cH)$ comes from the need to ensure boundedness of the set of competitors. We have indeed
    \begin{align}
    \label{eq:boundedness_qAdm}
        \qAdm \subset B_{\mathfrak{T}_1}(0, t_0) :=
        \big\{
            \pi \in \tH_\geq(\cH) 
                \suchthat 
            \Tr[\pi] \leq t_0
        \big\}
            \qquad 
        \text{for } 
            t_0= \frac{q_0}{\lambda_1(Q_0)}
                , 
    \end{align}   
    where we recall that $\lambda_1(Q_0) \geq 0$ denotes the smallest eigenvalue of $Q_0$. This readily follows from the single constraint $q_0 = \Tr[Q_0 \pi]$ and $Q_0 \geq \lambda_1(Q_0) \Id$.
\end{rem}
In order to guarantee that the above minimization problem is meaningful, we need to ensure that the set of admissible plans is not empty, which in general is not guaranteed. 

Indeed, we shall prove that the set of observables $\{Q_i \}_{i=0}^M$ and outcomes $\{ q_i \}_{i=0}^M$ must necessarily satisfy a compatibility condition, in order to have $\qAdm \neq \emptyset$. 

More precisely: we define the set 
\begin{align}
    \uAdm := 
    \big\{
        \pi \in \tH(\cH) 
            \suchthat 
        \Tr[Q_i \pi] = q_i, 
            \, 0 \leq i \leq M
    \big\}
        , 
\end{align}
so that by definition $\qAdm = \uAdm \cap \tH_{\geq}(\cH) \subset \uAdm$. Our first result yields a characterization for the nonemptiness of $\uAdm$, hence providing a necessary condition for $\qAdm \neq \emptyset$. 
Recall that $\tH(\cH)$ endowed with the usual scalar product $(A,B) \mapsto \Tr[AB]$ is a \textit{real} vector space, and as such when writing \textit{linear independent} below, we mean with respect to linear combination with \textit{real} coefficients.

\begin{prop}[Non-emptiness of $\uAdm$]
\label{prop:nonmptiness}
    Let $\mathbf{Q} = (Q_0,Q_1, \dots, Q_M) \in \rmH(\cH)^{M+1}$ be a vector of Hermitian matrices, and $\mathbf{q} = (q_0,q_1, \dots, q_M ) \in \R^{M+1}$.  
\begin{enumerate}
    \item[1)] If the operators $\{ Q_j \}_{j=0}^M$ are linearly independent, $\uAdm \neq \emptyset$ for every $\mathbf{q} \in \R^{M+1}$.
    \item[2)] Assume instead that, up to a reordering, there exists $j_0 \in [M]$ such that $\{ Q_0, Q_1 \dots, Q_{j_0} \}$ is a system of linearly independent operators such that 
    \begin{align}
    \label{eq:def_VQ}
        \mathcal{V}_{\mathbf{Q}}
            :=
        \emph{Span} 
        \big(
            \{ Q_0, Q_1 \dots, Q_M \}
        \big)
            =
        \emph{Span} 
        \big(
            \{ Q_0, Q_1 \dots, Q_{j_0} \}
        \big)
            \, .
    \end{align}
    In particular, for every $j > j_0$, there exists $\{ t_j^i \, : \, i \in [j_0] \}\subset \R$ such that 
    \begin{align}
    \label{eq:linear_dependence_of_Qi}
        Q_j
            =
        \sum_{i=0}^{j_0}
            t_j^i Q_i 
                \, .
    \end{align}
    Then we have that
    \begin{align}
    \label{eq:iff_nonempty}
        \uAdm \neq \emptyset
            \qquad \Longleftrightarrow \qquad 
        q_j 
            = 
        \sum_{i=0}^{j_0}
            t_j^i q_i 
                \, , \quad 
            \forall j > j_0 
                \, .
    \end{align}
\end{enumerate}
\end{prop}
 The condition \eqref{eq:iff_nonempty} is a consistency condition for the $\{q_j \}_{j=0}^M$ which must be satisfied to have a nontrivial primal problem, in the case when the $\{ Q_j \}_{j=0}^M$ are not linearly independent. 
In particular, a \textit{necessary condition} for $\qAdm \neq \emptyset$ is that 
\begin{align}
\label{eq:necessary}
    Q_j
        =
    \sum_{i=0}^{j_0}
        t_j^i Q_i 
            \quad \Longrightarrow \quad 
      q_j 
            = 
        \sum_{i=0}^{j_0}
            t_j^i q_i    
                \, .
\end{align}
Note for example that, if $Q_i = Q_j$ for some $i,j \in [M]$ corresponds to the case $t_j^k = \delta_{i,k}$, and therefore forces $q_i = q_j$ (which is obviously necessary to have a nonempty set of admissible operators). Whenever the set of observables $\{Q_i\}_{i=0}^M$ is made of linearly independent operators, this condition is always satisfied. 

To each primal problem, we associate a dual problem. Throughout the whole paper, we work with a function $\psi \in C(\R)$ which is convex, superlinear at infinity, and bounded from below, i.e. 
\begin{align}
\label{eq:assumptions_psi_intro}
    \psi\in C(\R) 
        \, , \quad \text{convex} \, , \quad 
    \lim_{t \to +\infty}
        \frac{\psi(t)}{t}
    = + \infty \, , 
        \tand 
    m:= \inf \psi > - \infty 
        .
\end{align}
When dealing with duality results for the primal problem \eqref{eq:def_primal}, $\psi$ is typically the Legendre transform of a $\varphi$ satisfying \eqref{eq:assumptions_varphi_intro}, namely of the form
\begin{align}
    \psi(t) 
        =
    \sup_{x \in [0,+\infty)}
        \left\{
            tx - \varphi(x)
        \right\}
        =
    \varphi^*(t)
        \, , \quad \forall t \in \R
            ,
\end{align}
where, when talking about Legendre transform, we may implicitly extend the definition of $\varphi$ on the full real line by setting $\varphi(x) \equiv +\infty$, for every $x < 0$. The validity of \eqref{eq:assumptions_psi_intro} readily follows in this case from the properties \eqref{eq:assumptions_varphi_intro} of $\varphi$.

\begin{defin}[Dual problem]
For $\psi \in C(\R)$ satisfying \eqref{eq:assumptions_psi_intro},  we define the dual functional $\dualf  \colon \R^{M+1} \to \R$ as
\begin{align}
\label{eq:def_dual_balanced}
    \dualf (\bm{\alpha})
        :=
\sum^M_{i=0}\alpha_iq_i - \varepsilon \Tr\left[\psi\left(\frac{1}{\varepsilon}\left(\sum^M_{i=0}\alpha_iQ_i - H\right)\right)\right]
        , \qquad 
    \bm\alpha \in \R^{M+1} 
        .
\end{align}

We define the dual problem as 
\begin{gather}
\label{eq:def_dual_problems}
    \dual :=
        \sup 
    \Big\{
        \dualf (\bm\alpha)
    \suchthat 
            \bm \alpha \in \R^{M+1}
    \Big\}
        .  
\end{gather}
\end{defin}

The contribution of this paper is threefold: first of all, we show that primal and dual problems do indeed coincide, and we discuss under which assumptions optimizers exist and how they are related. Secondly, we provide an asymptotic analysis as $\eps \to 0$, describing the associated limit problems and their minimizers/maximizers. Finally, we discuss, together with several simulations, different applications of our setup, including quantum optimal transport and quantum tomography. This is the content of Section~\ref{sec:numerics}.
   
We start by presenting the duality result and the analysis of the optimizers.

\begin{theo}[Duality and optimizers]
\label{thm:main_duality}
    Let $\varphi:[0,+\infty) \to \R$ satisfy \eqref{eq:assumptions_varphi_intro}, and assume that $\psi=\varphi^*$ is $C^1$. 
    Moreover, we assume that 
    \begin{align}
        \label{eq:Slater_condition}
    \qAdm \cap \tH_>(\cH) \neq \emptyset
        \, .
    \end{align}
    Then the following statements hold: 
    \begin{itemize}
        \item[1)] (Duality) The dual and primal problems coincide, namely $\dual  = \prim$.
        \item[2)] (Existence of maximizers) There exists $\bm\alpha \in \R^{M+1}$ such that  $\dualf(\bm\alpha)=\dual$.
        \item[3)] (Existence of minimizers) There exists a unique minimizer $\pi^\eps \in \qAdm$ for the primal problem, i.e. $\prim = \primal(\pi^\eps)$, and it is given by
        \begin{align}
        \label{eq:main_formula_minimizer}
            \pi^\eps
                = 
            \psi' \left( \frac{\sum^M_{i=0}\alpha_iQ_i - H}{\varepsilon} \right)
                ,
        \end{align}
        where $\bm\alpha \in \R^{M+1}$ is (any) maximizer for the dual problem. 
    \end{itemize}
\end{theo}

The assumption $\qAdm \cap \tH_>(\cH) \neq \emptyset$ means, in other words, that there exists $\pi_0\in\mathrm{H}(\cH)$ such that 
\[
   \inf_{\xi\in\cH, ||\xi||=1}{\bra\xi \pi_0\ket\xi}=\omega_0>0\quad\quad \text{and} \quad\quad q_i = \Tr[Q_i \pi_0], \quad 0\leq i\leq M.
\]
This phenomenon has already been observed in the setting of QOT \cite{FelGerPor23,CapGerMonPor24}, where the existence of (bounded) maximizers was related to the fact that the marginals have no kernels, which is exactly what ensures the existence of at least one coupling with no kernel. 

The dual problem does not have, in general, uniqueness of the maximizers, even when $\psi$ is strictly convex. This happens precisely when the observables $\{Q_i\}_{i=0}^M$ are not linearly independent. This  is clearly related to the fact that the map 
\begin{align}
    \R^{M+1} \ni \bm\alpha \mapsto 
    \sum_{i=0}^{M}
        \alpha_i Q_i
            \in \tH(\cH)
\end{align}
is injective if and only if the observables $\{Q_i\}_{i=0}^M$ are linearly independent (note that here we mean as elements of $\tH(\cH)$ seen as a \textit{real} vector space). Whenever the operators are linearly dependent, one has infinitely many maximizers for the dual problem, that can be explicitly described as follows (see  Proposition~\ref{prop:uniqueness_and_continuity}): 
assume there exists $j_0 \in \mathbb{N}$, $\{ t_j^i \}_{i,j}\subset \R$ as given in Proposition~\ref{prop:nonmptiness}, for which we have
\begin{align}
     Q_j
        =
    \sum_{i=0}^{j_0}
        t_j^i Q_i 
            \tand 
    Q_0 , \dots, Q_{j_0}
        \quad \text{linearly independent}
    .
\end{align}
Nonetheless, whenever $\psi$ is strictly convex\footnote{When $\psi$ is not strictly convex, additional nonuniqueness may appear, as per usual.}, we can explicit determine the set of all maximizers as 
\begin{align}
\label{eq:def_argmax}
    \argmax \dualf
        =
    \left\{
        \alpha \in \R^{M+1}
            \suchthat 
        \alpha_i + \sum_{j=j_0+1}^M \alpha_j t_j^i 
            = 
        \alpha_i^\eps
            \, , \, 
        \text{for all }
            i \in \{ 0, \dots, j_0 \}
    \right\}
        ,
\end{align}
where $\alpha^\eps \in \R^{M+1}$ is the unique maximizer of $\dualf$ so that $\alpha_i^\eps =0$ for every $i > j_0$. For this special choice of maximizers, we also prove that $\eps \mapsto \alpha^\eps$ is continuous on $(0,+\infty)$ (cfr. Proposition~\ref{prop:uniqueness_and_continuity}).

Finally, the equality between primal and dual problems does not require $\psi$ to be of class $C^1$. Nonetheless, in this case, we cannot guarantee uniqueness for the minimizers of $\primal$. Optimality conditions which relate minimizers of $\primal$ and maximizers of $\dualf$ can still be written, in the more general form 
\begin{align}
\label{eq:main_formula_minimizer_nonsmooth}
    \pi^\eps
        \in 
    \partial \psi \left( \frac{\sum^M_{i=0}\alpha_iQ_i - H}{\varepsilon} \right)
        ,
\end{align}
where $\partial \psi$ denotes the convex subdifferential of $\psi$. For details, we refer to Proposition~\ref{prop:weak-duality}.

\begin{rem}[Duality]
    In fact, by using abstract tools such as Fenchel--Rockafellar's theorem, one would be able to show that $\prim = \dual$ in the more general setting where $\qAdm \neq \emptyset$ (no need to have a competitor with a trivial kernel). On the other hand, in this setting, the maximizers may not exist, as the maximizing sequences $\alpha_n$ are not expected to be bounded. As we are mostly interested in the characterization of the maximizers in this work, we focus on the setting where the stronger assumption \eqref{eq:Slater_condition} is satisfied.
\end{rem}

Moving to the asymptotic analysis, the second main contribution of this article is the analysis of the primal and dual problems, as well as their optimizers, as $\eps \to 0$. 
We introduce the functionals, for $\pi \in \tH_\geq(\cH)$ and $\bm\alpha \in \R^{M+1}$,
\begin{equation}
\label{eq:functionals_zero}
    \primalz[\pi] := \Tr[H \pi]
        \, , \qquad \text{and} \qquad   
    \dualfz[\bm\alpha] := \sum_{i=0}^M \alpha_i q_i - \chi[W_{\bm\alpha}]
    , \quad  
    W_{\bm\alpha} := \sum_{i=0}^M \alpha_i Q_i - H
        ,
\end{equation}
where $\chi:\tH(\cH) \to \R \cup \{ + \infty \}$ is the indicator function 
\begin{align}
    \chi(W)  = 
    \begin{cases}
        0
            &\text{if } W \leq 0 , 
    \\
        + \infty    
            &\text{otherwise}. 
    \end{cases}
\end{align}
Finally, we consider the following variational problems:
\begin{gather}
\label{eq:problems_zero}
    \primz 
        := 
    \inf
    \big\{
        \primalz[\pi]
            \suchthat 
        \pi \in \qAdm 
    \big\}
         , 
\\
    \dualz 
        := 
    \sup
    \bigg\{
        \dualfz[\bm\alpha]
            \suchthat 
    \bm\alpha \in \R^{M+1}
    \bigg\}  
        =
   \sup
    \bigg\{
     \sum_{i=0}^M \alpha_i q_i     
            \suchthat 
    \bm\alpha \in \R^{M+1}
        \, , \, 
     \sum_{i=0}^M \alpha_i Q_i \leq H
    \bigg\}      .
\end{gather}

The result of our asymptotic analysis is the following theorem. For the definition of $\Gamma$-convergence, we refer to Section \ref{sec:limit}. %for the definition).

\begin{theo}[Convergence as $\eps \to 0$ and duality at zero temperature]
   Let $\varphi:[0,+\infty) \to [0,+\infty)$ satisfy \eqref{eq:assumptions_varphi_intro}, and assume that $\psi=\varphi^*$ is strictly convex and $C^1$.  
   \begin{enumerate}
       \item We have that $\primal \xrightarrow{\Gamma} \primalz$ and $-\dualf \xrightarrow{\Gamma} -\dualfz$. 
       \item Whenever $\qAdm \neq \emptyset$, the sequence of minimizers $\{\pi^\eps\}_{\eps>0} \subset \qAdm$ for the primal problem $\primal$, up to a subsequence, converges as $\eps \to 0$ to some $\pi \in \qAdm$ which minimizes $\primalz$. 
       \item Assuming that $\qAdm \cap \tH_>(\cH) \neq \emptyset$, i.e.\ that there exists $\pi_0\in\mathrm{H}(\cH)$  such that 
\begin{align}
   \inf_{\xi\in\cH, ||\xi||=1}{\bra\xi \pi_0\ket\xi}=\omega_0>0\quad\quad \text{and} \quad\quad q_i = \Tr[Q_i \pi_0], \quad 0\leq i\leq M.
\end{align}
    Let $\alpha^\eps$ be the corresponding selected maximizers as given in \eqref{eq:def_argmax}. Then the sequence $\alpha^\eps$ converges, up to subsequence, as $\eps \to 0$ to a maximizer of $\dualfz$. 
    % {\color{orange} [LP]: check if we can say something about any sequence of maximizers, up to reshifiting} 
       \item Duality holds at zero temperature, i.e. $\primz = \dualz$. 
   \end{enumerate}
\end{theo}
The proof of this result is the content of Section~\ref{sec:limit}, where the reader can also find a more detailed analysis of the optimizers and their relations at zero temperature as well.

\section{Preliminaries}\label{sec:preliminaries}
Recall that, given a Hilbert space $\mathcal{V}$ and $f \colon \mathcal{V} \to \R$, we define the Legendre transform of $f$ as the function $f^* \colon \mathcal{V}^* \to \R$ defined as $f^*(v):=\sup \{  \langle v, x \rangle - f(x):\, x \in \mathcal{V} \} $.
Moreover, we define the subdifferential of $f$ at $x$ and we denote it by $\partial f (x)$ as the set
\begin{equation*}
    \partial f(x):= \left\{ v \in \mathcal{V}:\, f(y)\ge f(x)+\langle v,y-x \rangle \right\}.
\end{equation*}
In the case when $f$ is (Fr\'{e}chet) differentiable at $x$, then $\partial f(x)$ is a singleton and $\partial f(x)=\{f'(x)\}$.
One can easily show that $v \in \partial f(x)$ if and only if the following equality holds:
\begin{equation}
    \label{eq:link_subdifferential_legendre}
        \langle v,x \rangle = f(x)+f^*(v).    
\end{equation}

Throughout this section (and in general in this paper), we fix  $\varphi \colon [0,+\infty) \to \R$ a convex, superlinear at infinity, and bounded from below function, namely satisfying \eqref{eq:assumptions_varphi_intro}. When working within the dual framework, to $\varphi$ we associate the corresponding Legendre transform $\psi = \varphi^*$. 
In particular, $\psi \in C(\R)$ is a convex, superlinear at infinity, and bounded from below, i.e. satisfying \eqref{eq:assumptions_psi_intro}.

Albeit some results do require extra differentiability assumptions on $\psi$ and $\varphi$, others hold without such assumptions. As a drawback, when working with nondifferentiable regularization, optimizers in our problem satisfy optimality conditions which cannot be described by the usual spectral calculus, but rather by means of a multi-valued version of it.

We start with the following Proposition, the content of which is essentially as in~\cite[Prop.~A4]{CapGerMonPor24} with a few more details about the optimality conditions.  We set $\Phi(\pi):= \Tr[\varphi(\pi)]$, for $\pi \in  \rmH_\geq(\cH)$, and when necessary extended to $+\infty$ on $\rmH(\cH) \setminus \rmH_\geq(\cH)$.
% The next Proposition is essentially the content of  \cite[Prop.~A4]{CapGerMonPor24} with a few more details about the optimality conditions.

\begin{prop}[Legendre's transform for functional calculus]
\label{prop:computation_legendre_transform}
    Let $\varphi$ satisfy the standing assumption of the section. Let $\cH$ be a finite-dimensional Hilbert space of dimension $d \in \mathbb{N}$. 
    We have that 
    \begin{align}
    \label{eq:Legendre_transform_trance_function}
        \Phi^*(W)
            =
        \sup_{\pi \in \rmH_\geq(\cH)} \big\{ 
            \langle W, \pi \rangle 
                - 
            \Tr
            \left[ 
                \varphi(\pi)
            \right]
        \big\}
            =
        \Tr[\psi(W)]
    \end{align}
    for every $W \in \rmH(\cH)$. Furthermore, for $\pi \in \rmH_{\geq}(\cH)$, $W \in \rmH(\cH)$, we characterize the optimizers via the following equivalent conditions: 
    \begin{enumerate}
        \item[1)]  
    % \begin{align}
    % \label{eq:equality_Fenchel}
   $
        \langle W, \pi \rangle 
            - 
        \Tr
        \left[ 
            \varphi(\pi)
        \right]
            = 
        \Tr[\psi(W)]
    $.
    % \end{align}
        \item[2)] $\pi \in \partial \Phi^*(W)$.
        \item[3)] $W \in \partial \Phi(\pi)$.
        \item[4)] There exists a spectral decomposition for $W$ and $\pi$, 
    \begin{align}
    \label{eq:optimality_subdiff}
        \pi 
            = 
        \sum_{i=1}^d \pi_i \ket{\gamma_i}  \bra{\gamma_i},
            \,
        W 
            = 
        \sum_{j=1}^d W_j \ket{\xi_j}\bra{\xi_j}
            \,\text{ and }\pi_i \in \partial \psi(W_j)
            \, 
       \text{ if } \, 
        \langle \xi_j | \gamma_i \rangle \neq 0 . \qquad 
    \end{align}
    %satisfying 
    %\begin{align}
    %\label{eq:optimality_subdiff}
    %    \pi_i \in \partial \psi(W_j)
    %        \, , \qquad 
    %    \forall (i,j) \suchthat 
    %    \langle \xi_j , \gamma_i %\rangle \neq 0
    %        \, .
    %\end{align}
    \end{enumerate}
  
    In particular, if $\psi$ is strictly convex, then any of the above equivalent conditions yield that $\pi$ and $W$ necessarily commute, and we have that
    \begin{align}
    \label{eq:multivalue_spectralcalc}
        \pi 
            = 
        \sum_{i=1}^d \pi_i \ket{\gamma_i}  \bra{\gamma_i}
            \, , \quad 
        W 
            = 
        \sum_{i=1}^d W_i \ket{\gamma_i}\bra{\gamma_i}
            \, ,    \qquad \text{and} \quad  
        \pi_j \in \partial \psi(W_j)
            \, , \,j=1,\dots, d \,  
    \end{align}
    with some suitable family of eigenvectors $\{ \gamma_i \}_{i=1}^d$. 
\end{prop}
Whenever any of the aforementioned equivalent conditions hold, we use the slight abuse of notation and write 
\begin{align}
    \label{eq:subdifferential_spectral_calculus}
        \pi \in \partial \psi(W)
\end{align}
which is nothing but a multivalued generalization of the usual spectral calculus definition. Indeed, whenever $\partial \psi$ is single-valued (i.e. $\psi$ is differentiable), then
\begin{align}
    \pi \in \partial \psi(W) 
        \quad \Longleftrightarrow \quad 
    \pi = \psi'(W)
        \, ,
\end{align}
in the continuous spectral calculus sense. 
\begin{proof}

The equivalence between (1) -- (3) follows from the observation made in \eqref{eq:link_subdifferential_legendre}.

 To show (1) is equivalent to (4), we use the spectral decomposition, and by Fenchel-Young inequality we see that
    \begin{equation}
    \label{eq:computation_transform}
    \begin{aligned}
        \langle W, \pi \rangle = \Tr\left[ W \pi \right] 
        &= \sum_{i} \pi_i \bra{\gamma_i}W\ket{\gamma_i} = \sum_i\sum_j \pi_i W_j |\langle\xi_j |\gamma_i\rangle|^2
    \\
        &
        \le \sum_i \sum_j (\varphi(\pi_i) + \psi(W_j)) |\langle\xi_j |\gamma_i\rangle|^2
    \\
        &= \sum_i  \varphi(\pi_i) \left(\sum_j|\langle\xi_j |\gamma_i\rangle|^2\right) +  \sum_j \psi(W_j) \left(\sum_i |\langle\xi_j |\gamma_i\rangle|^2\right) 
    \\
        &=  \sum_i  \varphi(\pi_i) +  \sum_j \psi(W_j)  = \Tr\left[\varphi(\pi)\right]+ \Tr\left[ \psi(W) \right]
            \, ,
    \end{aligned}
    \end{equation}
   with equality precisely if and only if \eqref{eq:optimality_subdiff} hold true, hence proving the equivalence between 1) and 4). The latter computation also shows the validity of \eqref{eq:Legendre_transform_trance_function}.

   We are left to prove that only commuting matrices can reach equality when $\psi$ is strictly convex. Indeed, for a strictly convex $\psi$, we have that the the (monotone) multivalued function $\partial \psi$ is injective, in the sense that
   \begin{align}
   \label{eq:injectivity_subdiff}
       \partial \psi(x) \cap \partial \psi(y) \neq \emptyset 
            \quad \Longrightarrow \quad 
        x =y
            \, .
   \end{align}
   Assume now that $\pi$ and $W$ satisfy (1). Then from \eqref{eq:optimality_subdiff} and \eqref{eq:injectivity_subdiff} ensures that, defining
   \begin{align}
       \mathcal{I}_i 
            :=
        \left\{
            1\leq j \leq d
                \suchthat 
            \langle \xi_j | \gamma_i \rangle \neq 0
        \right\}
            \, , \qquad 
        1\leq i \leq d 
            \, , 
   \end{align}
   then we necessarily have that  $\pi_i \in \partial \psi(W_j)$
   \begin{align}
   \label{eq:constancy}
       W_j = W_k =: W^{(i)}
        \, , \qquad 
    \forall j,k \in \mathcal{I}_i
        \, , \quad 
    \forall 1\leq i \leq d 
        \, .
   \end{align}
   We claim that the latter condition implies that $[\pi,W]=0$. To show this, it is enough to show that $\gamma_i$ is an eigenvector of $W$, for every $1\leq i \leq d$. Using \eqref{eq:constancy}, we can write 
   \begin{align}
        W 
            =
        W^{(i)} \sum_{j \in \mathcal{I}_i} \ket{\xi_j}\bra{\xi_j}
            +
        \sum_{j \notin \mathcal{I}_i}
        W_j \ket{\xi_j}\bra{\xi_j}
            \, .
   \end{align}
   As from the very definition of $\mathcal{I}_i$ we have that $\langle \xi_j | \gamma_i \rangle = 0$ for every $j \notin \mathcal{I}_i$, we conclude that 
   \begin{align}
        W \ket{\gamma_i} 
            =
        W^{(i)} \sum_{j \in \mathcal{I}_i} \ket{\xi_j}
        \langle \xi_j | \gamma_i \rangle 
            =
        W^{(i)} \sum_{j =1}^d \ket{\xi_j}
        \langle \xi_j | \gamma_i \rangle 
            =
        W^{(i)} \ket{\gamma_i} 
            \, , 
   \end{align}
   which precisely shows that $\gamma_i$ is an eigenvector of $W$ with eigenvalue $W^{(i)}$.
\end{proof}

\begin{rem}[Uniqueness]
It is interesting to note that, even in the case when $\psi$ is strictly convex, for a given $W \in \rmH(\cH)$, there may exist more than one $\pi \in \rmH(\cH)$ so that \eqref{eq:multivalue_spectralcalc} is satisfied. This clearly relates to the possible lack of differentiability of $\psi$, which in turn is equivalent to its Legendre transform $\varphi$ being strictly convex. This is indeed coherent with the general primal-dual picture: the assumption that $\psi$ is strictly convex ensures uniqueness of maximizers for the dual problem, whereas the strict convexity of $\varphi$ guarantees uniqueness for the primal problem. When both assumptions are satisfied, i.e. $\psi$ is both strictly convex and differentiable, then for every $W \in \rmH(\cH)$, \eqref{eq:multivalue_spectralcalc} has a unique solution $\pi = \psi'(W)$. \fr
\end{rem}

The variational formulation in \eqref{eq:def_primal_problem} provides an interesting mathematical problem, provided that the set of admissible  matrices associated to $\{ Q_j \}_{j=0}^M$ and $\{ q_j\}_{j=0}^M$ is not empty. This unfortunate situation may clearly happen, think for example to the extreme case when e.g. $Q_i=Q_j$ but $q_i \neq q_j$ for some $i,j \in [M]$. It is therefore important to provide conditions to ensure the existence of admissible operators. This is precisely the content of  Proposition \ref{prop:nonmptiness}, which we prove next.

We denote by $\uAdm$ the set of all $\pi \in \rmH(\cH)$ such that $\Tr[Q_i \pi] = q_i$, so that the set of competitors for the primal problem is given by $\text{Adm}(\mathbf{Q}, \mathbf{q}) := \uAdm \cap \rmH_\geq(\cH)$. 
Recall that $\tH(\cH)$ endowed with the usual scalar product $(A,B) \mapsto \Tr[AB]$ is a \textit{real} vector space, and as such when writing \textit{linear independent} below, we mean with respect to linear combination with \textit{real} coefficients.

\begin{proof}[Proof of Prop.~\ref{prop:nonmptiness}]
We prove 1). By assumption, the observables are linearly independent but not necessarily forming a orthonormal set. Let $\{B_i\}_{i=0}^M$ with $B_i \in \rmH(\cH)$ such that $\Tr\left[ B_i B_j \right] = \delta_{i,j}$ and ${\rm Span}\{ Q_i:\,i=0,\dots,M\}={\rm Span}\{ B_i:\,i=0,\dots,M\}$. In particular, there exists $\gamma_i^j\in \R$ for every $i,j \in [M]$ such that $Q_i=\sum_{j=0}^M \gamma_i^j B_j$.

We look for $\pi=\sum_{k=0}^M \alpha_k B_k$ such that $\pi \in \uAdm$, which means
\begin{equation*}
    q_i = \Tr\left[ Q_i \pi \right] = \sum_{j,k} \gamma_i^j \alpha_k \Tr\left[ B_j B_k \right]=\sum_{j} \gamma_i^j \alpha_j.
\end{equation*}
Since the matrices $\{Q_i\}_i$ are linearly independent, the matrix $\Gamma \in \R^{(M+1)\times(M+1)}$ defined as $\Gamma_{i,j}=\gamma_i^j$ is invertible. Thus, there exists $(\alpha_0,\dots,\alpha_M) \in \R^{M+1}$ such that $\Tr\left[ Q_i \pi \right]=q_i$ for all $i \in [M]$.

We prove 2). Given $\pi \in \uAdm$, we have

\begin{equation}
\label{eq:chain_inequalities_linear_dependence}
    \Tr\left[ Q_j \pi \right]\stackrel{\eqref{eq:linear_dependence_of_Qi}}{=} \sum_{i=0}^{j_0} t_j^i \Tr\left[ Q_i \pi \right] = \sum_{i=0}^{j_0} t_j^i q_i\quad\text{for all }j >j_0.
\end{equation}
Thus, since $\Tr\left[ Q_j \pi \right]=q_j$, we conclude one implication.
Conversely, we repeat verbatim the argument in Step 1 to the collection $\{Q_i\}_{i=0}^{j_0}$ and we construct one $\pi \in \rmH(\cH)$ such that $\Tr[Q_j \pi]=q_j$ for all $j \in [j_0]$. For $j > j_0$, the equalities in \eqref{eq:chain_inequalities_linear_dependence} gives that $\Tr\left[ Q_j \pi\right]=q_j$ for all $j > j_0$, thus proving $\pi \in \uAdm$.
\end{proof}

\section{Duality and characterization of the optimizers}\label{sec:main_duality}
In this section, we prove the duality result in Theorem~\ref{thm:main_duality}. 
We let $ M\in\mathbb{N}$, $\varepsilon>0$ be a positive real number, $\cH$ be a finite-dimensional Hilbert space of dimension $d \in \mathbb{N}$. Let $H,Q_0, Q_1,\dots,Q_M \in \rmH(\cH)$ be Hermitian matrices over  $\cH$ and $(q_0,q_1,\dots,q_M)\in\mathbb{R}^{M+1}$. We also pick $Q_0 \in \tH_>(\cH)$ any positive definite Hermitian operator and $q_0\in (0,+\infty)$.
We recall the definition of the set of admissible states associated with this class of observables and measurements, given by
\begin{align}
    \qAdm = 
    \big\{
        \pi \in \tH_\geq(\cH) 
            \suchthat 
        \Tr[Q_i \pi] = q_i, 
            \, 0 \leq i \leq M
    \big\}
        .
\end{align}

We start by showing that the dual functional is a lower bound for the primal, and characterize the equality cases, i.e., the optimizers for the problems. 
\begin{prop}[Weak duality and optimality]
\label{prop:weak-duality}
 Let $\varphi:[0,+\infty)\to\R$ be function satisfying assumption \eqref{eq:assumptions_varphi_intro}, and $\psi= \varphi^*$ be its Legendre transform. Then 

\begin{equation}
\label{eq:weak-duality_inequality}
    \primal(\pi)
        \geq 
    \dualf(\bm\alpha)
        \, , \qquad 
    \forall \pi \in \qAdm
        \, , \, \, \bm\alpha \in \R^{M+1}
            \, .
\end{equation}
Moreover, for some  $\pi \in \qAdm$ and $\bm\alpha \in \R^{M+1}$, the following conditions are equivalent: 
\begin{itemize}
    \item  $\primal(\pi) = \dualf(\bm\alpha)$.
    \item We have the differential inclusion 
\begin{align}
\label{eq:optimality}
    \pi \in\partial\psi\left(\frac{\sum_{i=0}^M\alpha_iQ_i - H}{\varepsilon}\right)
        \, , 
\end{align}
in the sense of \eqref{eq:subdifferential_spectral_calculus}.
\end{itemize}

In particular, whenever $\psi$ is differentiable, this happens if and only if $\pi =\psi'\left(\frac{\sum_{i=0}^M\alpha_iQ_i - H}{\varepsilon}\right)$. Whenever one of the above two equivalent conditions holds, then any such $\pi$ is a minimizer for the primal problem, and $\bm\alpha$ is a maximizer for the dual one. 
\end{prop}

\begin{proof} 
As a consequence of \eqref{eq:Legendre_transform_trance_function} in Proposition~\ref{prop:computation_legendre_transform}, we have that for all $\rmH_{\geq}(\cH)$ and $A$ Hermitian, one has
\begin{equation}\label{eq:Fenchel-Young}
    \Tr[\varphi(\pi)] + \Tr[\psi(A)] \geq \Tr[A\pi],
\end{equation}
and the equality happens if and only if $\pi\in\partial \psi(A)$.

Thus, for our case, choosing $A=\frac{\sum_{i=0}^M\alpha_i Q_i - H}{\eps}$ and multiplying by $\eps$ in equation~\eqref{eq:Fenchel-Young}
\begin{align*}
    \eps \Tr[\varphi(\pi)] +\eps \Tr\left[\psi\left(\frac{\sum_{i=0}^M\alpha_i Q_i - H}{\eps}\right)\right] \geq &
    % \Tr[ (\sum_{i=0}^M\alpha_i Q_i - H)\pi] \\
    % \geq & 
    \sum_{i=0}^M\alpha_i \Tr[Q_i\pi] - \Tr[H\pi]
\end{align*}
Now, restricting $\pi$ to the competitors class $\qAdm$, one obtains
\begin{align*}
    \primal(\pi)=\Tr[H\pi] +\eps \Tr[\varphi(\pi)] \geq \sum_{i=0}^M\alpha_i q_i - \eps \Tr\left[\psi\left(\frac{\sum_{i=0}^M\alpha_i Q_i - H}{\eps}\right)\right] = \dualf(\bm\alpha).
\end{align*}
This proves \eqref{eq:weak-duality_inequality}. As the equality case in the previous inequality is precisely characterized by  \eqref{eq:optimality}, the claimed equivalence readily follows.
\end{proof}

The previous result shows that $\prim \geq \dual$. Next, we shall discuss the existence of maximizers in the dual problem, and as a consequence, the validity of strong duality. Note that, in order to have the existence of a maximizer, additional assumptions are necessary, as suggested by various works \cite{CapGerMonPor24, FelGerPor23} (therein, the assumption was on the kernel of the marginals). This is a typical feature in linear programming: the lack of a duality gap and the existence of maximizers are consequences of the existence of a feasible element which lies in the \textit{interior} of the admissibility set, an assumption which is known as \textit{Slater condition} \cite{slater2013lagrange}. In this framework, the right hypothesis is to assume the existence of an admissible operator $\pi \in \qAdm$ with \textit{no kernel}, as described in the next theorem. We next provide a detailed proof in our setting, as the (quantitative in $\eps>0$) coercivity estimates find application in later convergence results.

\begin{theo}[Existence of a maximizer]
\label{thm:max_exist}
Let $\psi:\R\to\R$ be a superlinear, convex function bounded from below. Assume that 
\begin{align}
\label{eq:condition_for_exist_max}
    \qAdm \cap \tH_>(\cH) \neq \emptyset 
        .
\end{align} 
Then the dual functional $\dualf$ defined in~\eqref{intro:maindual} admits a maximizer $\bm\alpha^* = (\alpha_0^*,...,\alpha_M^*)\in\mathbb{R}^{M+1}$.
\end{theo}

\begin{proof}
By \eqref{eq:condition_for_exist_max}, pick $\pi_0\in\mathrm{H}(\cH)$ such that 
\[
   \inf_{\xi\in\cH, ||\xi||=1}{\bra\xi \pi_0\ket\xi}=\omega_0>0\quad\quad \text{and} \quad\quad q_i = \Tr[Q_i \pi_0], \qquad \forall i \in [M]
    .
\]

\underline{Step 1}. \ 
We prove that, for every $A \in \R$, the set
\begin{equation*}
    \mathcal{S}_A:=
    \left\{  \sum_{i=0}^M \alpha_i Q_i:\, \bm\alpha \in \R^{M+1},\, \dualf(\bm\alpha) \ge A  \right\}
\end{equation*}
is bounded. Let us first assume $\eps \in (0,\frac12]$. By convexity of $\psi$, by \cite[Theorem 2.10]{carlen2010trace} $\Tr[\psi(\cdot)]$ is convex, thus
    \begin{align*}
        A \leq& \sum_{i=0}^M \alpha_i q_i - \eps \Tr\left[\psi\left(\frac{\sum_{i=0}^M \alpha_i Q_i - H}{\eps}\right)\right] \\
        \leq & \sum_{i=0}^M \alpha_i q_i - 
        \Tr\left[\psi\left(\sum_{i=0}^M \alpha_i Q_i\right)\right]  + (1-\eps) \Tr\left[\psi\left(\frac{H}{1-\eps}\right)\right].
    \end{align*}
    For simplicity, denote $h_\eps = (1-\eps)\Tr[\psi(\frac{H}{1-\eps})])$. Now, consider the eigendecomposition $\sum_i \alpha_i Q_i = \sum_k \lambda_k \ket{\xi_k}\bra{\xi_k}$, where  $\{\ket{\xi_k}\}_{k=1}^d$ forms a full basis of $\cH$ and $\lambda_1\leq ...\leq \lambda_d$ are the eigenvalues. By relation $q_i = \Tr[Q_i \pi_0]$ for all $0\leq i\leq M$ we then have
    \begin{align}\label{eq:superlevel_ineq}
        A -h_\eps \leq& \sum_{i=0}^M \alpha_i q_i - 
        \Tr\left[\psi\left(\sum_{i=0}^M \alpha_i Q_i\right)\right] = \sum_{k=1}^d\lambda_k \bra{\xi_k}\pi_0\ket{\xi_k} - \sum_{k=1}^d \psi(\lambda_k).
    \end{align}
    Recall from \eqref{eq:boundedness_qAdm} that every competitor has trace which is uniformly bounded by $t_0:= \frac{q_0}{\lambda_1(Q_0)}$. Using this and that $\psi$ is bounded from below by some $m \in \R$, on one hand \eqref{eq:superlevel_ineq} yields
    \begin{align*}
        A -h_\eps \leq \lambda_d \sum_{k=1}^d\bra{\xi_k}\pi_0\ket{\xi_k} - \sum_{k=1}^d \psi(\lambda_k)\leq  t_0 \lambda_d - \psi(\lambda_d) - m(d-1).
                          \end{align*}
    Consequently, arguing as in~\cite[Remark 3.1]{CapGerMonPor24}, by superlinearity of $\psi$, from 
    \begin{equation}\label{eq:bound_lambda_max}
        t_0\lambda_d - \psi(\lambda_d) \geq A -h_\eps + m(d-1)
    \end{equation}
    one obtains that $\lambda_d\leq R $ for some $R>0$ which depends on $A - h_\eps+ m(d-1)$.

    On the other hand, again using  \eqref{eq:superlevel_ineq} and lower bounds of $\psi$, 
    \begin{align*}
        A -h_\eps 
            &\leq
        \lambda_1 \bra{\xi_1}\pi_0\ket{\xi_1} 
            +
        \lambda_d
            \sum_{k=2}^d \bra{\xi_k}\pi_0\ket{\xi_k} - \sum_{k=1}^d \psi(\lambda_k) 
    \\
            &\leq 
        \lambda_1 \bra{\xi_1}\pi_0\ket{\xi_1} + \lambda_d
        ( 
            \Tr[\pi_0]-\bra{\xi_1}\pi_0\ket{\xi_1}
        ) - md
    \\
            &\leq 
        \lambda_1\omega_0 + \lambda_d(\Tr[\pi_0]-\omega_0) - md
    \\
            &\leq 
        \lambda_1\omega_0 + R(\Tr[\pi_0]-\omega_0) - md
            \leq 
         \lambda_1\omega_0 + R(t_0 -\omega_0) - md
            ,
    \end{align*}
    where in the above we used that $\bra{\xi_1}\pi_0\ket{\xi_1}\geq \omega_0>0$, $\omega_0 < \Tr[\pi_0] \leq t_0$, and that $\lambda_1\leq ...\leq \lambda_d$. Therefore, we conclude the following lower bound on $\lambda_1$
    \begin{equation}\label{eq:bound_lambda_min}
        \lambda_1 \geq \frac{1}{\omega_0}(A - h_\eps + md - R(t_0-\omega_0)).
    \end{equation}
    This, together with $\lambda_d \leq R$, concludes the proof of Step 1
    for $\eps \in (0,\frac12]$. 
    Additionally, we in fact showed that there exists a function $f=f(\eps,A,H)$\footnote{As well as on $\psi, d, t_0$, but we omit this dependence for simplicity} such that 
    \begin{align}
    \label{eq:quantitative_coercivity}
        \max(|\lambda_1|, |\lambda_d| )
            \leq 
        f(\eps, A, H)
            < +\infty
            \, .
    \end{align}
    Moreover, it is clear by construction that for every $\eps \in (0,\frac12]$, the function $f$ is locally bounded in $A$ and $H$. We see in the next corollary that it is possible to prove that it is also uniformly bounded in $\eps \in (0,\frac12]$. 

    Now, let's fix $\eps > \frac12$. 
    Using the bound from below from $\psi$, we see that 
    \begin{align}
        \dualf(\bm \alpha) 
            &= 
        2\eps
        \bigg(
        \sum_{i=0}^M 
            \frac1{2\eps}\alpha_i q_i
                - 
            \frac12
        \Tr\left[\psi\left(\frac{\sum_{i=0}^M \alpha_i Q_i - H}{\eps}\right)\right]
    \bigg)
    \\
            &=
        2\eps {\rm D}_{\frac12}
        \Big( 
            \frac1{2\eps} \bm \alpha 
                ; 
            \frac1{2\eps} H
        \Big)
            \, , 
    \end{align}
    where with $D_\eps(\cdot; K)$ denotes the dual functional with Hamiltonian $K$. Therefore, if $\dualf(\bm \alpha) >A$, then thanks to \eqref{eq:quantitative_coercivity} we deduce that 
    \begin{align}
    \label{eq:quantitative_eps>12}
       \max( |\lambda_1| , |\lambda_d| )
            \leq 
        2\eps f\Big( \frac12, \frac{A}{2\eps} , \frac{H}{2\eps} \Big)
            < +\infty 
                \, ,
    \end{align}
    where as before $\lambda_i$ denotes the $i$-th eigenvalue of $\sum_{i=0}^M \alpha_i Q_i$. Step 1 is now complete. 

\underline{Step 2}. \
We prove that there exists a bounded maximizing sequence. More precisely, we claim that for every $A \in \R$, there exist $\bm\alpha^A \in \R^{M+1}$ so that 
\begin{align}
    \sum_{i=0}^M \alpha_i^A Q_i \in \mathcal{S}_A
        \tand 
    \|\bm\alpha^A \| \leq C_A
\end{align}
for some $C_A\in \R$.
We have 2 cases.

    \begin{enumerate}
        \item[Case 1:] Suppose that $\{Q_0,...,Q_M\}$ are linearly independent. Then the map $(\alpha_0,...,\alpha_M)\mapsto \sum_{i=0}^M \alpha_i Q_i$ is injective, and thus invertible on its image. Since the inverse is a linear map between finite-dimensional Hilbert spaces, it is bounded. Therefore, by Step 1 we conclude.
        \item[Case 2:] Conversely, if there exists some $\beta_0,...,\beta_M$ with $\beta_j\neq 0$ for some $j$ and $\sum_{i=0}^M \beta_i Q_i=0$, define
        \[
        \tilde \alpha_i := \alpha_i - \frac{\alpha_j}{\beta_j}\beta_i.
        \]
        Then $\sum_{i=0}^M \alpha_iQ_i = \sum_{i=0}^M \tilde\alpha_i Q_i$ and $\tilde\alpha_j=0$, thus $\sum_{i=0}^M \alpha_iQ_i = \sum_{i \neq j} \tilde\alpha_i Q_i$. In particular, it can be readily checked that $\dualf(\tilde{\bm\alpha}) \ge A$. Now, if $\{Q_i\}_{i \neq j}$ are linearly independent, one repeats verbatim the proof of Case 1 with the family $\{Q_i\}_{i \neq j}$ and we conclude. Otherwise, repeat the procedure of Case 2 for $\sum_{i\neq j}^N \tilde\alpha_i Q_i$. 
    \end{enumerate}

\underline{Step 3}. 
We conclude by observing that from Steps 1 and 2, we ensure the existence of a converging sequence of maximizers. As the functional $\dualf$ is readily continuous, we conclude that every limit point must necessarily be a maximizer.
\end{proof}

In our work, we are interested in the behavior of the problem when the regularization parameter $\eps$ is close to zero. 
A more careful analysis of the previous proof provides equi-coercivity properties for small values of $\eps$ (and in general for $\eps \in [0,\eps_0]$ for any $\eps_0 >0$), which is crucial for the analysis as $\eps \to 0$.

\begin{prop}[Equi-coercivity in $\eps>0$]\label{prop:alpha_indepentent_eps}
Let $\psi:\R\to\R$ be a superlinear, non-decreasing\footnote{Note that for $\psi = \varphi^*$ and $\varphi$ satisfying \eqref{eq:assumptions_varphi_intro}, then monotonicity of $\psi$ is always true, see e.g. \cite[Remark~A.1]{CapGerMonPor24}.}, convex function bounded from below. Assume that 
% \begin{align}
% \label{eq:condition_for_exist_max}
$
    \qAdm \cap \tH_>(\cH) \neq \emptyset 
        .
$
% \end{align} 
Assume that we have $\eps_0 >0$ and a sequence $\bar{\bm \alpha}^\eps\in\R^{M+1}$ such that
\begin{align}
\label{eq:hp_equicoerc}
    \inf_{0 < \eps < \eps_0} 
        \dualf(\bar{\bm \alpha}^\eps)\geq A
            .
\end{align}

Then the norm of the operator $\sum_{i=0}^M \bar \alpha^\eps_i Q_i$ is uniformly bounded by a constant $L= L(A)< \infty$ independent of $\eps$. Consequently, assume (up to reordering) that all $Q_i$ with $i=0, \dots, j_0$ is a basis for the whole collection $\mathcal{V}_{\mathbf{Q}}$ (cfr. \eqref{eq:def_VQ}), and we denote by $\hat{\bm\alpha}^\eps$ the vector such that 
\begin{align}
    \sum_{i=0}^M 
        \bar \alpha_i^\eps Q_i 
            =
    \sum_{i=0}^{j_0}
        \hat \alpha_i^\eps Q_i
            \, ,
\end{align}
then $\dualf(\bar{\bm\alpha}^\eps) =\dualf(\hat{\bm\alpha}^\eps)$ and  $\hat{\bm\alpha}^\eps$ is uniformly bounded.
\end{prop}
\begin{proof}
    Recall the notion of $f$ introduced in \eqref{eq:quantitative_coercivity}. Firstly, we claim that we can choose $f$ in such a way that 
    \begin{align}
    \label{eq:uniform_bound_<12}
        C_1  := 
        \sup_{\eps \in ( 0, \frac12 ) }
            f(\eps, A, H) 
                < \infty 
                    \, .
    \end{align}
    Indeed, 
  going back to inequality in Eq.\eqref{eq:superlevel_ineq}, we had
    \[
    A -h_\eps \leq \sum_{i=0}^M \bar\alpha_i^\eps q_i - 
        \Tr\left[\psi\left(\sum_{i=0}^M \bar\alpha_i^\eps Q_i \right)\right].
    \]
    The only part that depends on $\eps$ is $h_\eps = (1-\eps)\Tr[\psi(\frac{H}{1-\eps})])$, thus it is enough to bound it from above. Now, as $0 < \eps < \frac12$, one has 
    \[
    \frac{H}{1-\eps} \leq \frac{\lambda_d(H)}{1-\eps}\Id \leq \max\{0, 2 \lambda_d(H)\} \Id
    \]
    and thus, by monotonicity and nonnegativity of $\Tr[\psi(\cdot)]$, 
    \[
        h_\eps 
            = 
        (1-\eps)\Tr\left[\psi\left(\frac{H}{1-\eps}\right)\right] 
            \leq 
        (1- \eps)\Tr\left[\psi\left(\max\{0, 2 \lambda_d(H)\}\Id\right)\right]  \leq d \psi(\max \left\{0,  2 \lambda_d(H)\right\}) 
            \, 
         =:h,
    \]
    thus providing a uniform bound in $\eps < \frac12$. 
    Consequently, all the conclusions about the norm of $\sum_{i=0}^M \bar\alpha^\eps_iQ_i$ (and thus for $\bar{\bm \alpha}^\eps$) can now be directly replicated using $A-h$ instead of $A-h_\eps$.
    Moreover, if $\eps_0 > \frac12$, since $f(\frac12, \cdot,\cdot)$ is locally bounded (as already observed in the proof of Theorem~\ref{thm:max_exist}), we also have that
    \begin{align}
    \label{eq:uniform_bound_>12_locally}
        C_2 
            :=
        \sup_{\eps \in [\frac12, \eps_0]} 
             2\eps f \Big( \frac12, \frac{A}{2\eps}, \frac{H}{2\eps} \Big) 
              < \infty  
                \, .
    \end{align}
    Now, assume without loss of generality that $\eps_0 \geq \frac12$, and let $\bar{\bm\alpha}^\eps$ such that \eqref{eq:hp_equicoerc} is satisfied. From  \eqref{eq:quantitative_coercivity}, we infer that for $\eps \leq \frac12$, 
    \begin{align}
        \Big\|
            \sum_{i=0}^M 
                \bar \alpha_i^\eps Q_i 
        \Big\|_\infty 
            \leq 
        f(\eps,A,H)
            \, ,
    \end{align} 
    while for $\eps > \frac12$, \eqref{eq:quantitative_eps>12} yields
    \begin{align}
        \Big\|
            \sum_{i=0}^M 
                \bar \alpha_i^\eps Q_i 
        \Big\|_\infty 
            \leq 
         2 \eps f  \Big( \frac12, \frac{A}{2\eps}, \frac{H}{2\eps} \Big) 
            \, .
    \end{align}
    All in all, this shows that 
    \begin{align}
        \sup_{\eps \in (0,\eps_0)}
        \Big\|
            \sum_{i=0}^M 
                \bar \alpha_i^\eps Q_i 
        \Big\|_\infty 
            \leq 
        \max (C_1 , C_2 )
            < \infty .
    \end{align}
    This provides the sought uniform bound in $\eps \in (0,\eps_0)$ for $\sum_{i=0}^M \bar\alpha_i^\eps Q_i$. The claimed bounds on $\bar{\bm \alpha}^\eps$ (or on their suitable renormalised version $\hat{\bm \alpha}^\eps$) then follows on the very same lines of the proof of Theorem~\ref{thm:max_exist}. 
\end{proof}

\begin{oss}[Counterexample to existence of maximizers] 
Without the assumption on the existence of an admissible $\pi_0$ with trivial kernel, maximizers may fail to exist. 

Consider the following simple example: let $\cH = \R^2$, and choose
\begin{align}
Q_0 = \Id
\quad\quad
Q_1 = \begin{pmatrix}
    1 & 0 \\
    0 & 0
\end{pmatrix} 
\quad\quad 
q_0=1
\quad\quad 
q_1 = 0
    \, .
\end{align}
Let $\pi \in \rmH(\cH)$ be an admissible operator, i.e.
\begin{align}
\pi = \begin{pmatrix}
    \pi_{11} & \pi_{12} \\
    \pi_{12} & \pi_{22}
\end{pmatrix} 
    \quad \,  \text{s.t.} \quad 
\Tr[\pi] = \pi_{11} + \pi_{22} = 1 
    \, , \quad 
\Tr[\pi Q_1] = \pi_{11} = 0
    \, , \tand 
\pi \geq 0
    \, .
\end{align}
The first two constraints ensure that $\pi_{11} =0$ and $\pi_{22} = 1$, and thus for $\pi$ being nonnegative necessarily needs $\pi_{12}=0$. This computation shows that the unique admissible $\pi \ge 0$ is given by 
\begin{align}
    \pi_0 := 
    \begin{pmatrix}
        0 & 0 \\
        0 & 1
    \end{pmatrix} 
\end{align}
which in particular has kernel. The assumptions of Theorem~\ref{thm:max_exist} are thus not satisfied.

We consider the Hamiltionian $H=\begin{pmatrix}
    h_{0} & 0 \\
    0 & h_{1}
\end{pmatrix}$. For $\alpha = (\alpha_1, \alpha_2) \in \R^2$, we have that 
\begin{align}
\hspace{-2mm}
    \dualf(\alpha)
        =
    \alpha_0- \varepsilon\Tr\left[ \psi \left(  \frac{\alpha_0 \Id +\alpha_1 Q_1-H}{\varepsilon} \right)  \right] 
% \\
        = 
    \alpha_0 - \varepsilon \psi \left( \frac{\alpha_0 + \alpha_1 -h_0 }{\varepsilon} \right) -\varepsilon \psi\left( \frac{\alpha_0-h_1}{\varepsilon} \right)
    .
\end{align}
In particular, the functional $\dualf$ does not admit a maximizer for every choice of $\psi$ which admits no global minimizer (e.g. $\psi=\exp$), as in this case the functional $F \colon \R^2 \to \R$ 
\begin{equation*}
    F(x,y) = x - \varepsilon \psi \left( \frac{x -h_1 }{\varepsilon} \right) -\varepsilon \psi\left( \frac{y-h_0}{\varepsilon} \right)  
        \, , \qquad 
    (x,y) \in \R^2 
        \, , 
\end{equation*}
does not.
\end{oss}

We are almost ready to prove our main duality result. With Proposition~\ref{prop:weak-duality} and Theorem~\ref{thm:max_exist} at our disposal, the missing step is the following primal-dual optimizers characterization, typical of any dual argument. 

\begin{theo}[Equivalent characterizations for maximizers]
\label{thm:characterization_maximizers}
   Let $\psi:\R\to\R$ be $C^1$, convex, superlinear, and bounded from below. Let $\varphi:=\psi^*$.
   % LP: Dont need
   % Assume that $\qAdm \cap \tH_>(\cH) \neq \emptyset$. 
Let $\bm\alpha=(\alpha_0,\dots,\alpha_M) \in \R^{M+1}$. Then the following are equivalent:
    \begin{itemize}
        \item[1)] \emph{(Maximizers)} $\alpha$ maximizes $\dualf $, i.e. $\dualf (\bm\alpha)=\dual$.
        % \item[2)] (Maximality condition) $\alpha_i^* = \mathscr{F}_i^{(H,\psi,\eps)}(\hat{\alpha}^*_i)$, for all $0\leq i\leq M$.
        \item[2)] \emph{(Compl. slackness)} The operator 
        $ \displaystyle 
        \pi:=\psi' \left( \frac{\sum_{i=0}^M\alpha_iQ_i - H}{\varepsilon} \right)$ is so that $\pi \in \qAdm$.
        % \item[4)] (Duality \emph{I}) The operator $\Gamma^*$ defined in 3) is such that  $\primal(\Gamma^*)= \emph{D}^\eps (U^*,V^*)$.
        \item[3)] \emph{(Duality)} There exists $\pi \in \qAdm$ such that $\primal(\pi)= \dualf (\bm\alpha)$.
    \end{itemize}
If one (and thus all) condition holds, $\pi$, as defined in 2), is the unique minimizer to \eqref{intro:main} and $\bm\alpha$ a maximizer of the dual. 
\end{theo}
\begin{proof}
    We proceed by showing 1)$\Rightarrow$ 2) $\Rightarrow$ 3)$\Rightarrow$ 1).

    $1)\Rightarrow 2):$ First, by functional calculus and first-order perturbation theory, we know that 
    \begin{align}
    \frac{d}{dt}\Tr[\psi(A+ tB)]|_{t=0} = \Tr[B\psi'(A)]
    \end{align}
    for every hermitian operators $A,B \in \rmH(\cH)$\footnote{More in general known as Birman-Solomyak formula, see \cite[Section~V3]{Bhatia1997MatrixAnalysis}.}. By applying this formula to $A = \frac{\sum_{i=0}^n \alpha_i Q_i - H}{\eps}$ and $B=\frac{Q_j}{\eps}$, we find \begin{equation}\label{eq:derivative_dual}
        \frac{\partial \dualf}{\partial \alpha_j}(\bm\alpha) = q_i - \Tr\left[Q_j\psi'\left(\frac{\sum_{i=0}^M \alpha_i Q_i - H}{\eps}\right)\right]
            \, .
    \end{equation}

    If $\bm\alpha$ maximizes $\dualf$, then all $\frac{\partial \dualf}{\partial \alpha_j} (\bm\alpha)= 0$, and 2) follows.

    $2)\Rightarrow 3):$ By Proposition~\ref{prop:weak-duality}, we know that the equality $\primal(\pi)=\dualf(\bm\alpha)$ may happen only when $\pi = \psi'(\frac{\sum_{i=0}^N \alpha_i Q_i - H}{\eps})$.

    $3)\Rightarrow 1):$ Let $\pi \in \qAdm$. Then by again Proposition~\ref{prop:weak-duality}, we conclude that 
    \begin{equation*}
    \begin{aligned}
      \primal(\pi) \geq \prim \geq \dual \geq \dualf(\bm\alpha) 
    \end{aligned}
    \end{equation*}
    holds. Therefore, by assumption, all the previous inequalities are equalities, so $\alpha$ is a maximizer of $\dualf$.
\end{proof}

We are finally ready to prove our main duality theorem.

\begin{proof}[Proof of Theorem~\ref{thm:main_duality}]
On one hand, the assumption $\qAdm \cap \tH_>(\cH) \neq \emptyset$ ensures the existence of a maximizer $\bm\alpha$ by Theorem~\ref{thm:max_exist}. Using $1) \Rightarrow 3)$ in Theorem~\ref{thm:characterization_maximizers}, we infer the existence of a $\pi \in \qAdm$ so that $\primal(\pi) = \dualf(\bm\alpha)$. This, in particular, shows that $\prim \leq \dual$. As the other inequality follows from Proposition~\ref{prop:weak-duality}, we conclude that $\dual = \prim$. Moreover, $\pi$ is the unique minimizer for the primal problem thanks to Theorem~\ref{thm:characterization_maximizers}, where the uniqueness part follows from the strict convexity of $A  \mapsto \Tr[\varphi(A)]$ (which follows from the fact that $\psi$ is of class $C^1$). 
\end{proof}

We finally conclude the section by discussing the class of all maximizers for the dual problem and the existence of a continuous selection in $\eps \in (0,+\infty)$.

\begin{prop}[Maximizers and continuity in $\eps$]
\label{prop:uniqueness_and_continuity}
    Let $\psi\colon\R\to\R$ be a superlinear and {strictly} convex function bounded from below. Assume that $\qAdm \cap \tH_>(\cH) \neq \emptyset$.    
Then:
\begin{itemize}
    \item[1)] If $\{Q_i\}$'s are linearly independent, then there exists a unique maximizer of $\dualf$, denoted by $\bm\alpha^\eps=(\alpha^{\varepsilon}_0,\dots,\alpha^{\varepsilon}_M)\in\mathbb{R}^{M+1}$. If additionally $\psi \in C^1$, then the curve
    \begin{equation}
        \label{eq:continuity_as_a_function_of_eps}
        (0,\infty) \ni \eps \mapsto \bm\alpha^\eps \in \mathbb{R}^{M+1} \text{ is continuous;}
    \end{equation}
    \item[2)] If $\{Q_i\}_i$'s are linearly dependent, consider $j_0 \in \mathbb{N}$, $\{ t_j^i \}_{i,j}\subset \R$, and the setting described in 2), Proposition~\ref{prop:nonmptiness}. Then we have
    \begin{align}
        \argmax \dualf
            =
        \left\{
            \bm\alpha \in \R^{M+1}
                \suchthat 
            \alpha_i + \sum_{j=j_0+1}^M \alpha_j t_j^i 
                = 
            \alpha_i^\eps
                \, , \, 
            \text{for all }
                i \in \{ 0, \dots, j_0 \}
        \right\}
            \, ,
    \end{align}
    where $\bm\alpha^\eps \in \R^{M+1}$ is the unique maximizer of $\dualf$ so that $\alpha_i^\eps =0$ for every $i > j_0$. Moreover, if additionally $\psi \in C^1$, then the curve $\eps \mapsto \bm\alpha^\eps$ is also continuous on $(0,+\infty)$.
\end{itemize}
\end{prop}
In other words, whenever the set of $\{Q_j\}_{j=0}^M$ has some linear dependency, we lose the uniqueness of the maximizer, and the set of all optimizers is an affine space of dimension $M - j_0$. Nonetheless, in both the independent and dependent cases, we have a natural, continuous selection $\eps \mapsto \bm\alpha^\eps$ of maximizers.

\begin{proof}
Fix $\varepsilon >0$. Note that we can write the dual functional using the coupling $\pi_0 \in \qAdm \cap \tH_>(\cH)$ as  
\begin{align}
    \dualf(\bm \alpha) 
        =
    \Tr[S \pi_0] 
        -
    \eps 
    \Tr\left[\psi\left(\frac{S - H}{\eps}\right)\right]
        \, , \qquad 
    S := \sum_{i=0}^M \alpha_i Q_i
        \in \tH(\cH) 
            \, .
\end{align}
 Using that $\psi$ is strictly convex, we deduce that the map on the right-hand side is strictly convex in $S$, and therefore
 \begin{align}
\label{eq:sum_equal}
    \sum_{i=0}^M \alpha_i Q_i= \sum_{i=0}^M \beta_i Q_i
        \, , \qquad
    \forall \bm \alpha, \bm \beta \in \argmax \dualf 
        \, .
\end{align} 

We prove 1). Let $\{Q_i\}_i$ be linearly independent.
The invertibility of the map 
\begin{align}
    \R^{M+1} \ni \bm\alpha \mapsto \sum_{i=0}^M \alpha_i Q_i \in \rmH(\cH)
\end{align}
together with \eqref{eq:sum_equal} provides the sought uniqueness.

On the other hand, for every $0<\varepsilon_1<\varepsilon_2<\infty$
\[
    \inf_{\eps_1 \leq \eps \leq \eps_2} 
        \dualf({\bm \alpha}^\eps)\geq \inf_{\eps_1 \leq \eps \leq \eps_2} 
        \dualf({\bm 0}) \in\R.
\] 
and by Prop.~\ref{prop:alpha_indepentent_eps} the family $\left\{ \dualf:\, \varepsilon \in [\varepsilon_1,\varepsilon_2]  \right\}$ is equi-coercive.
Thus, given a sequence $\{ \varepsilon_n \}_n$ with $\varepsilon_n >0$ and $\varepsilon_n \to \varepsilon$, there exists a not relabeled subsequence such that $\bm\alpha^{\varepsilon_n}$ converges to an element $\bm\alpha \in \R^{M+1}$. Since $\bm\alpha^{\varepsilon_n}$ are maximizers, whenever $\psi \in C^1$, by the equivalence 1)--2) in Theorem \ref{intro:main}, they satisfy
\begin{equation}
\label{eq:optimality_condition_for_n}
   \Tr\left[\psi' \left( \frac{\sum_{i=0}^M\alpha^{\varepsilon_n}_iQ_i - H}{\varepsilon_n} \right) Q_j\right] =q_j  \quad\text{for all }j =0,\dots, M.
\end{equation}
Since $\psi \in C^1$, then $A \mapsto \psi'(A)$ is continuous, thus passing to the limit in \eqref{eq:optimality_condition_for_n} we get
\begin{equation}
\label{eq:optimality_condition_for_limit}
    \Tr\left[\psi' \left( \frac{\sum_{i=0}^M\alpha_i Q_i - H}{\varepsilon} \right) Q_j\right] =q_j  \quad\text{for all }j =0,\dots, M  
\end{equation}
and again by the equivalence 1)--2) in Theorem \ref{intro:main}, we conclude that $\bm\alpha=\bm\alpha^\varepsilon$, i.e.\ the unique maximizer of $\dualf$.  This proves that $\{ \bm\alpha^{\eps_n}\}_{n \in \mathbb{N}}$ has a unique accumulation point given by $\bm\alpha^\eps$, hence showing that $\bm\alpha^{\varepsilon_n}$ converges to $\bm\alpha^\varepsilon$.

In order to prove 2). We start by observing that the uniqueness and continuity (for $\psi \in C^1$) in $\eps \in (0,+\infty$) of the maximizer $\bm\alpha^\eps$ so that $\alpha_j^\eps =0$ for every $j>j_0$ readily follows by the arguments used in the proof of 1). Secondly, if $\bm\alpha$ is another maximizer, then  from \eqref{eq:sum_equal} we infer
\begin{align}
\label{eq:some_algebra_on_Qi}
    \sum_{i=0}^{j_0} 
        \alpha_i^\eps Q_i = \sum_{i=0}^M \alpha_i Q_i 
            &=
    \sum_{i=0}^{j_0} 
        \alpha_i Q_i 
            + 
        \sum_{j=j_0+1}^M \alpha_j 
        \Big(
            \sum_{i=0}^{j_0} t_j^i Q_i 
        \Big) 
\\
            &= 
    \sum_{i=0}^{j_0} 
    \left( 
        \alpha_i + \sum_{j=j_0+1}^M t_j^i \alpha_j 
    \right) Q_i
        \, .
\end{align}
By linear independence of $\{Q_i\}_{i=0}^{j_0}$, we have that 
\begin{align}
\label{eq:constraints_ai}
    \alpha_i + \sum_{j=j_0+1}^M \alpha_j t_j^i 
        = 
    \alpha_i^\eps
        \qquad
    \text{for all }i=0,\dots,j_0
        \, .
\end{align}
Conversely, whenever $\alpha$ satisfies \eqref{eq:constraints_ai}, then the very computation performed in \eqref{eq:some_algebra_on_Qi} shows that $\sum_{i=0}^{M} \alpha_i Q_i = \sum_{i=0}^{j_0} \alpha_i^\eps Q_i$, hence $\bm\alpha$ is a maximizer as well. 
\end{proof}

\section{Finite temperature limit ($\varepsilon\to 0^+$)}\label{sec:finitetemplimit}
\label{sec:limit}
In this section, we discuss the behavior of the problems defined in \ref{intro:main} and \ref{intro:maindual} when $\eps\to 0$.
As usual, we fix $M\in\mathbb{N}$, $\cH$ a  finite-dimensional Hilbert space, $H, Q_0,Q_1,\dots,Q_M \in \rmH(\cH)$ to be Hermitian matrices over $\cH$ and $(q_0,q_1,\dots,q_M)\in\mathbb{R}^{M+1}$. 
 We also pick $Q_0 \in \tH_>(\cH)$ a positive definite Hermitian operator and $q_0\in (0,+\infty)$.

We show that as $\varepsilon\to 0$ the regularized functionals for the primal and (minus) dual problem $\Gamma-$converge, respectively, to the following limit problems
\begin{align}
\label{converged:primal}
\primz &:= \inf 
\left\{ 
\primalz(\pi) 
\suchthat 
\pi \in \qAdm
\right\} 
\\
&:=  \inf \left\lbrace \Tr[H\pi]  \suchthat  \pi \in {\rmH}_{\ge}(\cH), \, {\rm Tr}[Q_i\pi] = q_i, \, 0\leq i\leq M \right\rbrace , \quad 
\end{align}
and 
\begin{align}
\label{converged:dual}
\dualz &:= \sup \left\lbrace \dualfz(\bm\alpha) \, : \, \bm\alpha \in \R^{M+1} \right\rbrace \\
&:=\sup \left\lbrace \sum_{i=0}^M\alpha_i q_i - \chi(W) \, : \, W=\sum_{i=0}^M \alpha_iQ_i - H, ~ \bm\alpha \in \R^{M+1}  \right\rbrace,
% \\
% &=\sup \left\lbrace \sum_{i=0}^M\alpha_i q_i \suchthat \sum_{i=0}^M \alpha_i Q_i \leq H,  ~ \alpha_i\in\R, \, 0\leq i\leq M  \right\rbrace,
\end{align}
where $\chi \colon \rmH(\cH) \to \R$ is the indicator defined as
\begin{equation}
    \chi(W) = \begin{cases}
        0, &\text{if }W \leq 0, \\
        +\infty, & \text{otherwise.}
    \end{cases}
\end{equation}
In addition, we also show that respective solutions of regularized problems must converge to the optimizers of the limit problems accordingly. As a final result, we may also obtain strong duality of the limit problems. We want to stress that one could address, as usual, the equality $\primz = \dualz$ without passing a limit through the regularized problems, we just adopt this approach as we have the analysis ready from the previous sections.

\begin{oss}(Weak duality of limit problems)
We consider the case in which $\qAdm \neq \emptyset$, namely there exists some $\pi_0 \in \tH_\geq(\cH)$ such that $\Tr[Q_i\pi_0]=q_i$ for all $0\leq i\leq M$. Recall that we assume $Q_0$ is invertible (which ensures boundedness of any admissible coupling, cfr. Remark \ref{rem:boundedness}).
It is clear that we can restrict the maximization runs of the dual problem and recast it as 
\begin{align}
\label{converged:dual_constraint}
\dualz 
=
\sup \left\{ \sum_{i=0}^M \alpha_i q_i \suchthat \sum_{i=0}^M \alpha_i Q_i \leq H,  ~ \alpha_i\in\R, \, 0\leq i\leq M  \right\}.
\end{align}

Let us show that the admissible set of $\bm\alpha$'s is not empty. Indeed, we define $\alpha_0 = -\delta$, $\alpha_1=\dots=\alpha_M=0$, with $\delta$ to be chosen later. We have that $\sum_{i=0}^M \alpha_i Q_i =-\delta Q_0 \le - \lambda_d(H) {\rm Id} \le -H$, where the second-to-last inequality is satisfied if we choose $\delta \ge \lambda_d(H)\lambda_1(Q_0)^{-1}$, where $\lambda_1(Q_0)$ is the smallest eigenvalue of $Q_0$. 

Then for all admissible $\bm\alpha$ and any $\pi_0\in\qAdm$ we have
\[
    \sum_{i=0}^M \alpha_i q_i = \sum_{i=0}^M \alpha_i \Tr[Q_i \pi_0] = \Tr\left[\left(\sum_{i=0}^M\alpha_i Q_i\right) \pi_0\right] \leq \Tr[H\pi_0]\leq \lambda_d(H) \Tr[\pi_0].
    \] 
where in the first inequality we use that  $\Tr[(H-\sum \alpha_i Q_i) \pi_0]=\Tr[\sqrt{\pi_0}(H-\sum \alpha_i Q_i) \sqrt{\pi_0}]\geq 0$.
This shows $\dualz \leq \primz$ and that the supremum in the dual problem is finite.
\fr
\end{oss}

\subsection{Convergence of the Dual problem}
\label{subsec:convergence_dual_problem}
We start by proving the following general preliminary result.
Recall that a sequence of functions $f_\eps: X \to \R \cup \{ +\infty \}$ defined on a topological space is said to $\Gamma$-converge to a $f:X \to \R \cup \{ +\infty \}$ if 
\begin{enumerate}
    \item For every $x_\eps \to x$ in $X$, we have 
    $   \displaystyle
        \liminf_{\eps \to 0}
            f_\eps(x_\eps) \geq f(x)
    $.
    \item For every $x\in X$, there exists $\bar x_\eps \to x$ such that 
    $    \displaystyle 
        \limsup_{\eps \to 0}
            f_\eps(\bar x_\eps) \leq f(x)
    $.
\end{enumerate}
We call $\bar x_\eps$ a \textit{recovery sequence} for $x \in X$. For more details, we refer to the textbook \cite{DalMaso93}.

\begin{prop}\label{prop:tr_psi_converges}
    Let $\psi:\R\to\R$ be a proper convex superlinear and nondecreasing function bounded from below. Let $\cH$ be a finite-dimensional Hilbert space. Then the map 
    \begin{equation}
    {\rm H}(\cH) \ni W\mapsto \eps \Tr\left[\psi\left(\frac{W}{\eps}\right)\right]
    \end{equation}
    $\Gamma-$converges to the map ${\rm H}(\cH) \ni W\mapsto \chi(W) \in \R \cup \{+\infty\}$ as $\eps\to 0$.
\end{prop}
\begin{proof}
    1. ($\Gamma-\liminf$). \ 
    Let $W^{\ep}\to W$ as $\ep \to 0$. We may consider two cases. First, suppose that $-W\geq 0$, then $\chi(W) = 0$. Also, recall that $\psi$ is bounded from below by some constant $m \in \R$, and thus
\[
    \liminf_{\ep\to 0} \ep \Tr\left[\psi\left(\frac{W^{\ep}}{\ep}\right)\right] \geq \liminf_{\ep\to 0} \ep \Tr[m\Id]=\liminf_{\ep\to 0} \ep m d = 0 = \chi(W),
\]
where $d=\dim(\cH)$.
Now, suppose that $-W$ is not $\geq 0$, then there must exist a unit vector $\ket{\xi}$ in $\cH$ such that $\bra{\xi}W\ket{\xi} > 2\delta >0$, but also it is known that $\bra{\xi}W\ket{\xi} \leq \lambda_d(W)$,  where $\lambda_d(W)$ is the largest eigenvalue. At the same time, since $W^{\ep}\to W$, we have $\lambda_d(W^{\ep})\to \lambda_d(W)$. Thus there exists $\ep_0>0$ such that $\lambda_d(W^{\ep})> \delta>0$ for all $0< \ep<\ep_0$.
Now, see that 
\[
\ep \Tr\left[\psi\left(\frac{W^{\ep}}{\ep}\right)\right] = \ep \sum_{i=1}^d \psi\left(\frac{\lambda_i(W^{\ep})}{\ep}\right)  \geq \ep \psi\left(\frac{\lambda_d(W^{\ep})}{\ep}\right) +  \ep (d-1) m,
\]
so it remains to conclude that $\ep \psi(\frac{\lambda_d(W^{\ep})}{\ep})\to +\infty$. Indeed, for $0<\ep<\ep_0$ we have may now use the superlinearity  and monotonicity of $\psi$
\[
    \ep \psi\left(\frac{\lambda_d(W^{\ep})}{\ep}\right) \geq \ep \psi\left(\frac{\delta}{\ep}\right) 
    % = \delta \frac{\psi(\frac{\hat{\delta}}{\ep}) }{\
    % \frac{\hat{\delta}}{\ep}}
    \xrightarrow[\ep\to 0]{} +\infty,
\]
and thus we conclude that $\liminf\limits_{\ep\to 0}\ep \Tr[\psi(\frac{W^{\ep}}{\ep})] =  +\infty = \chi(W)$.

2. ($\Gamma-\limsup$). \  We claim that the constant sequence $W^{\ep}:= W$ for all $\ep>0$ is a recovery sequence. Indeed, if $-W$ is not $\geq 0$, then $\chi(W)=+\infty$, and we don't need to show anything. Conversely, if $-W\geq 0$, then $\lambda_i(W)\leq 0$ for all $1\leq i \leq d$, and using monotonicity of $\psi$
\begin{equation*}
\begin{aligned}
\limsup_{\ep\to 0} \ep \Tr\left[\psi\left(\frac{W}{\ep}\right)\right] &= \limsup_{\ep\to 0} \ep \sum_{i=1}^d \psi\left(\frac{\lambda_i(W)}{\ep}\right) \leq \lim_{\ep \to 0} \ep \sum_{i=1}^d \psi(0)\\
&= \lim_{\ep\to 0}\ep \psi(0) d = 0 = \chi(W).
\end{aligned}
\end{equation*}
This concludes the proof.
\end{proof}

\begin{theo}[Convergence of the dual problems and maximizers]\label{thm:max_converged}
Let $\psi\colon\R\to\R$ be a superlinear, convex and nondecreasing function bounded from below. Assume
\[\qAdm \cap \tH_>(\cH) \neq \emptyset.\] 
%Assume that there exists a density matrix $\pi_0$ on $\cH$ such that 
%\[
%   \inf_{||\xi||=1}{\bra\xi \pi_0\ket\xi}=\omega_0>0\quad\quad \text{and} \quad\quad q_i = \Tr[Q_i \pi_0], ~i=0,...,N.
%\] 
Then $-\dualf$ defined in \eqref{intro:maindual} $\Gamma-$converges to $-\dualfz$ defined in \eqref{converged:dual} as $\eps\to 0$. Additionally, the sequence of maximizers $\bm\alpha^\eps\in\R^{N+1}$ of $\dualf$ defined in 2) of Proposition~\ref{prop:uniqueness_and_continuity}  converges (up to subsequence) to a maximizer   $\bm\alpha\in\R^{N+1}$ of $\dualfz$.
\end{theo}

\begin{proof}
    The $\Gamma-$convergence argument follows immediately from Proposition \ref{prop:tr_psi_converges} with the fact that if $\bm\alpha^\eps\to\bm\alpha$ in $\R^{M+1}$, then $W^\eps = \sum_{i=0}^M\alpha^\eps_i Q_i$ converges to $W = \sum_{i=0}^M\alpha_i Q_i$ in $\rmH(\cH)$ as $\ep \to 0$, as well as $\sum_{i=0}^M\alpha^\eps_i q_i$ to $\sum_{i=0}^M\alpha_i q_i$.

    Now, due to Theorem \ref{thm:max_exist}, for every $\eps>0$ there exists a maximizer $\bm\alpha^\eps$ of $\dualf$. Moreover, we can consider only $\eps\in(0,\delta]$ for some $0<\delta<1$. In particular, by choosing $\hat{\bm\alpha}=(-||H||_\infty, 0,...,0) \in\R^{M+1}$, we find that $\sum_{i=0}^M\hat\alpha_i Q_i - H \leq 0$, which together with the fact that $\Tr[\psi(\cdot)]$ is nondecreasing it yields
    \begin{align*}
        \dualf(\alpha^\eps) 
        \geq & \sum_{i=0}^M\hat\alpha_iq_i - \eps\Tr\left[\psi\left(\frac{\sum\hat\alpha_i Q_i - H}{\eps}\right)\right]\\
        \geq& -||H||_\infty - \eps\Tr[\psi(0)\Id] \geq -||H||_\infty - \delta d \psi(0) =: A.
    \end{align*}
    We can now apply Corollary \ref{prop:alpha_indepentent_eps} to conclude that $\{\bm\alpha^\eps \}_\eps$ is uniformly bounded in $\R^{M+1}$, and thus admits an accumulation point $\bm\alpha$. By the fundamental theorem of $\Gamma$-convergence, any limit point of a sequence of maximizers must necessarily be a maximizer of the limiting problem, and the values of the maxima converge.
\end{proof}

\subsection{Convergence of the primal problem and duality}
In this subsection, we show the convergence of the minimizers of the regularized problem \eqref{intro:main} as $\eps\to 0$. Consequently, by duality and also convergence of the dual problem, we conclude the statement of duality between \eqref{converged:primal} and \eqref{converged:dual}.

\begin{theo}[Convergence of the minimizers of $\primal$]\label{thm:min_congerged}
Let $\varphi:[0,+\infty) \to \R$ be a convex function satisfying \eqref{eq:assumptions_varphi_intro}.  
Suppose that the set
    \[ \qAdm=
    \{\pi\in\rmH(\cH) \suchthat \pi\geq 0 \text{ and } \Tr[Q_i \pi]=q_i,   0\leq i\leq M\} \neq \emptyset.
    \]
Then for any collection $\{\pi^\eps\}_{\eps>0} \subset {\rm Adm}(\mathbf{Q}, \mathbf{q})$ converging to some $\pi \in {\rm Adm}(\mathbf{Q}, \mathbf{q})$  holds 
\[
    \lim_{\eps\to 0^+}\primal(\pi^\eps) = \primalz(\pi),
\]
where $\primalz(\pi)$ is defined in \eqref{converged:primal}.

Moreover, if a collection $\{\pi^\eps\}_{\eps>0} \subset {\rm Adm}(\mathbf{Q}, \mathbf{q})$ consists of minimizers of $\primal$ for every $\eps>0$ respectively, then it admits an accumulation point $\pi^* \in \qAdm$ which minimizes \eqref{converged:primal}.
\end{theo}
\begin{proof}
    For the first part of the statement, since $\pi^\eps\to \pi$ we directly conclude that $\Tr[H\pi^\eps]\to \Tr[H\pi]=\rm \primalz(\pi)$. 
    On the other hand, by Remark~\ref{rem:boundedness} we know that $\lambda_d(\pi) \leq t_0:=\frac{q_0}{\lambda_1(Q_0)} < \infty$ for every $\pi \in \qAdm$. As $\varphi$ is convex on $[0,+\infty)$, it is in particular continuous on every compact set, such as $[0, t_0]$. Therefore $\sup_\eps |\Tr[\varphi(\pi^\eps)]| <\infty$, and thus $\eps\Tr[\varphi(\pi^\eps)]\to 0$.

    Now, suppose that $\{\pi^\eps\}_{\eps>0}$ is a collection of minimizers for $\primal$. Then it is bounded as its elements satisfy the constraints of the primal problem (cf.\ Remark \ref{rem:boundedness}). Let $\pi^*$ be its accumulation point up to a taking a subsequence $\{\eps_k\}_{k\geq 1}$.
    For every $\pi \in \qAdm$, we have 
    % \begin{align*}  
    $
        \primalz(\pi)
            = 
        \lim_{k \to \infty} 
            \primalz_{\eps_k}(\pi) 
            \geq 
        \lim_{k \to \infty} 
            \primalz_{\eps_k}(\pi^{\eps_k})
            =
        \primalz(\pi^*)
            \, .
    % \end{align*}
$
    Therefore, $\pi^*$ must be a minimizer.
\end{proof}

As a corollary, we obtain a possible proof of the strong duality for \eqref{converged:primal} and \eqref{converged:dual}.

\begin{theo}[Duality]
\label{thm:duality_zero}
    We assume that $\qAdm \cap \rmH_{>}(\cH)\neq \emptyset$. Then there exist both a maximizer to \eqref{converged:dual} and a minimizer to \eqref{converged:primal}, and  $\primz=\dualz$.
\end{theo}
\begin{proof}
    We show the result by using the duality and the convergence of regularized problems.
    Let $\varphi:[0, +\infty)\to \R\cup \{+\infty\}$ be a function satisfying \eqref{eq:assumptions_varphi_intro}  such that its Legendre transform $\psi$ is $C^1$, e.g. choose $\varphi(x) = x (\log x -1)$, thus $\psi(y) = \exp(y)$. Then, for every $\eps>0$ define the primal and dual problems as in \eqref{intro:main} and \eqref{intro:maindual}. Thus, by Theorem \ref{thm:characterization_maximizers} there exist optimal solutions $\pi^\eps \in \tH(\cH)$ and $\bm\alpha^\eps \in \R^{M+1}$ to these problems respectively and duality holds for every $\eps>0$. Moreover, each $\pi^\eps$ must be of form $\pi^\eps = \psi'(\frac{\sum_{i=0}^M \alpha^\eps_i Q_i - H}{\eps})$.
    
    Now, by Theorem \ref{thm:max_converged} and  Theorem \ref{thm:min_congerged}, we know there exist subsequences $\{\bm\alpha^{\eps_k}, \pi^{\eps_k}\}_{k\geq 1}$ so that 
   $\bm\alpha^{\eps_k} \to \bm\alpha^*$,  $\pi^{\eps_k} \to \pi^*$, where $\bm\alpha^*\in\R^{M+1}$  maximizes the limit dual problem \eqref{converged:dual}, whereas $\pi^*$ is a minimizer to \eqref{converged:primal}.
    Consequently, since each ${\rm F}_{\eps_{k_j}}(\pi^{\eps_{k_j}}) = {\rm D}_{\eps_{k_j}}(\bm\alpha^{\eps_{k_j}})$ by duality, and respectively converge to $\primz$ and $\dualz$, we conclude the sought equality.
\end{proof}

As in the regularized problem, we can describe the optimality conditions and relations between maximizers and minimizers as described in the following theorem. 

\begin{theo}[Complementary slackness]
\label{th:sdp_comp_slackness}
    Let $\bm\alpha$ satisfy $\sum_{i=0}^M \alpha_i Q_i \le H$ and let $\pi \in \qAdm$. Then the following conditions are equivalent:
    \begin{itemize}
        \item[(i)] The complementary slackness condition holds, namely
        \begin{equation}
    \label{eq:sdp_comp_slackness}
        \sqrt \pi \left( H - \sum_{i = 0}^M \alpha_i Q_i \right) \sqrt \pi= 0
            \,;
    \end{equation}
        \item[(ii)] We have $\primalz(\pi) = \dualfz(\bm\alpha)$.
    \end{itemize}
    If any of the above conditions hold true, then $\bm\alpha$ is a maximizer of $\dualfz$, $\pi$ is a minimizer of $\primalz$, and strong duality $\dualz = \primz$ holds.
\end{theo}

\begin{proof}
    (i) $\Rightarrow$ (ii). \ 
    Note that the complementary slackness condition implies
    \begin{equation*}
        \operatorname{Tr}\left[\left( H - \sum_{i = 0}^M \alpha_i Q_i \right)\pi\right] = \operatorname{Tr}\left[\sqrt\pi \left( H - \sum_{i = 0}^M \alpha_i Q_i \right)\sqrt \pi\right] = 0,
    \end{equation*}
    which implies, using that $\pi \in \qAdm$
    \begin{equation}
    \label{eq:sdp_comp_slack_proof_main_eq_new}
        \primalz(\pi) = \operatorname{Tr}[H \pi]
        = \sum_{i = 0}^M \alpha_i \operatorname{Tr}[Q_i \pi]
        = \sum_{i = 0}^M \alpha_i  q_i=\dualfz(\bm\alpha).
    \end{equation}

    \smallskip
    \noindent 
    (ii) $\Rightarrow$ (i). \ 
    Assume that $\primalz(\pi) = \dualfz(\bm\alpha)$ for admissible $\pi$ and $\bm\alpha$. 
    Then it similarly follows that
    \begin{equation}
    \label{eq:zero_trace_condition}
        \operatorname{Tr}\left[\left( H - \sum_{i = 0}^M \alpha_i Q_i \right)\pi\right] = \operatorname{Tr}\left[\sqrt\pi \left( H - \sum_{i = 0}^M \alpha_i Q_i \right)\sqrt \pi\right] = 0.
    \end{equation}
    Since $\alpha$ satisfies $H-\sum_{i=0}^M \alpha_i Q_i \ge 0$, we have that $\sqrt \pi \left( H - \sum_{i = 0}^M \alpha_i Q_i \right)\sqrt\pi \ge 0$, which, together with Eq.\eqref{eq:zero_trace_condition} implies Eq.\eqref{eq:sdp_comp_slackness}.
    
    Recall that for any $\pi' \in \qAdm$ and $\bm\alpha'$ such that $\sum_{i=0}^M\alpha_i' Q_i\le H$ we have $\operatorname{Tr}[H \pi']
        \geq \sum_{i=0}^M \alpha'_i q_i $. Therefore, Eq.~\eqref{eq:sdp_comp_slack_proof_main_eq_new} implies that $\pi$ is a minimizer of~\eqref{converged:primal}, because
    \begin{equation*}
        \min \left\{  \operatorname{Tr}[H \pi']:\, \pi' \in \qAdm \right\}
        \geq \sum_{i = 0}^M \alpha_i  q_i
        = \operatorname{Tr}[H \pi],
    \end{equation*}
    and by a similar argument we conclude that $\alpha$ is a maximizer of~\eqref{converged:dual}.
\end{proof}

\section{Computational algorithms}
\label{sec:numerics}
In this section, we introduce a novel computational algorithm to solve the dual problem~\eqref{intro:maindual}. We demonstrate its performance on two quantum information tasks that can be formulated as semidefinite programs:
(i) Quantum State Tomography, and (ii) Quantum Optimal Transport. For each task, we examine two choices of convex regularization: the entropy term $\varphi(z) = z \log z$ (corresponding to the von-Neumann entropy), and the quadratic penalty $\varphi(z) = \tfrac{1}{2} z^{2}$ (corresponding to the quadratic or Quantum $\chi^2$-divergence).

To compute the optimal dual variables, we employed the L-BFGS solver~\cite{LiuNoc-MathPr-1989}, using the implementation provided in the Optax library~\cite{optax} with default optimizer parameters. In our experiments, the optimization procedure terminated if either of the following held:
\begin{itemize}
    \item when the accuracy criterion is met, namely  the norm of the gradient of the regularized dual functional ~\eqref{intro:maindual}
    \begin{equation}
    \label{eq:dual_func_grad}
    \left\| \nabla \dualf (\bm \alpha) \right\|^2_2
    = \sum_{i = 0}^M \left( \alpha_i q_i - \operatorname{Tr}\left[ Q_i \psi'\left( \frac{1}{\varepsilon} \left(  \sum_{i = 0}^M \alpha_i Q_i - H\right) \right) \right] \right)^2.    
\end{equation}
    
    is smaller than a tolerance $\tau\in\lbrace 10^{-3}, 10^{-6}\rbrace$;
    \item maximum number of iterations 10000 is reached.\vspace{2mm}
\end{itemize}

All experiments described in this section are performed on a server equipped with an NVIDIA H100 NVL Tensor Core GPU.

\subsection{Quantum State Tomography}

In quantum state tomography, the goal is to reconstruct an unknown quantum state $\rho \in \rmH_{\geq}(\cH)$, i.e. positive semidefinite matrix with unit trace, from measurement outcomes. In variational terms, this problem can be posed as finding a density matrix $\rho\in\rmH_\geq(\cH)$ that best fits the data in a least-squares sense
\begin{equation}
\label{eq:q_tom_orig_problem}
\inf \left\lbrace \, \sum_{i=1}^M \left( \operatorname{Tr}[Q_i \pi] - q_i \right)^2 \, : \, \operatorname{Tr}[\pi] = 1,\  \pi \in \rmH_\geq(\cH) \, \right\rbrace,
\end{equation}
where, for all $i\in[M]$, $Q_i \in \rmH(\cH)$ are measurements, i.e. a Hermitian operators corresponding to an observable, where $Q_0=\Id$, $q_0=1$.

%The expected value of the $i$-th measurement is $q_i = \langle Q_i \rangle_{\rho} = \operatorname{Tr}[Q_i \rho], \, \forall \, i \in[M]$, where we consider .

%n practice, for all $i\in[M]$, we obtain empirical measurement outcomes $q_i \in \R$, which approximate these expectations. 

%More precisely, suppose we perform a set of measurements $\lbrace Q_i\rbrace_{i=1}^M$, where each $Q_i \in \rmH(\cH)$ is a Hermitian operator corresponding to an observable. The expected value of the $i$-th measurement is $q_i = \langle Q_i \rangle_{\rho} = \operatorname{Tr}[Q_i \rho], \, \forall \, i \in[M]$, where we consider $Q_0=\Id$, $q_0=1$.

%In practice, for all $i\in[M]$, we obtain empirical measurement outcomes $q_i \in \R$, which approximate these expectations. 

% \begin{equation}
% \inf \left\lbrace \Tr[\pi] + \varepsilon S_{\varphi}(\pi) \, : \, \pi\geq 0, \, \Tr[M_0\pi] = 1 \text{ and } \Tr[M_i\pi] = q_i, \, \forall \, i \in[N] \right\rbrace,
% \end{equation}
% where $M_0=\mathbb{I}$ is the identity. 
% 
% {\color{blue} [LP]: I really don't like this sentence below: the problem \eqref{eq:q_tom_orig_problem} and \eqref{eq:q_tom_aux_primal} are not equivalent. The second one is solvable iff $\qAdm \neq \emptyset$, and in this case I agree they coincide. But the first one is a relaxation of the second which makes sense also in the case when $\qAdm \neq \emptyset$. 
% } {\color{red}EC: You mean $=0$?}
Whenever $\qAdm = \left\{ \pi \in \rmH_{\geq}(\cH) \ : \ \operatorname{Tr}[Q_i \pi] = q_i, 0 \leq i \leq M \right\} \neq \emptyset$, the infimum is equal to zero, and the variational problem~\eqref{eq:q_tom_orig_problem} can be interpreted as a particular case~\eqref{intro:main} with $\varepsilon=0$ and $H = 0$, i.e.
% \begin{equation}
% \label{eq:q_tom_aux_primal}
%     \begin{aligned}
%         \inf_\pi\ & 0 \\
%         \operatorname{s. t.\ }& \operatorname{Tr}[Q_i \pi] = q_i,\quad i \in [M], \\
%                              & \operatorname{Tr}[\pi] = 1, \\
%                              & \pi \geq 0.
%     \end{aligned}
% \end{equation}
\begin{equation}
\label{eq:q_tom_aux_primal}
 \min \bigg\lbrace \, 
    0 \, : \pi \in \qAdm \bigg\rbrace  =  \min \bigg\lbrace \, 
    0 
    \, : \, 
    \operatorname{Tr}[Q_i \pi] = q_i,\ 0 \leq i \leq M,\ 
    \pi \in \rmH_\geq(\cH) \, 
    \bigg\rbrace.
\end{equation}

In general, the problem \eqref{eq:q_tom_aux_primal} may be ill-defined, as noisy observations $q_0, \dots, q_M$ could create an empty feasible set, as illustrates the following example. 

\begin{exam}
    Consider $\cH = \mathbb{C}^2$, and let $\{X, Y, Z, \Id\}$ be the Pauli matrices, which are
    \begin{align*}
        X = \begin{pmatrix}
            0 & 1 \\
            1 & 0
        \end{pmatrix},\ 
        Y = \begin{pmatrix}
            0 & -\imag \\
            \imag & 0
        \end{pmatrix},\ 
        Z = \begin{pmatrix}
            1 & 0 \\
            0 & -1
        \end{pmatrix},\ 
    \end{align*}
    and $I$ is the identity and $\imag = \sqrt{-1}$.
    Consider the case of the Quantum Tomography problem~\eqref{eq:q_tom_orig_problem} for
    $Q_1 = X,\ Q_2 = Z$.
    Recall that any two-dimensional density matrix $\pi \in \rmH(\cH)$ can be represented as
    \begin{equation*}
        \pi = \frac{1}{2} (\Id + a_x X + a_y Y + a_z Z)
    \end{equation*}
    for certain $a_x, a_y, a_z \in \mathbb{R}$.
    One has to have
    \begin{equation*}
        \operatorname{Tr}[\pi^2] = \frac{1 + a_x^2 + a_y^2 + a_z^2}{2} \leq 1,
    \end{equation*}
    i.e. $a_x^2 + a_y^2 + a_z^2 \leq 1$.
    Noting that $\operatorname{Tr}[Q_1 \pi] = a_x$ and $\operatorname{Tr}[Q_2 \pi] = a_z$, we have to demand $q_1^2 + q_2^2 \leq 1$.
    Let $q_1 = 0.8, q_2 = 0.6$, so we have the equality. However, the slightest increase of $q_1$ or $q_2$ leads to $q_1^2 + q_2^2 > 1$ which makes the feasible set for~\eqref{eq:q_tom_aux_primal} empty.
\end{exam}

\noindent
\textbf{Dual approach and regularization.}
The dual formulation of~\eqref{eq:q_tom_aux_primal} is given by
\begin{equation}
\label{eq:q_tom_dual_reg}
    \sup_{\bm\alpha \in \R^{M + 1}} \dualf (\bm \alpha)
    =
    \sup_{\bm\alpha \in \R^{M + 1}} \sum_{i = 0}^M \alpha_i q_i - \varepsilon \operatorname{Tr}\left[ \psi\left(\frac{1}{\varepsilon} \sum_{i = 0}^M \alpha_i Q_i \right) \right].
\end{equation}

Its gradient (squared) norm
\begin{equation}
\label{eq:q_tom_dual_reg_grad}
    \left\| \dualf (\bm \alpha) \right\|_2^2
    = \sum_{i = 0}^M \left( q_i - \operatorname{Tr}\left[ Q_i \psi'\left( \frac{1}{\varepsilon} \sum_{i = 0}^M \alpha_i Q_i \right) \right] \right)^2
\end{equation}
gives us an important notion of quantum tomography. Indeed,
by solving~\eqref{eq:q_tom_dual_reg}, i.e. finding vector $\bm \alpha^* \in \R^{M + 1}$, we actually solve the initial quantum tomography problem by finding the state that depends on the choice of the regularization.
By the Theorem~\ref{thm:characterization_maximizers}, we know the optimal primal variable $\pi^* \in \rmH_\geq(\cH)$ is
\begin{equation}
\label{eq:q_tom_opt_primal}
    \pi^*
    = \psi'\left( \sum_{i = 0}^M \alpha_i Q_i \right).
\end{equation}
Thus, we are solving the quantum tomography problem by searching the unknown state in the set
\begin{align*}
    S^\psi = 
    \left\{
    \pi \in \rmH(\cH)
    \ : \
    \exists \, \bm\alpha \in \R^{M + 1} \text{ such that }
    \pi = \psi' \left( \sum_{i = 0}^M \alpha_i Q_i \right)
    \right\}.
\end{align*}

\begin{oss}[Independence of the temperature paremeter $\varepsilon$]
It is important to notice that in this particular case the value of regularization parameter $\varepsilon$ does not change the optimal primal point, but rather changes the scale of optimal dual variables $\alpha_1, \dots, \alpha_M$.
Indeed, the regularized dual problem~\eqref{eq:q_tom_dual_reg} reads
\begin{align*}
    % \sup_{\bm \alpha \in \R^{M + 1}} \bm \alpha \cdot \bm q - \varepsilon \operatorname{Tr}\left[ \psi\left( \frac{\bm \alpha \cdot \bm Q}{\varepsilon} \right) \right]
    % &= \varepsilon \sup_{\bm \alpha \in \R^{M + 1}} \frac{\bm \alpha}{\varepsilon} \cdot \bm q - \operatorname{Tr}\left[ \psi\left( \frac{\bm \alpha}{\varepsilon} \cdot \bm Q \right) \right] \\
    % &= \varepsilon \sup_{\bm \alpha \in \R^{M + 1}} \bm \alpha \cdot \bm q - \operatorname{Tr}\left[ \psi\left( \bm \alpha \cdot \bm Q \right) \right],
    \sup_{\bm\alpha \in \R^{M + 1}} \sum_{i = 0}^M \alpha_i q_i - \varepsilon \operatorname{Tr}\left[ \psi\left(\frac{1}{\varepsilon} \sum_{i = 0}^M \alpha_i Q_i \right) \right]
    &= \varepsilon \sup_{\bm\alpha \in \R^{M + 1}} \sum_{i = 0}^M \frac{\alpha_i}{\varepsilon} q_i - \operatorname{Tr}\left[ \psi\left( \sum_{i = 0}^M \frac{\alpha_i}{\varepsilon} Q_i \right) \right] \\ 
    &= \varepsilon \sup_{\bm\alpha \in \R^{M + 1}} \sum_{i = 0}^M \alpha_i q_i - \operatorname{Tr}\left[ \psi\left( \sum_{i = 0}^M \alpha_i Q_i \right) \right]
    , 
\end{align*}
where the last equality comes from the fact that the maximization is unconstrained.
However, from practical considerations, it is worth keeping this parameter since it improves stability of the numerical methods (see Fig.~\ref{fig:q_tom_0}, Fig.~\ref{fig:q_tom_1}, and Fig.~\ref{fig:q_tom_2} for the impact of the $\varepsilon$ value).

% {\color{blue} [LP: was this observation/remark present in Mark's paper? Just curious]} 
\end{oss}

\noindent
\textbf{Numerical results.} Let $n \in \mathbb{N}$ denotes the number of qubits, $D = 2^n$ is dimension of the Hilbert space $\cH$.
By $\mathcal{Q}_n$ we denote the set of $n$-length Pauli strings
\begin{align}
\label{eq:pauli_string_set}
    \mathcal{Q}_n 
    = \left\{
    S \in \rmH(\cH)
    \ : \
    S = \bigotimes_{l = 1}^n \sigma_l, \ \sigma_l \in \{X, Y, Z, \Id\}
    \right\},
\end{align}
where $\{X, Y, Z, \Id\}$ are the Pauli matrices.

Let $\ket{1}, \dots, \ket{D} \in \R^D$ be the canonical basis of $\cH$.
Define the vector state as
\begin{align}\label{eq:cat-state}
    \ket{\zeta(\theta, \omega)}
    = \cos(\theta) \ket{1} + \sin(\theta) e^{\imag \omega} \ket{D}.
\end{align}
For $\beta \in \mathbb{C}$, define the \textit{cat state}, normalized to 1, as
\begin{align*}
    \ket{\text{cat}_\beta} = \ket{\beta} + \ket{-\beta},
\end{align*}
where $\ket{\beta} = e^{-|\beta|^2 / 2} \sum_{l = 0}^D \frac{\beta^l}{\sqrt{l!}}\ket{l}$.

In order to construct quantum tomography test problem, we have to pick observables $Q_i$ and the respective $q_i \in \R, i \in [M]$.
As for observables, we choose $M=2D$ and randomly pick $Q_i \in \mathcal{Q}_n,\ 1\leq i\leq 2 D$, 
% where $2 D$ is chosen arbitrarily \nat{was a repetetion?}
fixing $Q_0=\Id$.
Next, in order to guarantee $\qAdm \neq \emptyset$, for $1 \leq  j \leq 3$ we choose a $\rho_j$ state to generate $q_i = \operatorname{Tr}[\rho_j Q_i], i \in [2 D]$ thus ensuring existence of at least one admissible point.

Three numerical tests use the following states to generate 
% {\color{blue}[LP: write down what does it practically mean, i.e. that we define the $q_i$ as $q_i := \Tr[\rho_j Q_i]$, correct? Question: the state $\rho_i$ may not be optimal tho, right? For those values of $q_i$, even tho it is clearly admissible]}
the observation values $q_0, \dots, q_{2D}$:
% \begin{enumerate}
%     \item $\rho_1(p, \theta, \omega) = p \ket{\zeta(\theta, \omega)} \bra{\zeta(\theta, \omega)} + \frac{1 - p}{D} \Id$ for $p \in [0, 1]$ (\textbf{QT1} instance);
%     \item $\rho_2(\beta) = \ket{\text{cat}_\beta}\bra{\text{cat}_\beta}$ (\textbf{QT2} instance);
%     \item $\rho_3(t, \beta, p, \theta, \omega) = t \rho_1(p, \theta, \omega) + (1 - t) \rho_2(\beta)$ for $t \in [0, 1]$ (\textbf{QT3} instance).
% \end{enumerate}
\begin{align}
\rho_1(p,\theta,\omega)
&= p\,\ket{\zeta(\theta,\omega)}\bra{\zeta(\theta,\omega)}
   + \frac{1-p}{D}\,\Id,
\quad p \in [0,1],
\tag{QT1}\label{eq:QT1_inst}
\\[0.5em]
\rho_2(\beta)
&= \ket{\text{cat}_\beta}\bra{\text{cat}_\beta},
\tag{QT2}\label{eq:QT2_inst}
\\[0.5em]
\rho_3(t,\beta,p,\theta,\omega)
&= t\,\rho_1(p,\theta,\omega)
   + (1-t)\,\rho_2(\beta),
\quad t \in [0,1].
\tag{QT3}\label{eq:QT3_inst}
\end{align}

% {\color{blue} [LP: can we give a little bit of context here? Why are we using these observables? Where are they coming from?}
% As for observables, we randomly pick $2 D$ observables {\color{blue}[why 2D? important >D? are they differnet/lin indep?]} $Q_i \in \mathcal{Q}_n,\ i \in [2 D]$.
% That is, for the $j$-th test the the observation vector $\bm q \in \mathbb{R}^{2D}$ is $q_k = \operatorname{Tr}[\rho_j Q_k],\ k \in [2D]$.

% {\color{blue} [LP: add explanation of the shaded area -- you did it before, maybe it is useful to remind it here]}.
For all the tests, we fix $n = 9$ qubits and $p = 0.7,\ \theta = \pi/6,\ \omega = \pi/4,\ \beta = 2,\ t = 0.5$.

% Plots illustrate the convergence behavior of the value of interest towards 0. 
% For quantum tomography, the value of interest is norm of the gradient (the choice is explained below), and for the others is the absolute difference between the dual objective and the known ground truth (which is precomputed using MOSEK~\cite{mosek} solver). 
% % {\color{blue} Why?}

% Each subplot corresponds to a distinct regularization value, showing the convergence over iterations for the chosen regularizers. 
% The solid lines represent the mean values computed across repeated runs, while the shaded regions indicate the range between the minimum and maximum observed values at each iteration.

% \renewcommand{\arraystretch}{0.85}
% \setlength{\aboverulesep}{0.3ex}
% \setlength{\belowrulesep}{0.3ex}
% \setlength{\extrarowheight}{0pt}

On the Figure~\ref{fig:q_tom_0}-\ref{fig:q_tom_2} we plot the norm of the gradient given by the eq.~\eqref{eq:q_tom_dual_reg_grad} for the test 1-3, respectively.
The data shown in the Tables~\ref{tab:q_tom_0}, \ref{tab:q_tom_1}, and \ref{tab:q_tom_2} describes the time and number of iterations needed for L-BFGS to reach fixed tolerance $\{10^{-3}, 10^{-6}\}$.

\begin{figure}[!htb]
  \centering

  % --- Row 1 ---
  \begin{subfigure}{0.27\textwidth}
    \centering
    \includegraphics[width=\linewidth,height=0.26\textheight,keepaspectratio]{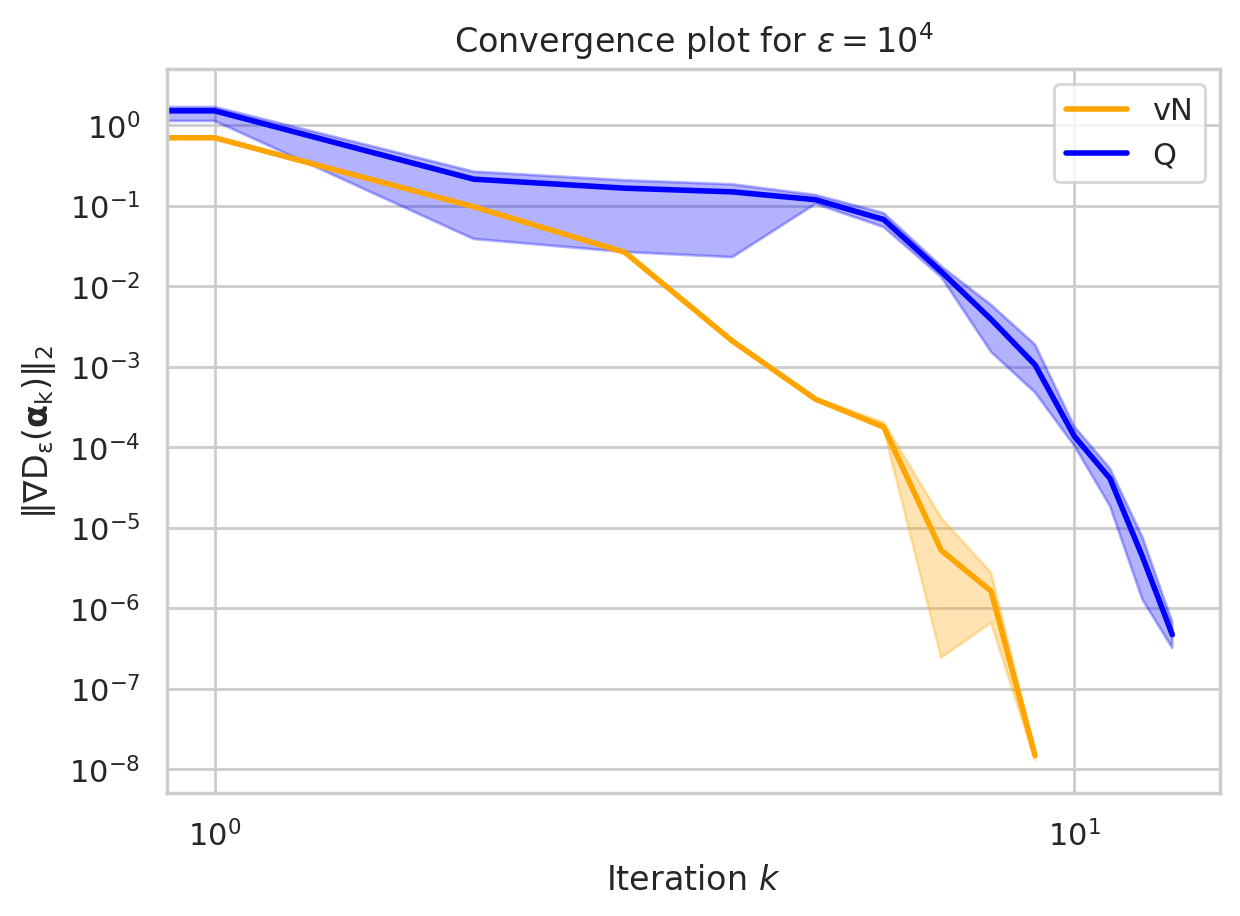}
    \caption{\textbf{QT1} for $\varepsilon = 10^{4}$}
  \end{subfigure}\hfill
  \begin{subfigure}{0.27\textwidth}
    \centering
    \includegraphics[width=\linewidth,height=0.26\textheight,keepaspectratio]{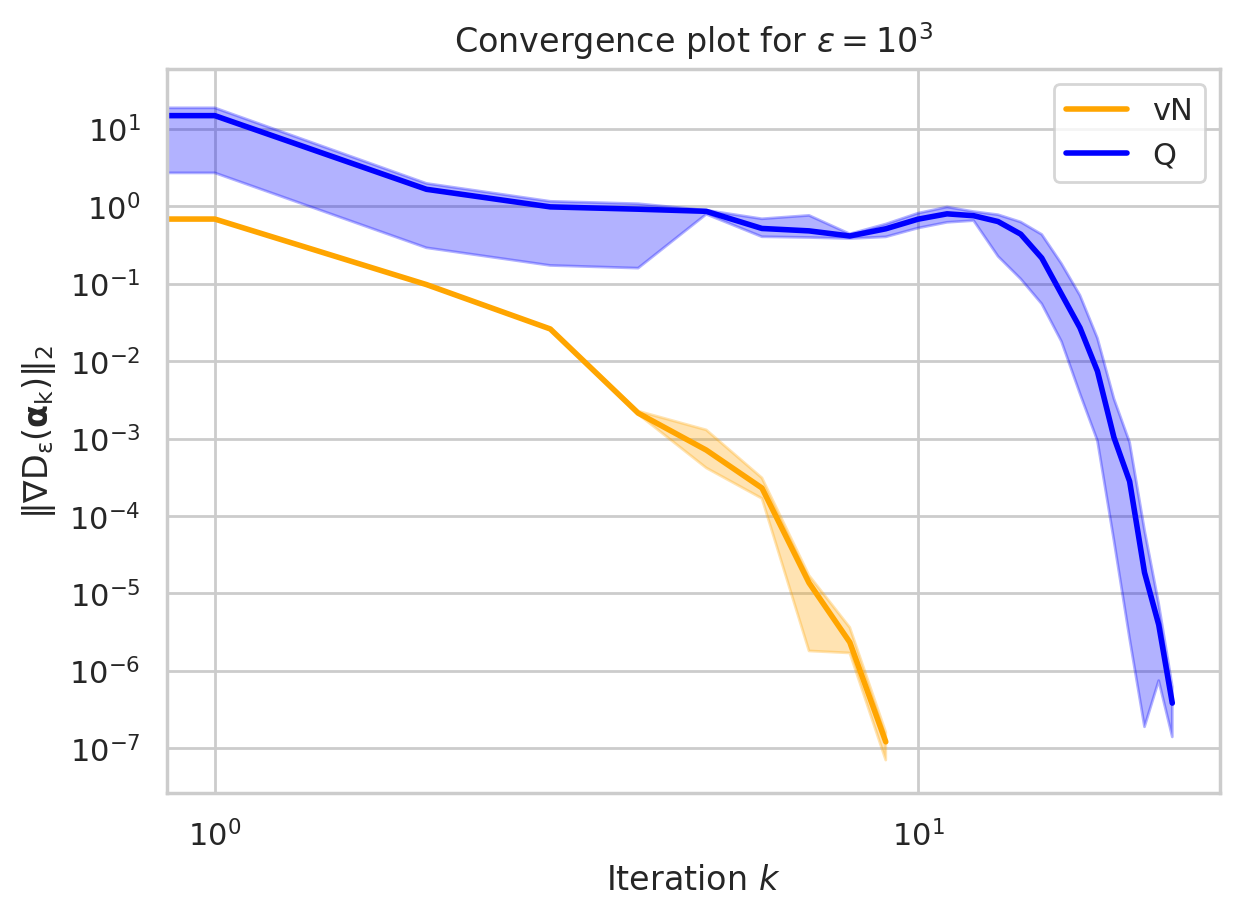}
    \caption{\textbf{QT1} for $\varepsilon = 10^{3}$}
  \end{subfigure}\hfill
  \begin{subfigure}{0.27\textwidth}
    \centering
    \includegraphics[width=\linewidth,height=0.26\textheight,keepaspectratio]{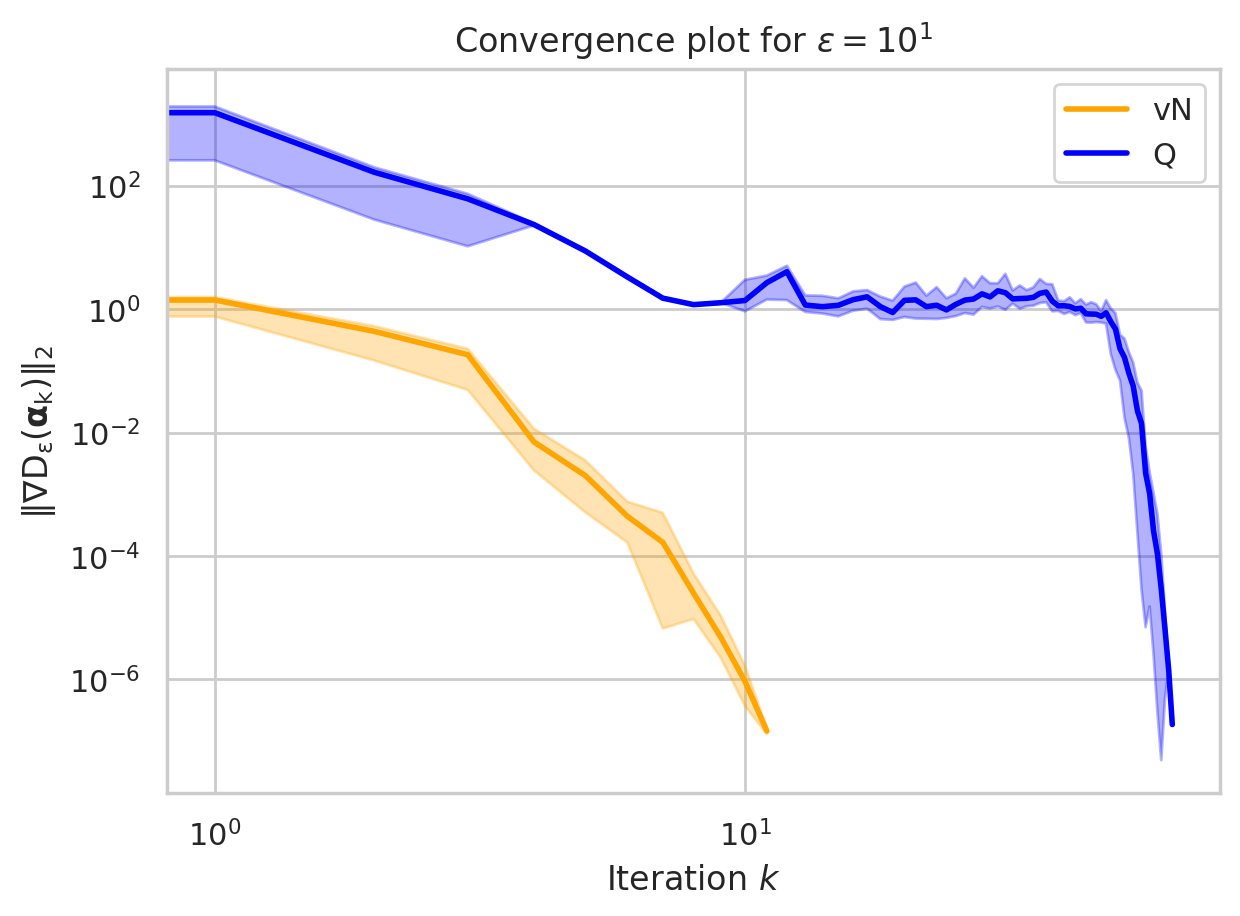}
    \caption{\textbf{QT1} for $\varepsilon = 10^{1}$}
  \end{subfigure}

  % \vspace{0.3em}

  % --- Row 2 ---
  \begin{subfigure}{0.27\textwidth}
    \centering
    \includegraphics[width=\linewidth,height=0.26\textheight,keepaspectratio]{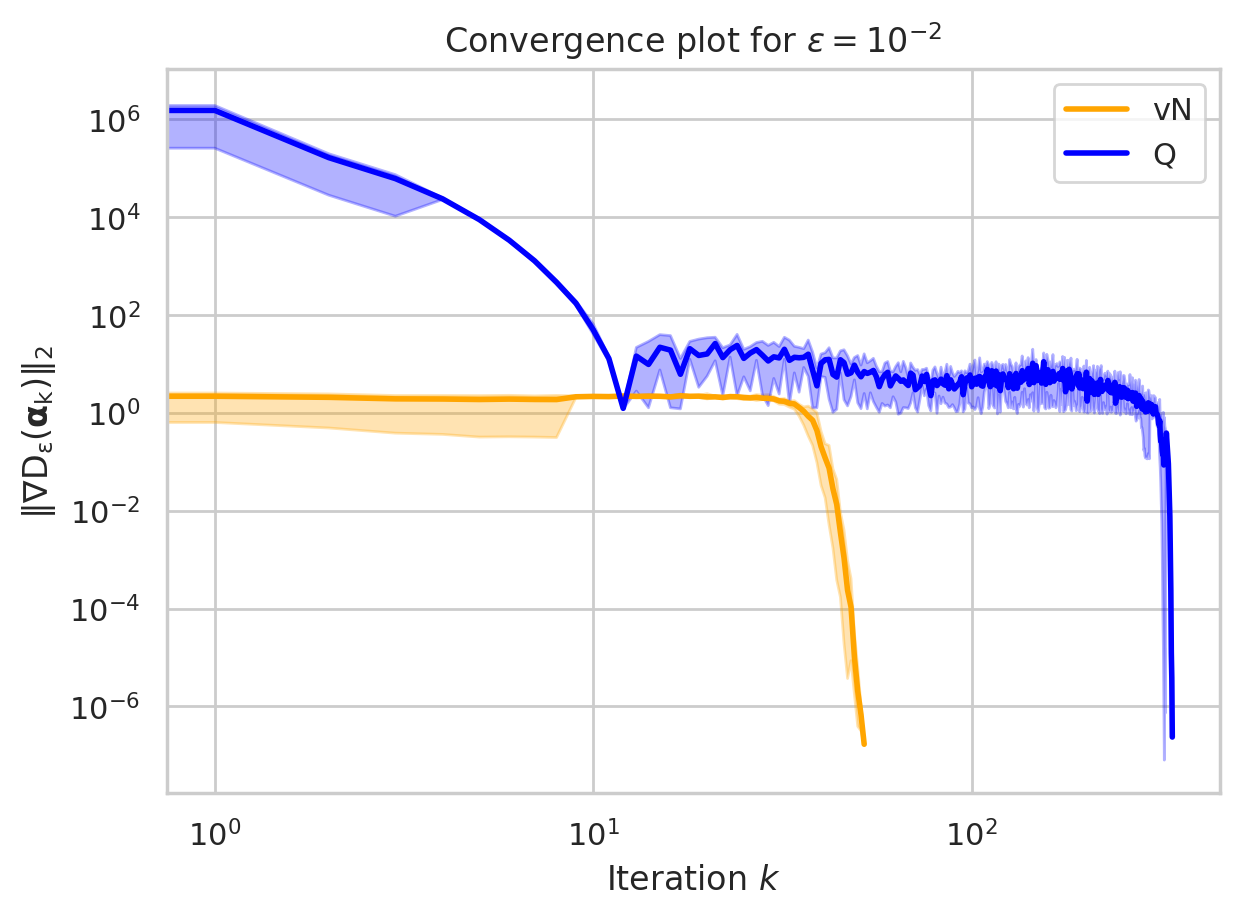}
    \caption{\textbf{QT1} for $\varepsilon = 10^{-2}$}
  \end{subfigure}\hfill
  \begin{subfigure}{0.27\textwidth}
    \centering
    \includegraphics[width=\linewidth,height=0.26\textheight,keepaspectratio]{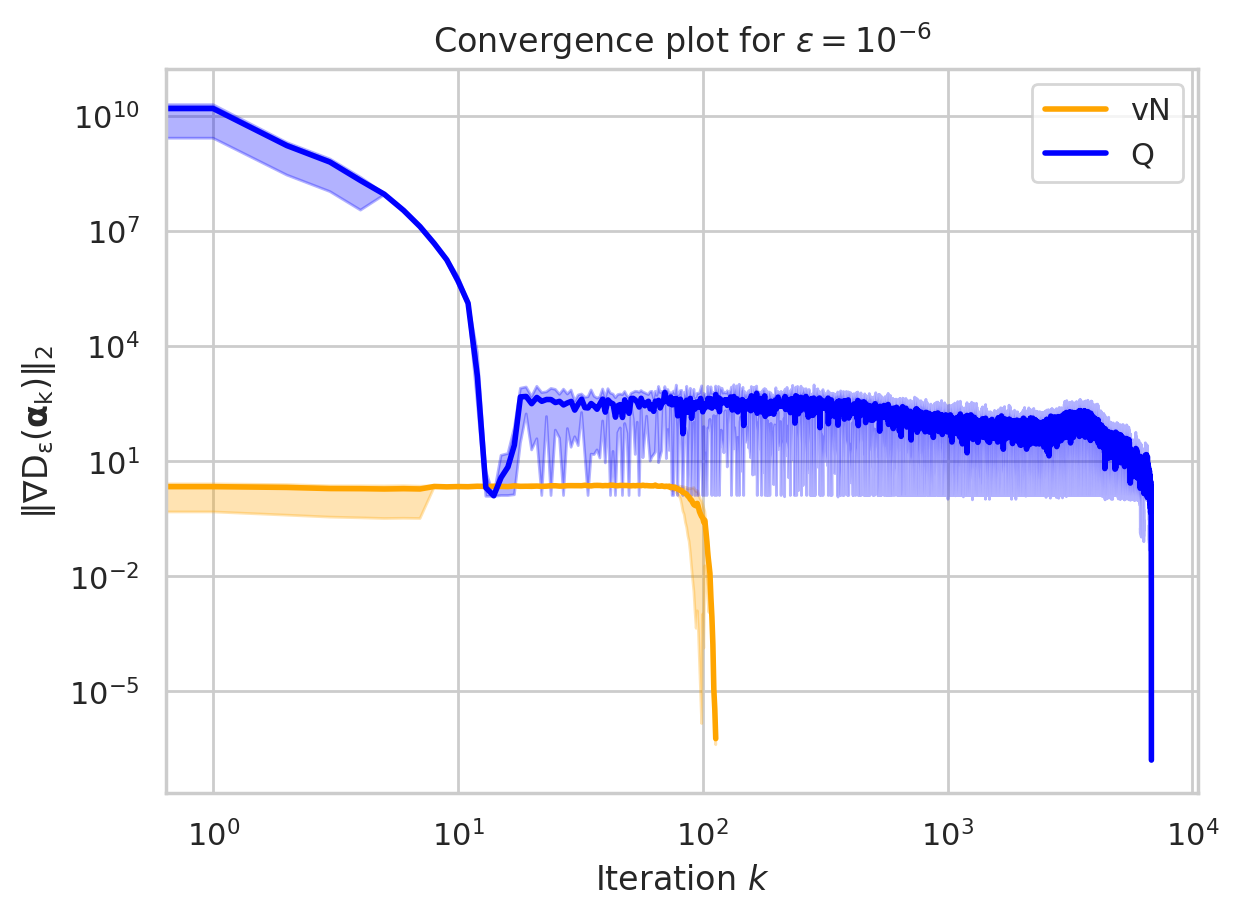}
    \caption{\textbf{QT1} for $\varepsilon = 10^{-6}$}
  \end{subfigure}\hfill
  \begin{subfigure}{0.27\textwidth}
    \centering
    \includegraphics[width=\linewidth,height=0.26\textheight,keepaspectratio]{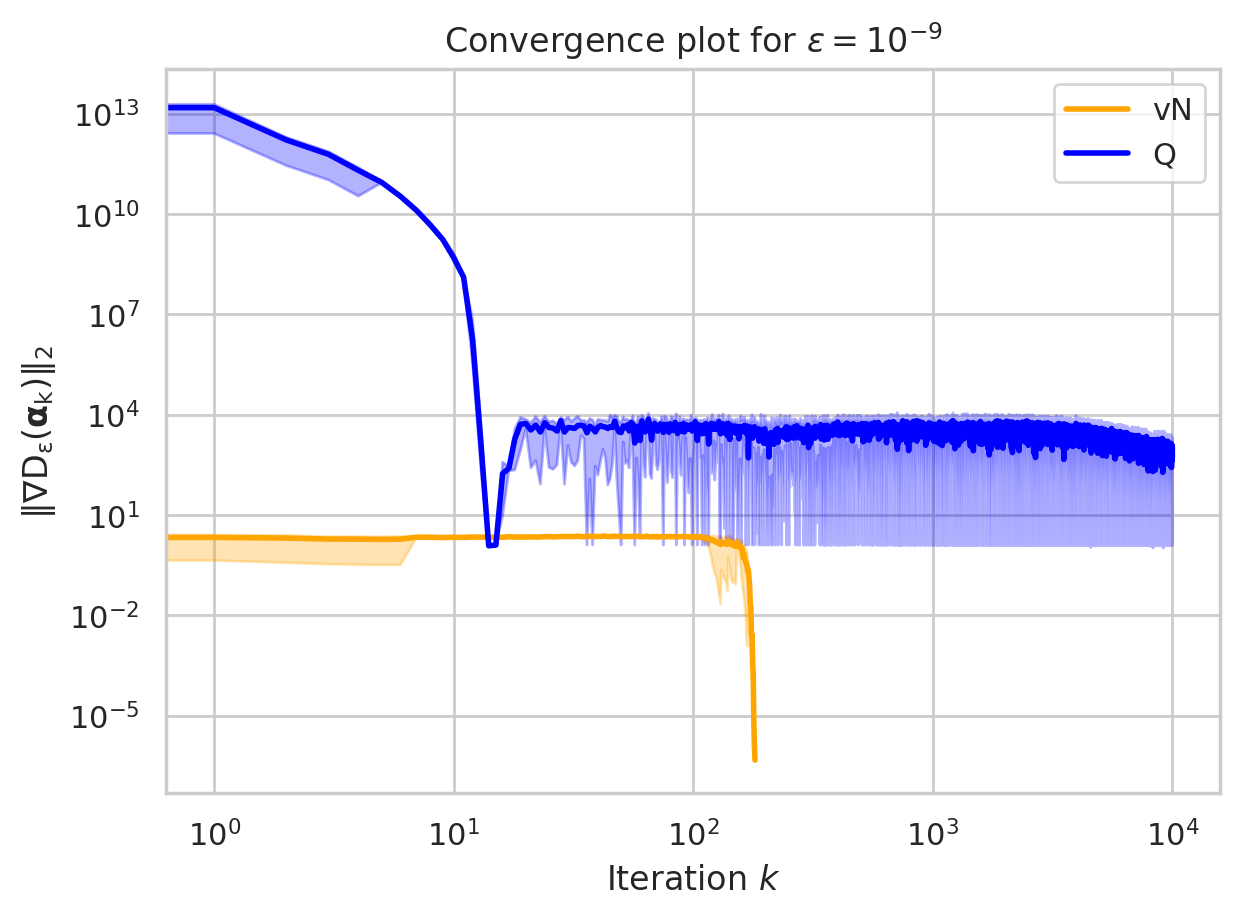}
    \caption{\textbf{QT1} for $\varepsilon = 10^{-9}$}
  \end{subfigure}

  \caption{\textbf{Quantum Tomography problem instance \textbf{QT1}.} Graph of the $L_2$-norm of the gradient of the dual functional $\dualf$ in Eq.~\eqref{intro:maindual} for the $\rho_1$ state; see~\eqref{eq:QT1_inst} in section~\ref{sec:numerics}. Panels~(a)--(f) display the iteration trajectories for different values of the regularization parameter $\varepsilon$
  % \in \{10^{4}, 10^{3}, 10^{1}, 10^{-2}, 10^{-6}, 10^{-9}\}$ 
  for both von Neumann (orange) and quadratic regularization (blue).}
  \label{fig:q_tom_0}
\end{figure}

\begin{table}[ht]
\centering
\footnotesize
\setlength{\tabcolsep}{2pt}

% ---------- Shared header row (only once) ----------
\begin{minipage}[t]{0.31\textwidth}\centering
\begin{tabular}{@{} L T Y N I A @{}}
\toprule Reg. & Tol. & Type & Time & Iters & Ach. \\ \midrule
\end{tabular}
\end{minipage}\hspace{0.02\textwidth}
\begin{minipage}[t]{0.31\textwidth}\centering
\begin{tabular}{@{} L T Y N I A @{}}
\toprule Reg. & Tol. & Type & Time & Iters & Ach. \\ \midrule
\end{tabular}
\end{minipage}\hspace{0.02\textwidth}
\begin{minipage}[t]{0.31\textwidth}\centering
\begin{tabular}{@{} L T Y N I A @{}}
\toprule Reg. & Tol. & Type & Time & Iters & Ach. \\ \midrule
\end{tabular}
\end{minipage}

% \vspace{0.1em}

% ---------- Row 1 ----------
\begin{minipage}[t]{0.31\textwidth}\centering
\begin{tabular}{@{} L T Y N I A @{}}
\multirow{4}{*}{\textbf{$10^{4}$}} & $10^{-3}$ & \texttt{vN} & 0.3 & 6 & Yes \\
 % & $10^{-4}$ & \texttt{vN} & 0.4 & 8 & Yes \\
 % & $10^{-5}$ & \texttt{vN} & 0.4 & 8 & Yes \\
 & $10^{-6}$ & \texttt{vN} & 0.4 & 10 & Yes \\
 & $10^{-3}$ & \texttt{Q} & 0.3 & 10 & Yes \\
 % & $10^{-4}$ & \texttt{Q} & 0.4 & 12 & Yes \\
 % & $10^{-5}$ & \texttt{Q} & 0.4 & 13 & Yes \\
 & $10^{-6}$ & \texttt{Q} & 0.5 & 14 & Yes \\
\bottomrule
\end{tabular}
\end{minipage}\hspace{0.02\textwidth}
\begin{minipage}[t]{0.31\textwidth}\centering
\begin{tabular}{@{} L T Y N I A @{}}
\multirow{4}{*}{\textbf{$10^{3}$}} & $10^{-3}$ & \texttt{vN} & 0.3 & 6 & Yes \\
 % & $10^{-4}$ & \texttt{vN} & 0.4 & 8 & Yes \\
 % & $10^{-5}$ & \texttt{vN} & 0.4 & 9 & Yes \\
 & $10^{-6}$ & \texttt{vN} & 0.5 & 10 & Yes \\
 & $10^{-3}$ & \texttt{Q} & 0.7 & 20 & Yes \\
 % & $10^{-4}$ & \texttt{Q} & 0.7 & 21 & Yes \\
 % & $10^{-5}$ & \texttt{Q} & 0.7 & 22 & Yes \\
 & $10^{-6}$ & \texttt{Q} & 0.8 & 23 & Yes \\
\bottomrule
\end{tabular}
\end{minipage}\hspace{0.02\textwidth}
\begin{minipage}[t]{0.31\textwidth}\centering
\begin{tabular}{@{} L T Y N I A @{}}
\multirow{4}{*}{\textbf{$10^{1}$}} & $10^{-3}$ & \texttt{vN} & 0.3 & 7 & Yes \\
 % & $10^{-4}$ & \texttt{vN} & 0.3 & 9 & Yes \\
 % & $10^{-5}$ & \texttt{vN} & 0.4 & 10 & Yes \\
 & $10^{-6}$ & \texttt{vN} & 0.4 & 11 & Yes \\
 & $10^{-3}$ & \texttt{Q} & 2.1 & 61 & Yes \\
 % & $10^{-4}$ & \texttt{Q} & 2.2 & 63 & Yes \\
 % & $10^{-5}$ & \texttt{Q} & 2.2 & 64 & Yes \\
 & $10^{-6}$ & \texttt{Q} & 2.2 & 65 & Yes \\
\bottomrule
\end{tabular}
\end{minipage}

\vspace{0.3em}

% ---------- Row 2 ----------
\begin{minipage}[t]{0.31\textwidth}\centering
\begin{tabular}{@{} L T Y N I A @{}}
\multirow{4}{*}{\textbf{$10^{-2}$}} & $10^{-3}$ & \texttt{vN} & 1.4 & 47 & Yes \\
 % & $10^{-4}$ & \texttt{vN} & 1.4 & 48 & Yes \\
 % & $10^{-5}$ & \texttt{vN} & 1.5 & 50 & Yes \\
 & $10^{-6}$ & \texttt{vN} & 1.5 & 51 & Yes \\
 & $10^{-3}$ & \texttt{Q} & 11.5 & 337 & Yes \\
 % & $10^{-4}$ & \texttt{Q} & 11.5 & 338 & Yes \\
 % & $10^{-5}$ & \texttt{Q} & 11.5 & 339 & Yes \\
 & $10^{-6}$ & \texttt{Q} & 11.6 & 340 & Yes \\
\bottomrule
\end{tabular}
\end{minipage}\hspace{0.02\textwidth}
\begin{minipage}[t]{0.31\textwidth}\centering
\begin{tabular}{@{} L T Y N I A @{}}
\multirow{4}{*}{\textbf{$10^{-6}$}} & $10^{-3}$ & \texttt{vN} & 3.0 & 110 & Yes \\
 % & $10^{-4}$ & \texttt{vN} & 3.0 & 112 & Yes \\
 % & $10^{-5}$ & \texttt{vN} & 3.0 & 113 & Yes \\
 & $10^{-6}$ & \texttt{vN} & 3.1 & 114 & Yes \\
 & $10^{-3}$ & \texttt{Q} & 221.2 & 6754 & Yes \\
 % & $10^{-4}$ & \texttt{Q} & 221.3 & 6755 & Yes \\
 % & $10^{-5}$ & \texttt{Q} & 221.3 & 6756 & Yes \\
 & $10^{-6}$ & \texttt{Q} & 221.3 & 6757 & Yes \\
\bottomrule
\end{tabular}
\end{minipage}\hspace{0.02\textwidth}
\begin{minipage}[t]{0.31\textwidth}\centering
\begin{tabular}{@{} L T Y N I A @{}}
\multirow{4}{*}{\textbf{$10^{-9}$}} & $10^{-3}$ & \texttt{vN} & 4.5 & 172 & Yes \\
 % & $10^{-4}$ & \texttt{vN} & 4.6 & 174 & Yes \\
 % & $10^{-5}$ & \texttt{vN} & 4.6 & 175 & Yes \\
 & $10^{-6}$ & \texttt{vN} & 4.6 & 176 & Yes \\
 & $10^{-3}$ & \texttt{Q} & 318.8 & 10000 & No \\
 % & $10^{-4}$ & \texttt{Q} & 318.8 & 10000 & No \\
 % & $10^{-5}$ & \texttt{Q} & 318.8 & 10000 & No \\
 & $10^{-6}$ & \texttt{Q} & 318.8 & 10000 & No \\
\bottomrule
\end{tabular}
\end{minipage}

% \caption{
% Table reports the performance results for the \textbf{QT1} problem instance. Each block corresponds to a fixed value of the regularization parameter (\textbf{Reg.}), while the rows within each block compare the two choices of regularizer (\textbf{Type}) -- von Neumann entropy (\texttt{vN}, $\varphi(z) = z \log z$) and quadratic regularization (\texttt{Q}, $\varphi(z) = \tfrac{1}{2} z^2$) -- under increasingly accuracy tolerances. For each configuration, the table reports the average runtime (\textbf{Time}, in seconds), the number of iterations required for convergence (\textbf{Iter.}), and whether the prescribed tolerance was met (\textbf{Ach.}). 
% }
\caption{
Performance results for the \textbf{QT1} problem instance. Each block corresponds to a fixed regularization parameter (\textbf{Reg.}) and compares von Neumann entropy (\texttt{vN}, $\varphi(z)=z\log z$) and quadratic (\texttt{Q}, $\varphi(z)=\tfrac12 z^2$) regularizers across decreasing tolerances. Reported are runtime (\textbf{Time}, s), iteration count (\textbf{Iter.}), and tolerance attainment (\textbf{Ach.}).
}
\label{tab:q_tom_0}
\end{table}

For the quantum tomography, gradient norm is the main metric of consideration since it measures how well the state~\eqref{eq:q_tom_opt_primal} describes the observations $q_0, \dots, q_{2D}$.
On the contrast to the value $\dualf(\bm \alpha)$, which will merely show the value of the regularization, which is irrelevant to the problem of Quantum Tomography.

Both of the axes are given in log scale and demonstrate convergence of the norm towards 0.
Each subplot corresponds to a distinct regularization value, showing the convergence over iterations for the chosen regularizers.
The solid lines represent the mean values of the norm across repeated runs, while the shaded regions indicate the range between the minimum and maximum norm observed values at each iteration.

The shaded region represents the range of absolute errors across iterations, bounded below by the minimum and above by the maximum error at each iteration.
It intends to demonstrate similar convergence dynamics across different starting points.

\begin{figure}[!htb]
  \centering

  % --- Row 1 ---
  \begin{subfigure}{0.27\textwidth}
    \centering
    \includegraphics[width=\linewidth,height=0.26\textheight,keepaspectratio]{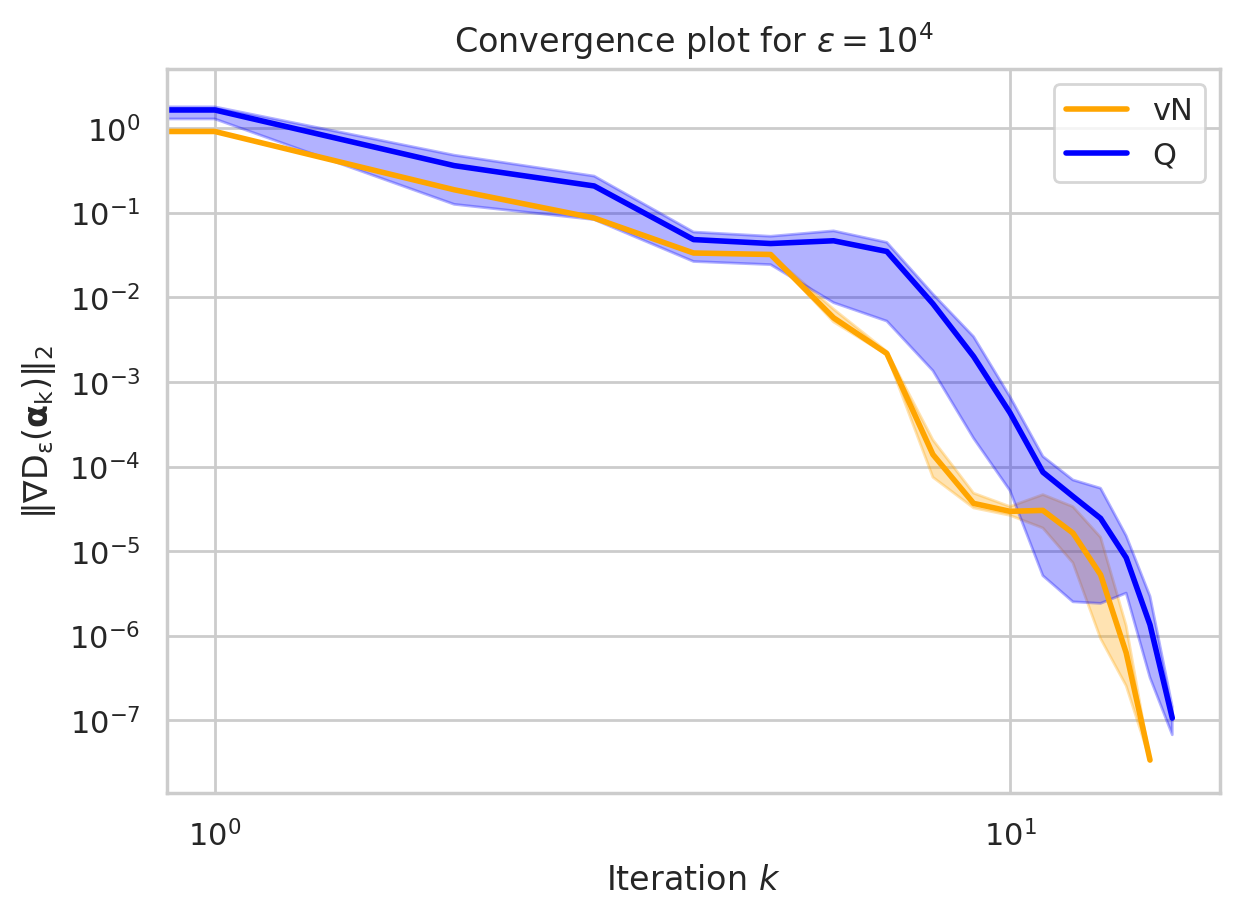}
    \caption{$\varepsilon = 10^{4}$}
  \end{subfigure}\hfill
  \begin{subfigure}{0.27\textwidth}
    \centering
    \includegraphics[width=\linewidth,height=0.26\textheight,keepaspectratio]{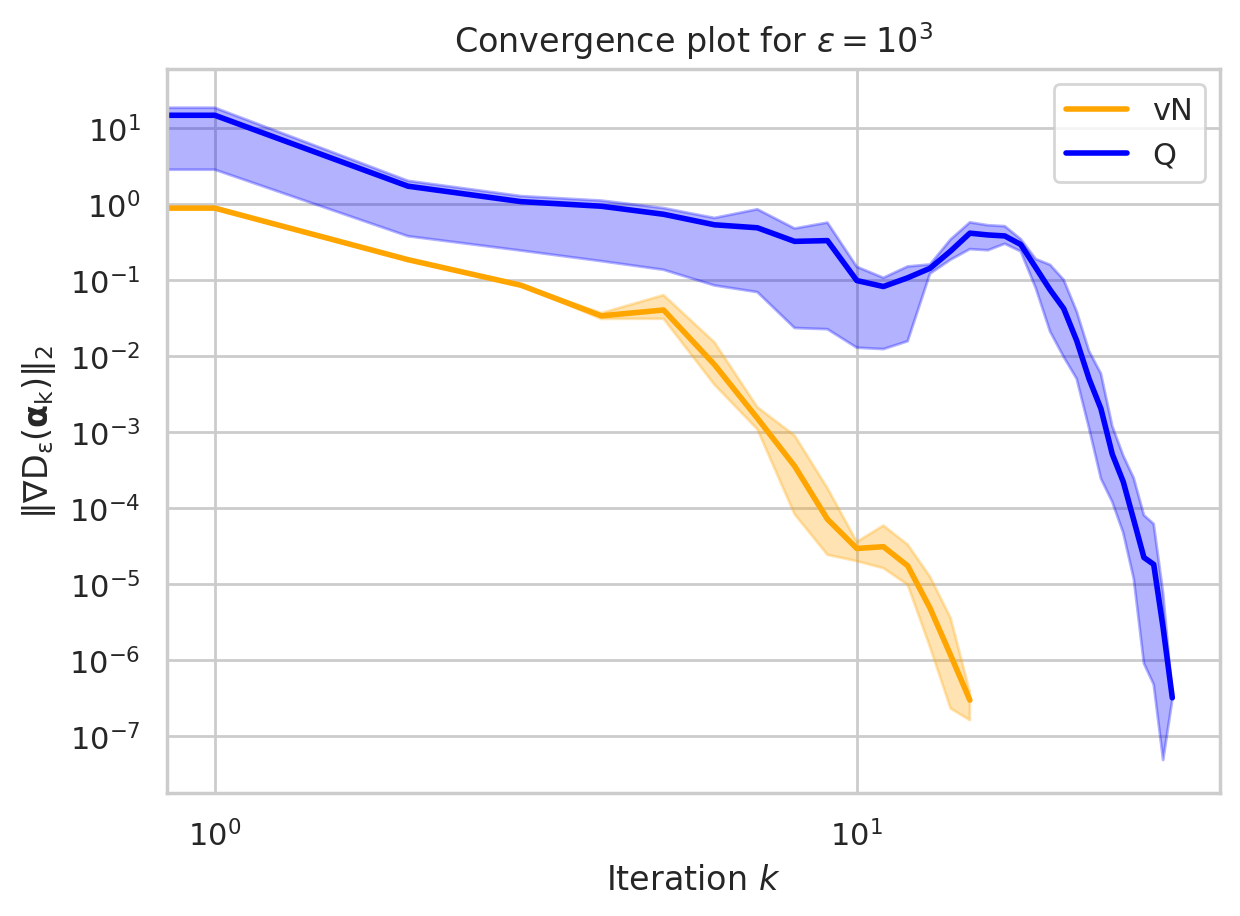}
    \caption{$\varepsilon = 10^{3}$}
  \end{subfigure}\hfill
  \begin{subfigure}{0.27\textwidth}
    \centering
    \includegraphics[width=\linewidth,height=0.26\textheight,keepaspectratio]{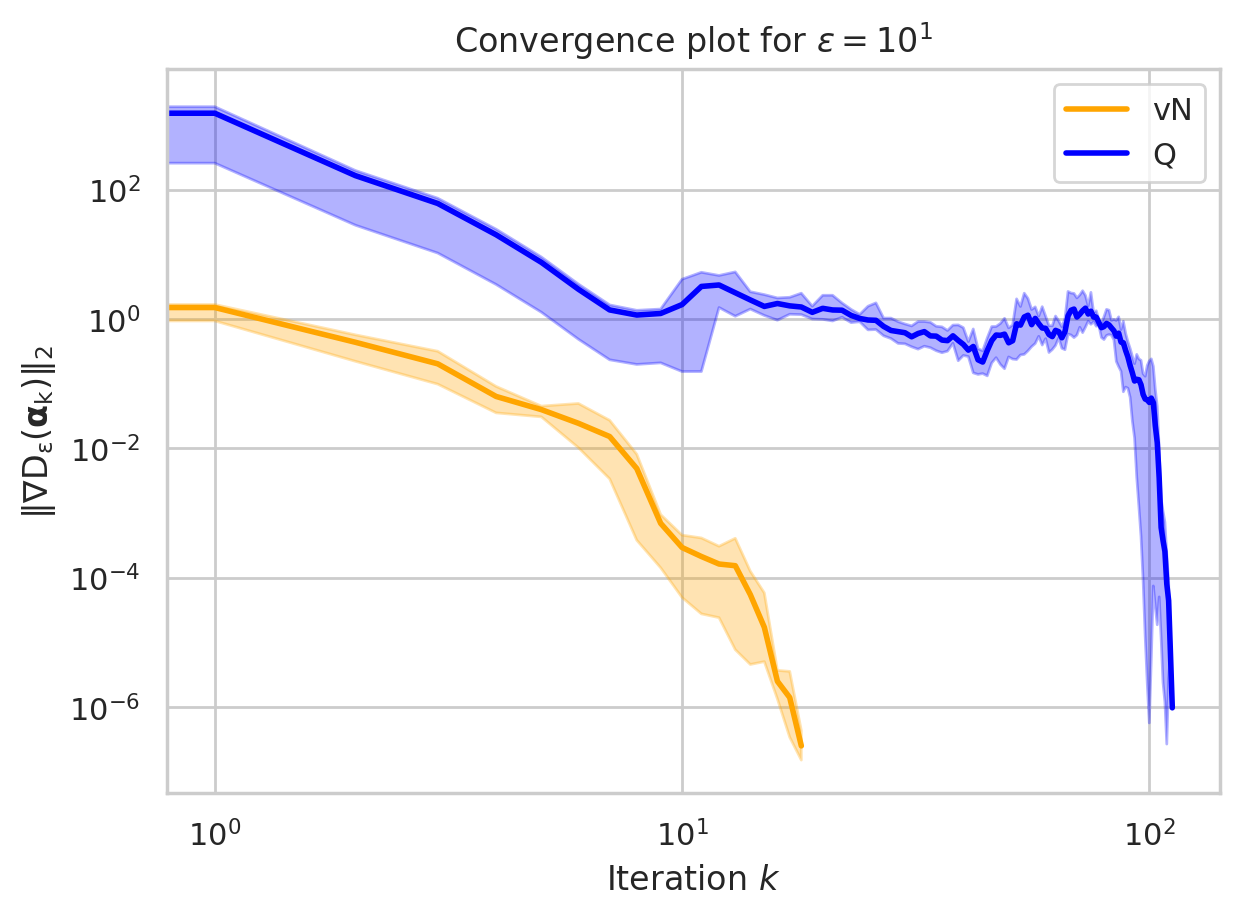}
    \caption{$\varepsilon = 10^{1}$}
  \end{subfigure}

  % \vspace{0.3em}

  % --- Row 2 ---
  \begin{subfigure}{0.27\textwidth}
    \centering
    \includegraphics[width=\linewidth,height=0.26\textheight,keepaspectratio]{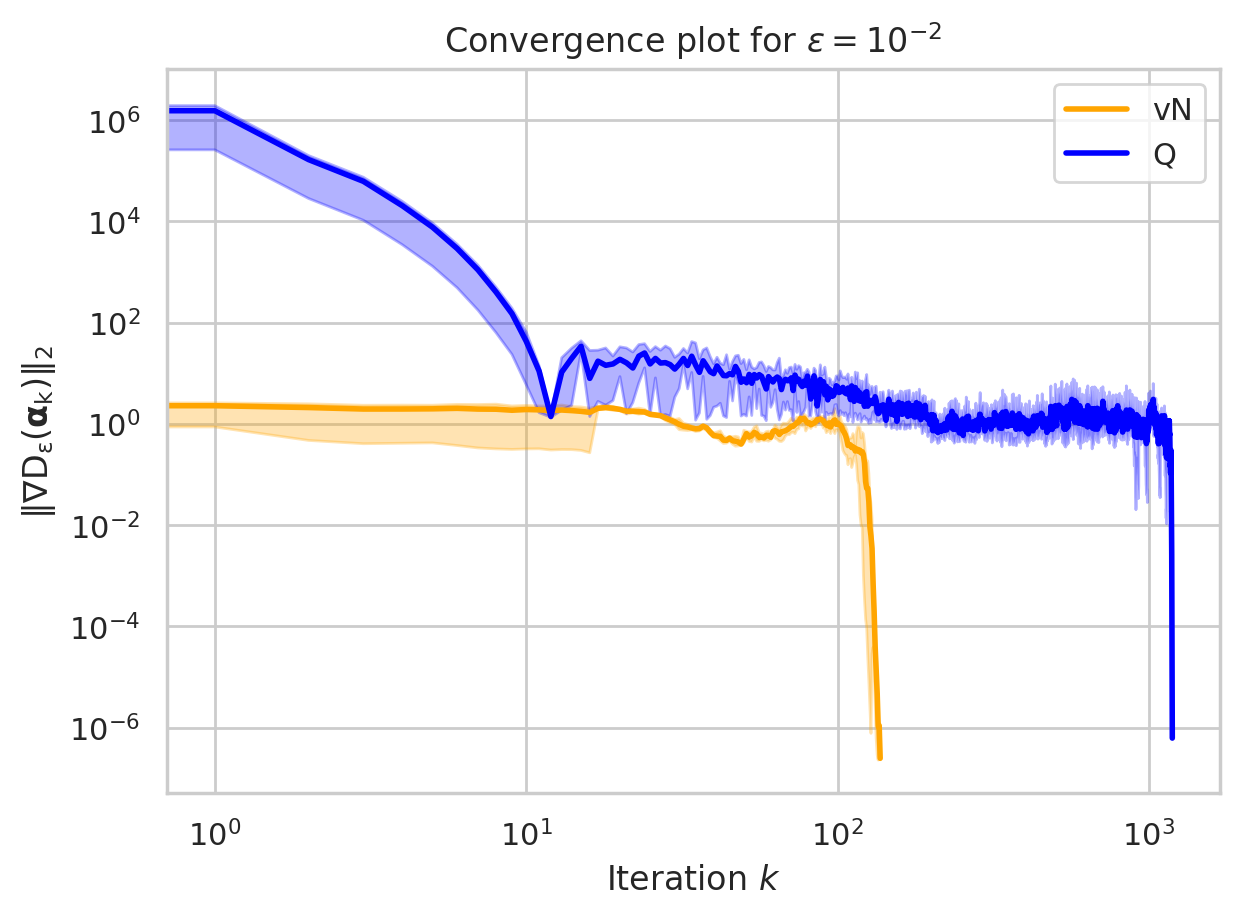}
    \caption{$\varepsilon = 10^{-2}$}
  \end{subfigure}\hfill
  \begin{subfigure}{0.27\textwidth}
    \centering
    \includegraphics[width=\linewidth,height=0.26\textheight,keepaspectratio]{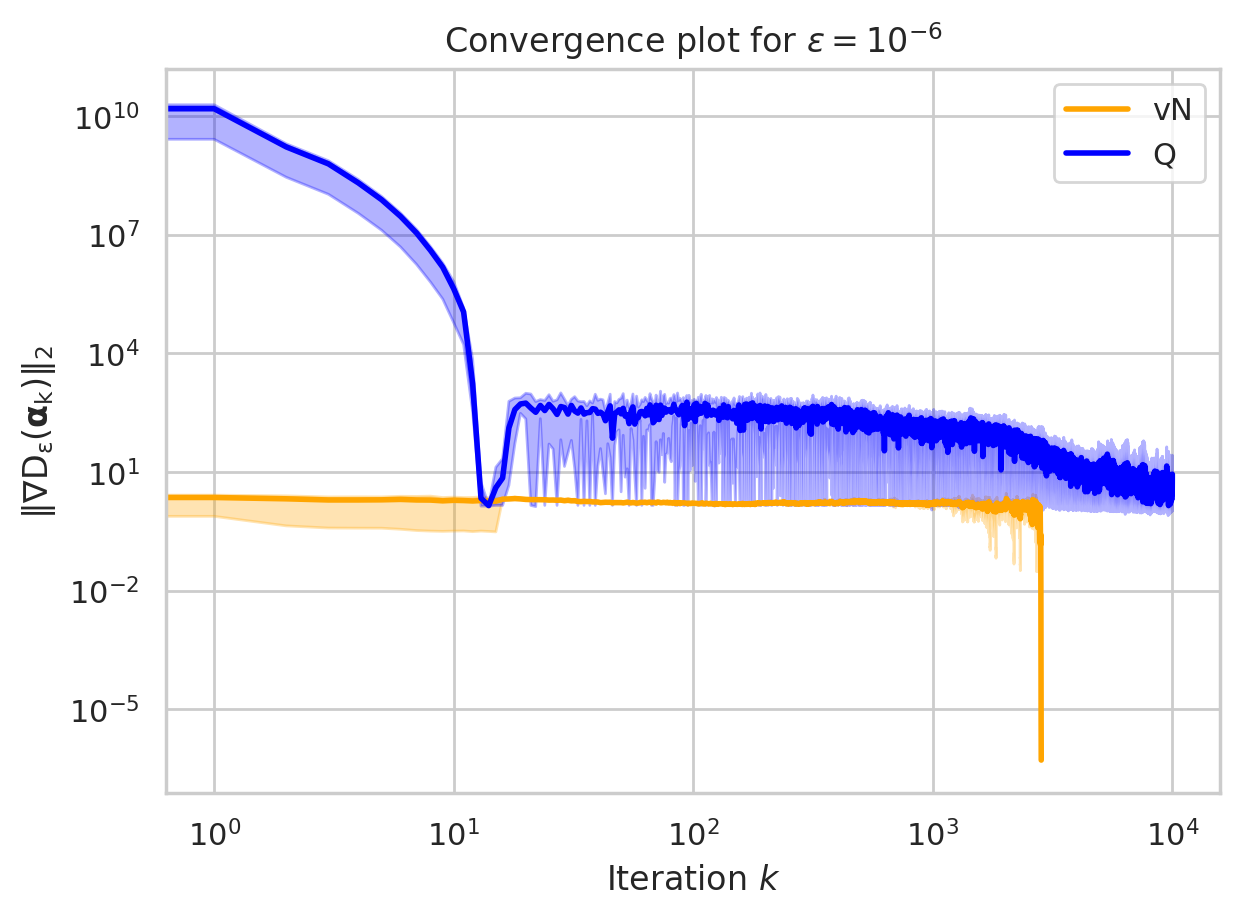}
    \caption{$\varepsilon = 10^{-6}$}
  \end{subfigure}\hfill
  \begin{subfigure}{0.27\textwidth}
    \centering
    \includegraphics[width=\linewidth,height=0.26\textheight,keepaspectratio]{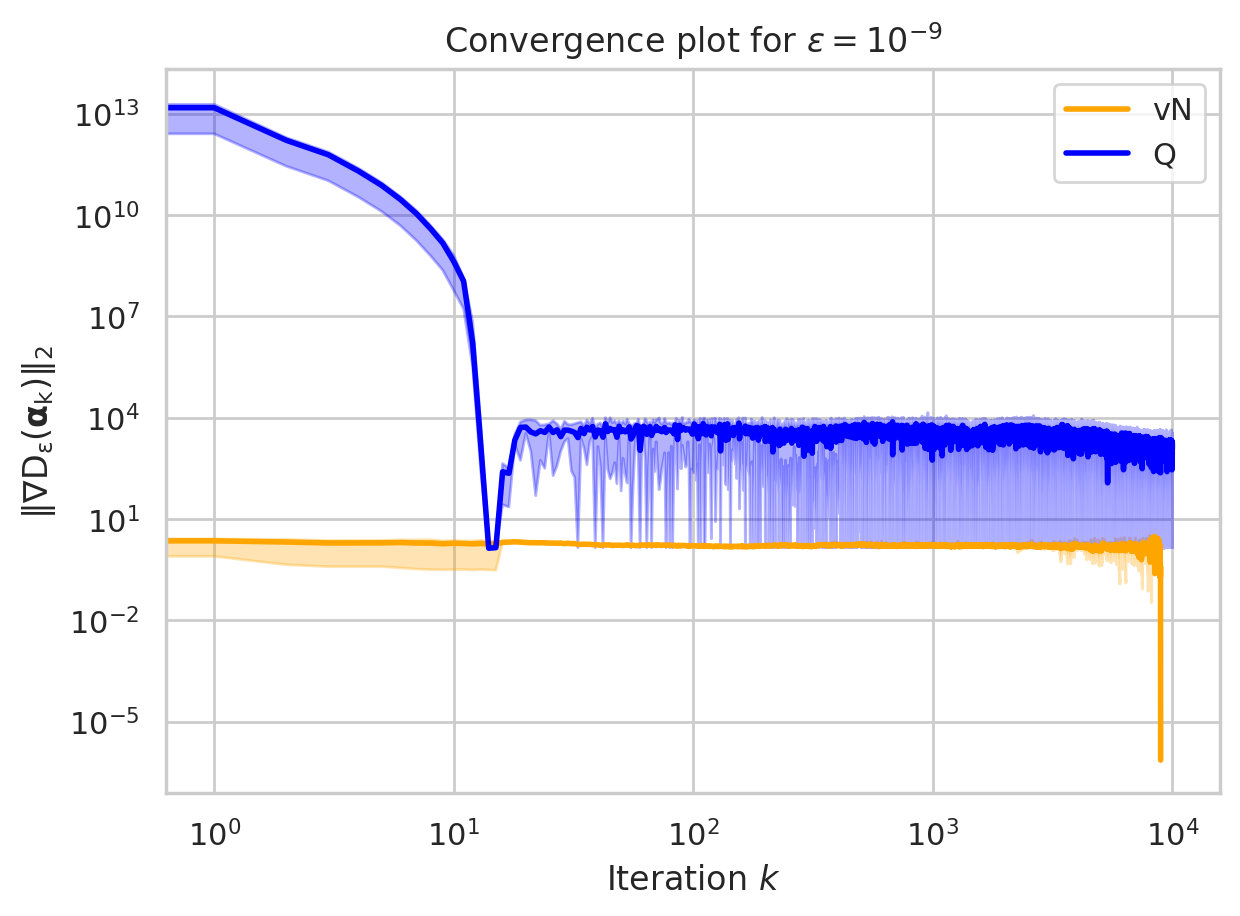}
    \caption{$\varepsilon = 10^{-9}$}
  \end{subfigure}

  \caption{\textbf{Quantum Tomography problem instance \textbf{QT2}.} Graph of the $L_2$-norm of the gradient of the dual functional $\dualf$ in Eq.~\eqref{intro:maindual} for the $\rho_2$ state; see~\eqref{eq:QT2_inst} in section~\ref{sec:numerics}. Panels~(a)--(f) display the iteration trajectories for different values of the regularization parameter $\varepsilon$ 
  % \in \{10^{4}, 10^{3}, 10^{1}, 10^{-2}, 10^{-6}, 10^{-9}\}$ 
  for both von Neumann (orange) and quadratic regularization (blue). }
  \label{fig:q_tom_1}
\end{figure}

\begin{table}[ht]
\centering
\footnotesize
\setlength{\tabcolsep}{2pt}

% ---------- Shared header row (only once) ----------
\begin{minipage}[t]{0.31\textwidth}\centering
\begin{tabular}{@{} L T Y N I A @{}}
\toprule Reg. & Tol. & Type & Time & Iters & Ach. \\ \midrule
\end{tabular}
\end{minipage}\hspace{0.02\textwidth}
\begin{minipage}[t]{0.31\textwidth}\centering
\begin{tabular}{@{} L T Y N I A @{}}
\toprule Reg. & Tol. & Type & Time & Iters & Ach. \\ \midrule
\end{tabular}
\end{minipage}\hspace{0.02\textwidth}
\begin{minipage}[t]{0.31\textwidth}\centering
\begin{tabular}{@{} L T Y N I A @{}}
\toprule Reg. & Tol. & Type & Time & Iters & Ach. \\ \midrule
\end{tabular}
\end{minipage}

% \vspace{0.1em}

% ---------- Row 1 ----------
\begin{minipage}[t]{0.31\textwidth}\centering
\begin{tabular}{@{} L T Y N I A @{}}
\multirow{4}{*}{\textbf{$10^{4}$}} & $10^{-3}$ & \texttt{vN} & 0.3 & 9 & Yes \\
 % & $10^{-4}$ & \texttt{vN} & 0.3 & 9 & Yes \\
 % & $10^{-5}$ & \texttt{vN} & 0.6 & 15 & Yes \\
 & $10^{-6}$ & \texttt{vN} & 0.6 & 15 & Yes \\
 & $10^{-3}$ & \texttt{Q} & 0.3 & 11 & Yes \\
 % & $10^{-4}$ & \texttt{Q} & 0.4 & 13 & Yes \\
 % & $10^{-5}$ & \texttt{Q} & 0.5 & 15 & Yes \\
 & $10^{-6}$ & \texttt{Q} & 0.5 & 16 & Yes \\
\bottomrule
\end{tabular}
\end{minipage}\hspace{0.02\textwidth}
\begin{minipage}[t]{0.31\textwidth}\centering
\begin{tabular}{@{} L T Y N I A @{}}
\multirow{4}{*}{\textbf{$10^{3}$}} & $10^{-3}$ & \texttt{vN} & 0.3 & 9 & Yes \\
 % & $10^{-4}$ & \texttt{vN} & 0.4 & 11 & Yes \\
 % & $10^{-5}$ & \texttt{vN} & 0.5 & 15 & Yes \\
 & $10^{-6}$ & \texttt{vN} & 0.6 & 16 & Yes \\
 & $10^{-3}$ & \texttt{Q} & 0.8 & 25 & Yes \\
 % & $10^{-4}$ & \texttt{Q} & 0.9 & 28 & Yes \\
 % & $10^{-5}$ & \texttt{Q} & 0.9 & 30 & Yes \\
 & $10^{-6}$ & \texttt{Q} & 1.0 & 31 & Yes \\
\bottomrule
\end{tabular}
\end{minipage}\hspace{0.02\textwidth}
\begin{minipage}[t]{0.31\textwidth}\centering
\begin{tabular}{@{} L T Y N I A @{}}
\multirow{4}{*}{\textbf{$10^{1}$}} & $10^{-3}$ & \texttt{vN} & 0.3 & 10 & Yes \\
 % & $10^{-4}$ & \texttt{vN} & 0.5 & 15 & Yes \\
 % & $10^{-5}$ & \texttt{vN} & 0.5 & 16 & Yes \\
 & $10^{-6}$ & \texttt{vN} & 0.6 & 18 & Yes \\
 & $10^{-3}$ & \texttt{Q} & 3.0 & 97 & Yes \\
 % & $10^{-4}$ & \texttt{Q} & 3.0 & 98 & Yes \\
 % & $10^{-5}$ & \texttt{Q} & 3.1 & 100 & Yes \\
 & $10^{-6}$ & \texttt{Q} & 3.1 & 101 & Yes \\
\bottomrule
\end{tabular}
\end{minipage}

\vspace{0.3em}

% ---------- Row 2 ----------
\begin{minipage}[t]{0.31\textwidth}\centering
\begin{tabular}{@{} L T Y N I A @{}}
\multirow{4}{*}{\textbf{$10^{-2}$}} & $10^{-3}$ & \texttt{vN} & 3.5 & 130 & Yes \\
 % & $10^{-4}$ & \texttt{vN} & 3.6 & 131 & Yes \\
 % & $10^{-5}$ & \texttt{vN} & 3.6 & 133 & Yes \\
 & $10^{-6}$ & \texttt{vN} & 3.6 & 134 & Yes \\
 & $10^{-3}$ & \texttt{Q} & 32.4 & 1137 & Yes \\
 % & $10^{-4}$ & \texttt{Q} & 32.4 & 1138 & Yes \\
 % & $10^{-5}$ & \texttt{Q} & 32.5 & 1140 & Yes \\
 & $10^{-6}$ & \texttt{Q} & 32.6 & 1142 & Yes \\
\bottomrule
\end{tabular}
\end{minipage}\hspace{0.02\textwidth}
\begin{minipage}[t]{0.31\textwidth}\centering
\begin{tabular}{@{} L T Y N I A @{}}
\multirow{4}{*}{\textbf{$10^{-6}$}} & $10^{-3}$ & \texttt{vN} & 70.1 & 2829 & Yes \\
 % & $10^{-4}$ & \texttt{vN} & 70.1 & 2831 & Yes \\
 % & $10^{-5}$ & \texttt{vN} & 70.2 & 2833 & Yes \\
 & $10^{-6}$ & \texttt{vN} & 70.3 & 2837 & Yes \\
 & $10^{-3}$ & \texttt{Q} & 317.4 & 10000 & No \\
 % & $10^{-4}$ & \texttt{Q} & 317.4 & 10000 & No \\
 % & $10^{-5}$ & \texttt{Q} & 317.4 & 10000 & No \\
 & $10^{-6}$ & \texttt{Q} & 317.4 & 10000 & No \\
\bottomrule
\end{tabular}
\end{minipage}\hspace{0.02\textwidth}
\begin{minipage}[t]{0.31\textwidth}\centering
\begin{tabular}{@{} L T Y N I A @{}}
\multirow{4}{*}{\textbf{$10^{-9}$}} & $10^{-3}$ & \texttt{vN} & 201.5 & 8195 & Yes \\
 % & $10^{-4}$ & \texttt{vN} & 201.6 & 8197 & Yes \\
 % & $10^{-5}$ & \texttt{vN} & 201.6 & 8198 & Yes \\
 & $10^{-6}$ & \texttt{vN} & 201.6 & 8200 & Yes \\
 & $10^{-3}$ & \texttt{Q} & 320.5 & 10000 & No \\
 % & $10^{-4}$ & \texttt{Q} & 320.5 & 10000 & No \\
 % & $10^{-5}$ & \texttt{Q} & 320.5 & 10000 & No \\
 & $10^{-6}$ & \texttt{Q} & 320.5 & 10000 & No \\
\bottomrule
\end{tabular}
\end{minipage}

\caption{
Performance results for the \textbf{QT2} problem instance. Each block corresponds to a fixed regularization parameter (\textbf{Reg.}) and compares von Neumann entropy (\texttt{vN}, $\varphi(z)=z\log z$) and quadratic (\texttt{Q}, $\varphi(z)=\tfrac12 z^2$) regularizers across decreasing tolerances. Reported are runtime (\textbf{Time}, s), iteration count (\textbf{Iter.}), and tolerance attainment (\textbf{Ach.}).
}
\label{tab:q_tom_1}
\end{table}

For all the three experiments,
increasing the regularization improves the stability of the algorithm and significantly reduces the number of iterations required for convergence --- for example, approximately $10$ iterations for $\varepsilon = 10^{4}$ compared with about $10^{4}$ iterations for $\varepsilon = 10^{-6}$.
In all cases, von Neumann entropy converges faster than quadratic regularization.
Considering the fact that actual solution is independent of the choice of $\varepsilon$, it is more favourable to use large regularization to obtain the solution.

\begin{figure}[!htb]
  \centering

  % --- Row 1 ---
  \begin{subfigure}{0.27\textwidth}
    \centering
    \includegraphics[width=\linewidth,height=0.26\textheight,keepaspectratio]{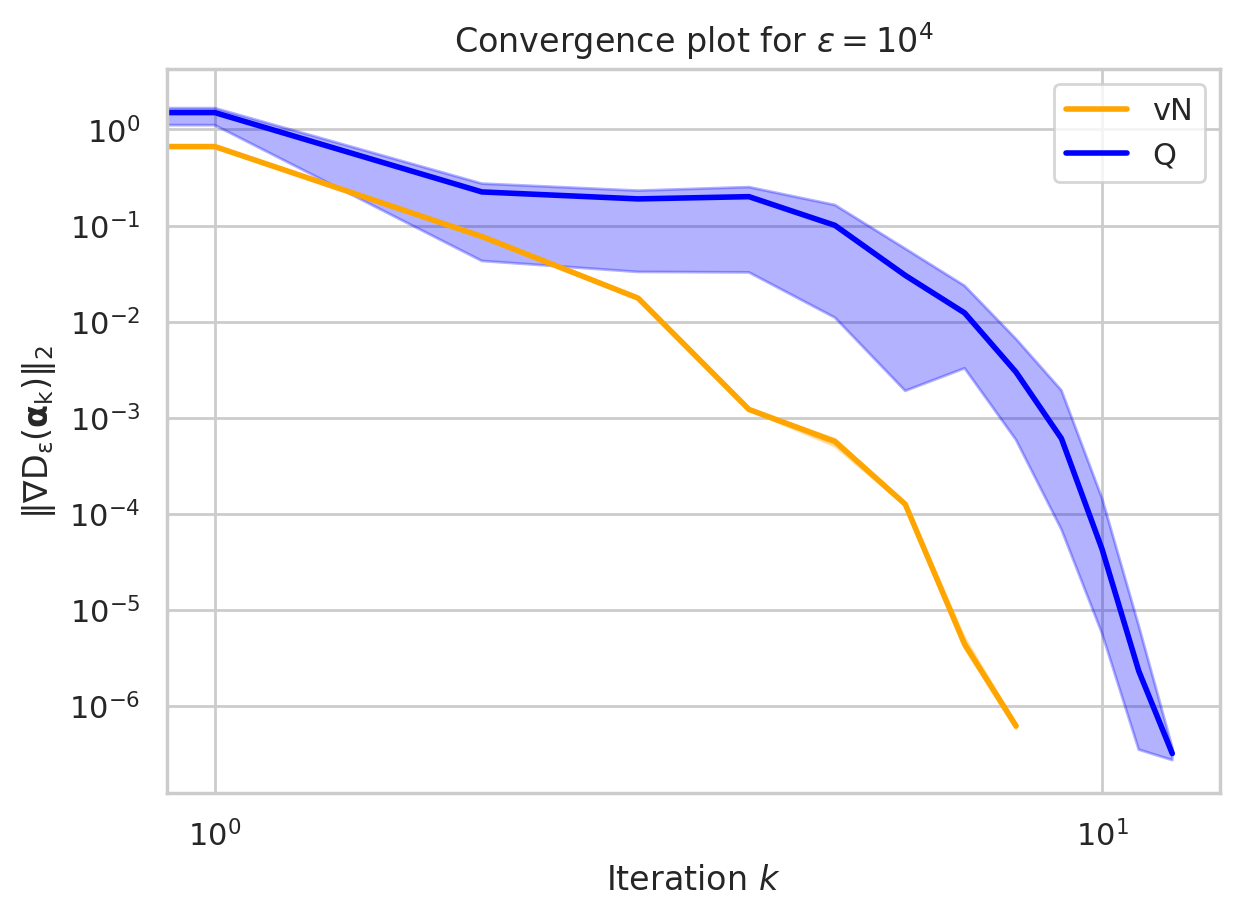}
    \caption{\textbf{QT3} for $\varepsilon = 10^{4}$}
  \end{subfigure}\hfill
  \begin{subfigure}{0.27\textwidth}
    \centering
    \includegraphics[width=\linewidth,height=0.26\textheight,keepaspectratio]{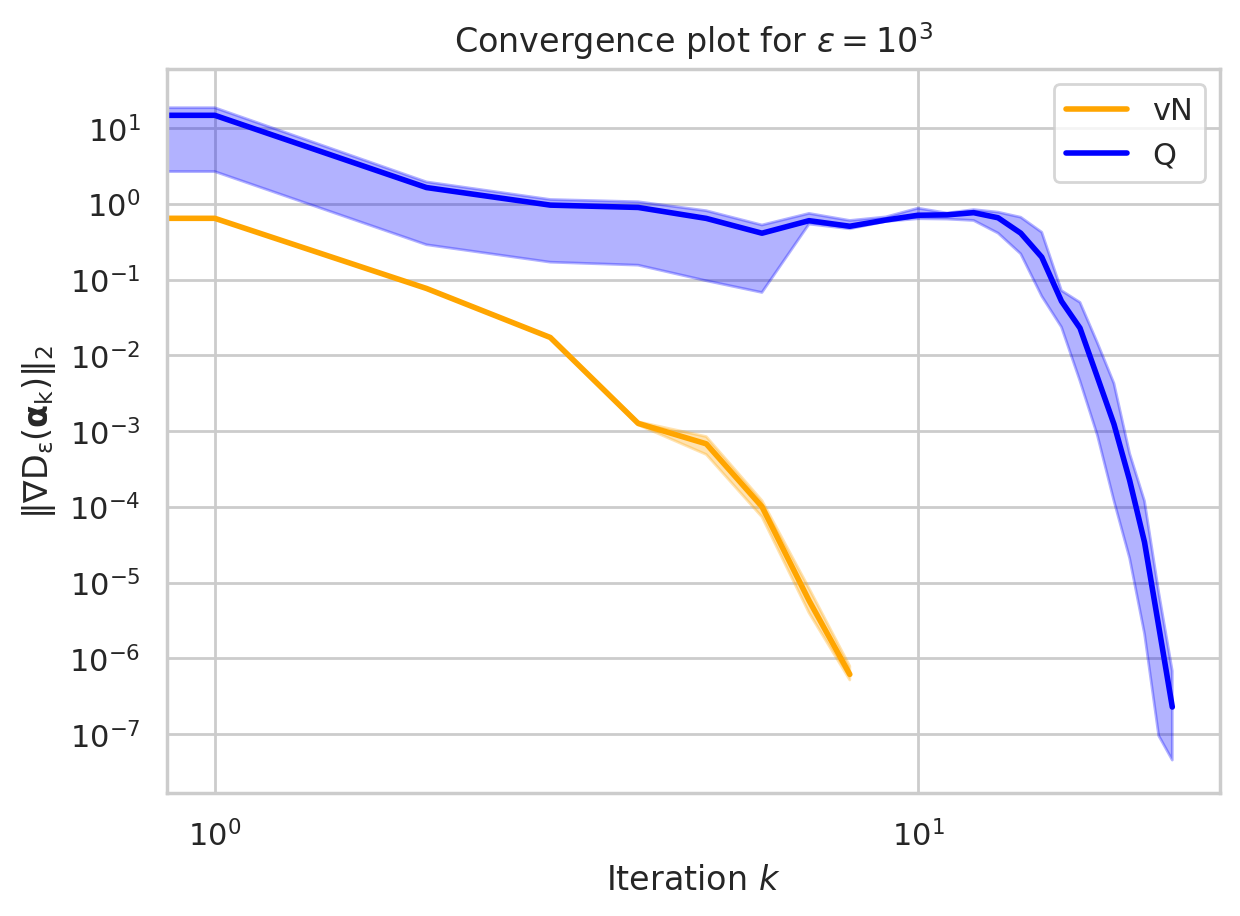}
    \caption{\textbf{QT3} for $\varepsilon = 10^{3}$}
  \end{subfigure}\hfill
  \begin{subfigure}{0.27\textwidth}
    \centering
    \includegraphics[width=\linewidth,height=0.26\textheight,keepaspectratio]{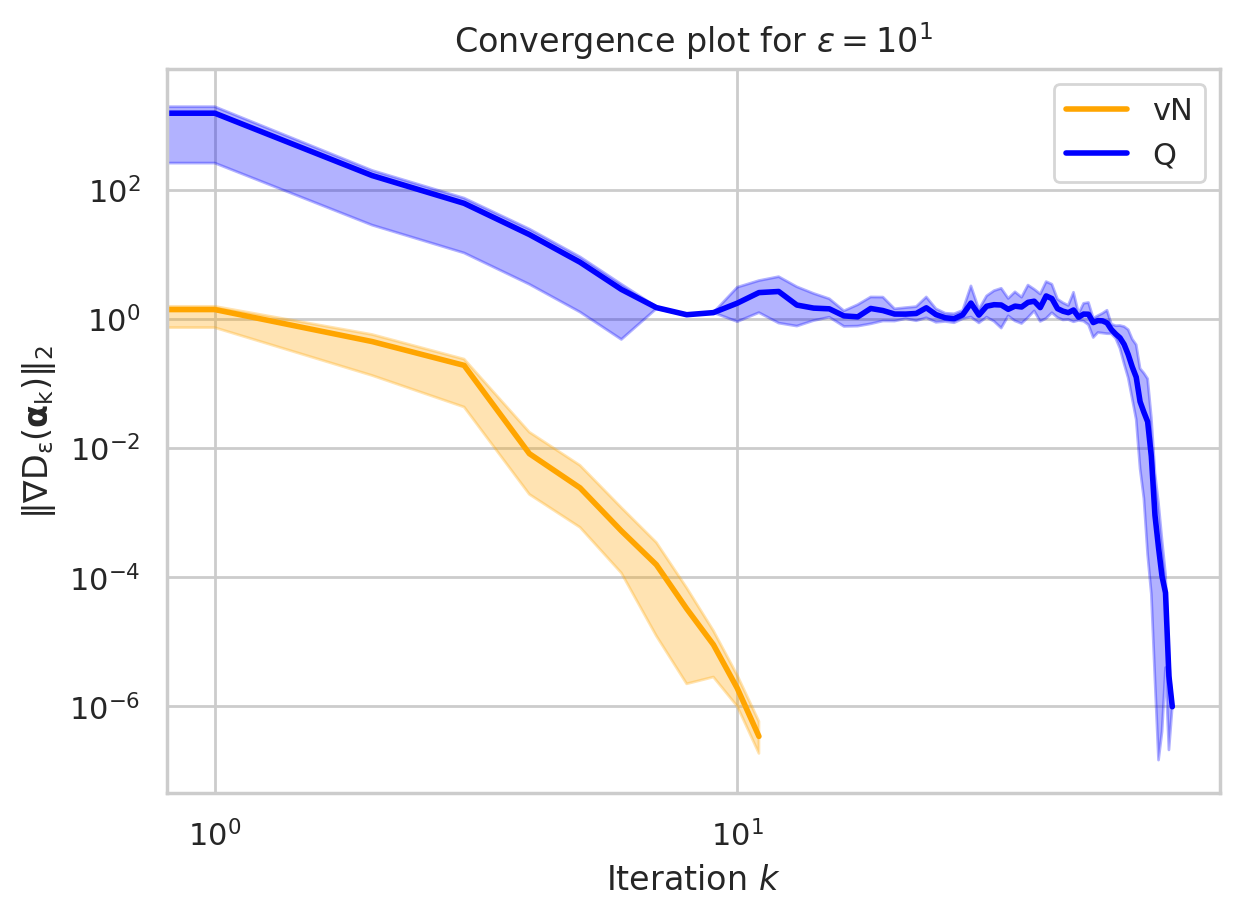}
    \caption{\textbf{QT3} for $\varepsilon = 10^{1}$}
  \end{subfigure}

  % \vspace{0.3em}

  % --- Row 2 ---
  \begin{subfigure}{0.27\textwidth}
    \centering
    \includegraphics[width=\linewidth,height=0.26\textheight,keepaspectratio]{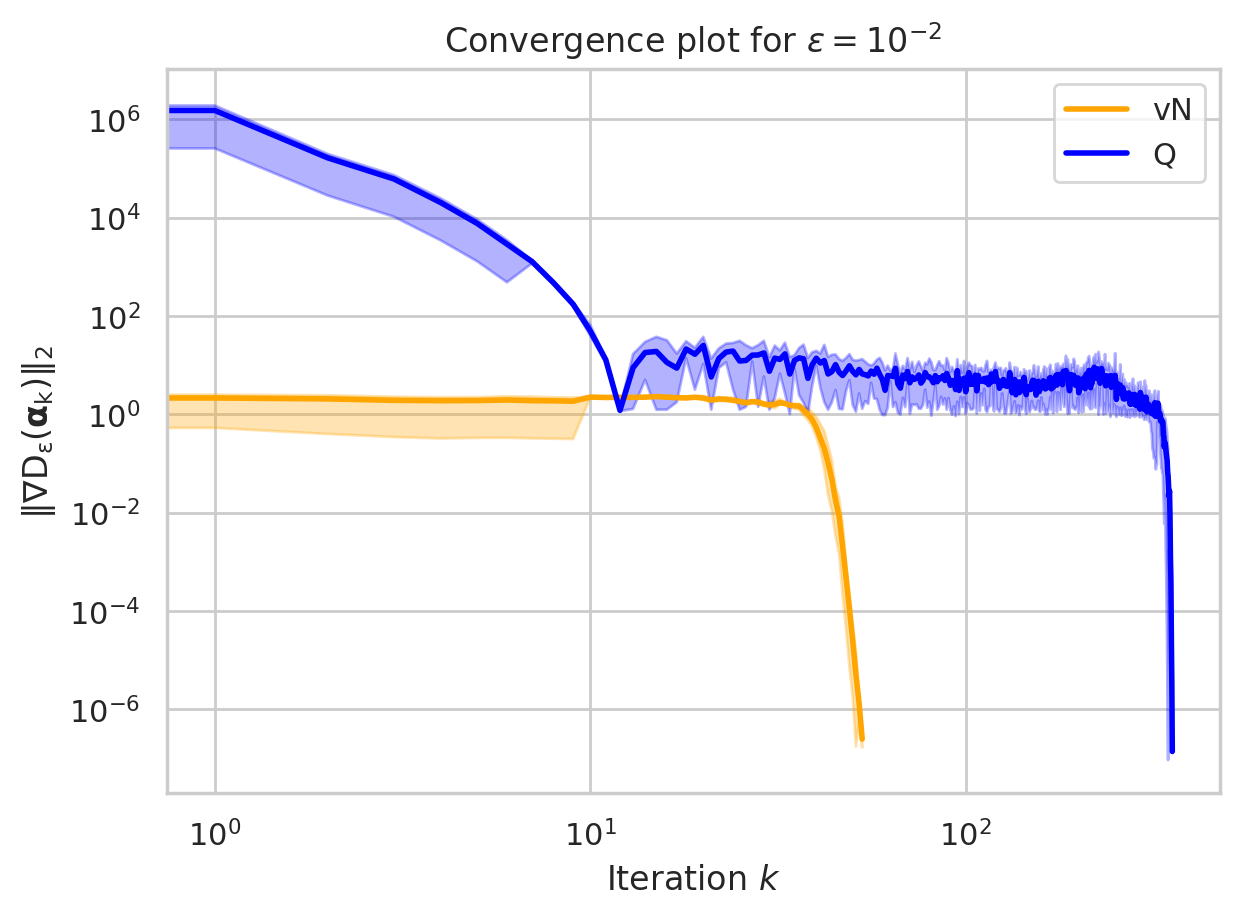}
    \caption{\textbf{QT3} for $\varepsilon = 10^{-2}$}
  \end{subfigure}\hfill
  \begin{subfigure}{0.27\textwidth}
    \centering
    \includegraphics[width=\linewidth,height=0.26\textheight,keepaspectratio]{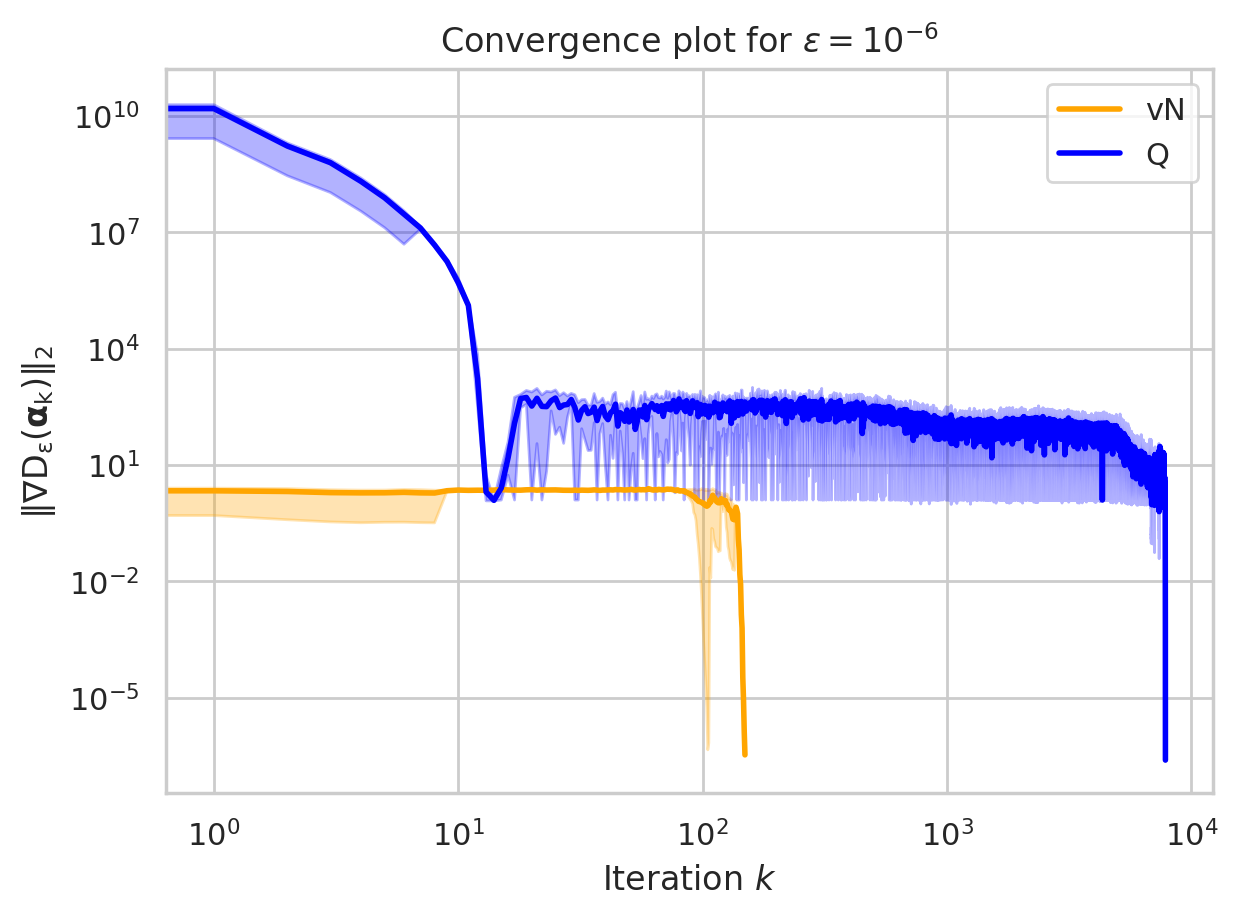}
    \caption{\textbf{QT3} for $\varepsilon = 10^{-6}$}
  \end{subfigure}\hfill
  \begin{subfigure}{0.27\textwidth}
    \centering
    \includegraphics[width=\linewidth,height=0.26\textheight,keepaspectratio]{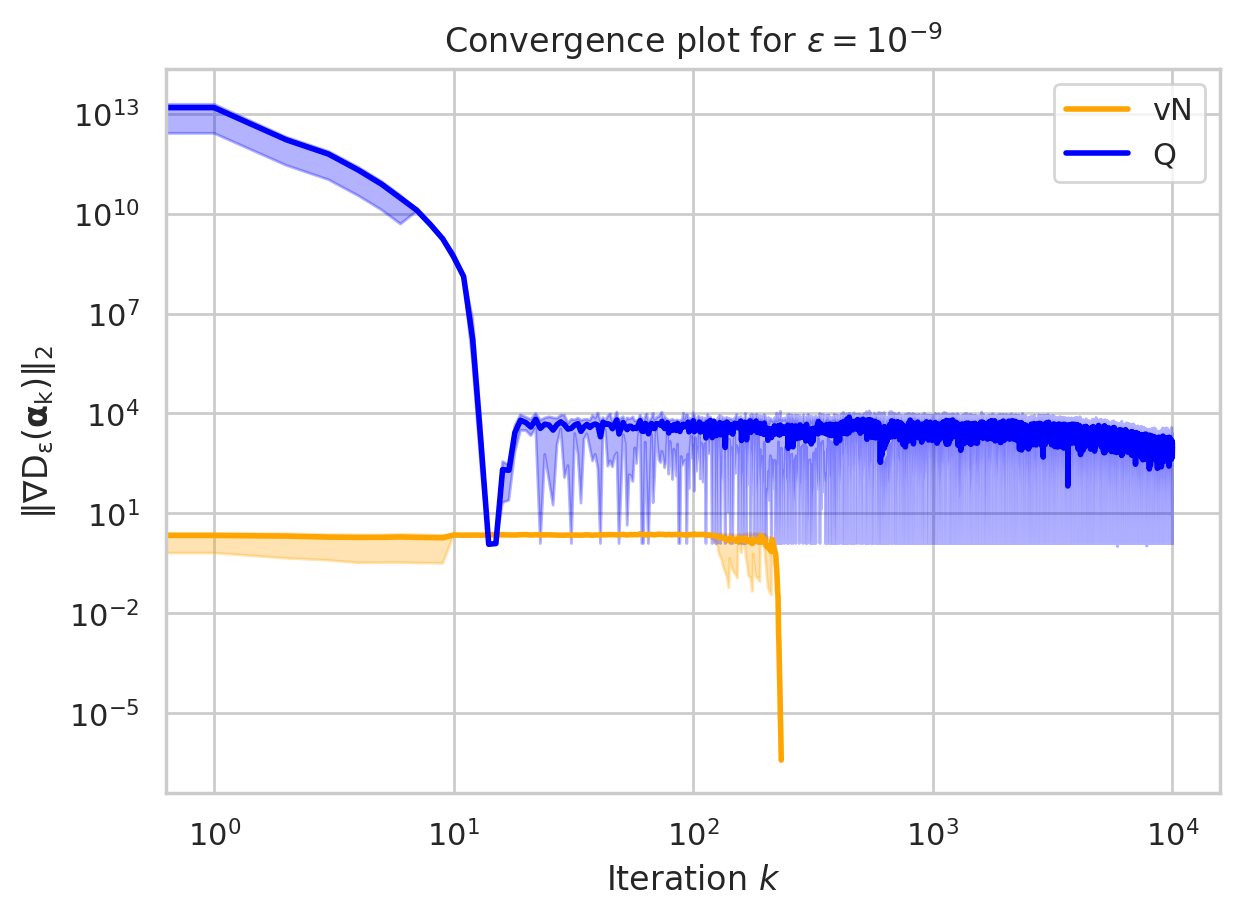}
    \caption{\textbf{QT3} for $\varepsilon = 10^{-9}$}
  \end{subfigure}

  \caption{\textbf{Quantum Tomography problem instance \textbf{QT3}.} Graph of the $L_2$-norm of the gradient of the dual functional $\dualf$ in Eq.~\eqref{intro:maindual} for the $\rho_3$ state; see~\eqref{eq:QT3_inst} in section~\ref{sec:numerics}. Panels~(a)--(f) display the iteration trajectories for different values of the regularization parameter $\varepsilon$ 
  % \in \{10^{4}, 10^{3}, 10^{1}, 10^{-2}, 10^{-6}, 10^{-9}\}$ 
  for both von Neumann (orange) and quadratic regularization (blue). }
  \label{fig:q_tom_2}
\end{figure}

\begin{table}[ht]
\centering
\footnotesize
\setlength{\tabcolsep}{2pt}

% ---------- Shared header row (only once) ----------
\begin{minipage}[t]{0.31\textwidth}\centering
\begin{tabular}{@{} L T Y N I A @{}}
\toprule Reg. & Tol. & Type & Time & Iters & Ach. \\ \midrule
\end{tabular}
\end{minipage}\hspace{0.02\textwidth}
\begin{minipage}[t]{0.31\textwidth}\centering
\begin{tabular}{@{} L T Y N I A @{}}
\toprule Reg. & Tol. & Type & Time & Iters & Ach. \\ \midrule
\end{tabular}
\end{minipage}\hspace{0.02\textwidth}
\begin{minipage}[t]{0.31\textwidth}\centering
\begin{tabular}{@{} L T Y N I A @{}}
\toprule Reg. & Tol. & Type & Time & Iters & Ach. \\ \midrule
\end{tabular}
\end{minipage}

% \vspace{0.1em}

% ---------- Row 1 ----------
\begin{minipage}[t]{0.31\textwidth}\centering
\begin{tabular}{@{} L T Y N I A @{}}
\multirow{4}{*}{\textbf{$10^{4}$}} & $10^{-3}$ & \texttt{vN} & 0.3 & 6 & Yes \\
 % & $10^{-4}$ & \texttt{vN} & 0.4 & 8 & Yes \\
 % & $10^{-5}$ & \texttt{vN} & 0.4 & 8 & Yes \\
 & $10^{-6}$ & \texttt{vN} & 0.4 & 9 & Yes \\
 & $10^{-3}$ & \texttt{Q} & 0.3 & 10 & Yes \\
 % & $10^{-4}$ & \texttt{Q} & 0.4 & 11 & Yes \\
 % & $10^{-5}$ & \texttt{Q} & 0.4 & 12 & Yes \\
 & $10^{-6}$ & \texttt{Q} & 0.4 & 12 & Yes \\
\bottomrule
\end{tabular}
\end{minipage}\hspace{0.02\textwidth}
\begin{minipage}[t]{0.31\textwidth}\centering
\begin{tabular}{@{} L T Y N I A @{}}
\multirow{4}{*}{\textbf{$10^{3}$}} & $10^{-3}$ & \texttt{vN} & 0.3 & 6 & Yes \\
 % & $10^{-4}$ & \texttt{vN} & 0.3 & 7 & Yes \\
 % & $10^{-5}$ & \texttt{vN} & 0.3 & 8 & Yes \\
 & $10^{-6}$ & \texttt{vN} & 0.4 & 9 & Yes \\
 & $10^{-3}$ & \texttt{Q} & 0.7 & 21 & Yes \\
 % & $10^{-4}$ & \texttt{Q} & 0.7 & 22 & Yes \\
 % & $10^{-5}$ & \texttt{Q} & 0.8 & 23 & Yes \\
 & $10^{-6}$ & \texttt{Q} & 0.8 & 24 & Yes \\
\bottomrule
\end{tabular}
\end{minipage}\hspace{0.02\textwidth}
\begin{minipage}[t]{0.31\textwidth}\centering
\begin{tabular}{@{} L T Y N I A @{}}
\multirow{4}{*}{\textbf{$10^{1}$}} & $10^{-3}$ & \texttt{vN} & 0.2 & 7 & Yes \\
 % & $10^{-4}$ & \texttt{vN} & 0.3 & 9 & Yes \\
 % & $10^{-5}$ & \texttt{vN} & 0.4 & 11 & Yes \\
 & $10^{-6}$ & \texttt{vN} & 0.4 & 12 & Yes \\
 & $10^{-3}$ & \texttt{Q} & 2.2 & 63 & Yes \\
 % & $10^{-4}$ & \texttt{Q} & 2.2 & 64 & Yes \\
 % & $10^{-5}$ & \texttt{Q} & 2.3 & 65 & Yes \\
 & $10^{-6}$ & \texttt{Q} & 2.3 & 66 & Yes \\
\bottomrule
\end{tabular}
\end{minipage}

\vspace{0.3em}

% ---------- Row 2 ----------
\begin{minipage}[t]{0.31\textwidth}\centering
\begin{tabular}{@{} L T Y N I A @{}}
\multirow{4}{*}{\textbf{$10^{-2}$}} & $10^{-3}$ & \texttt{vN} & 1.3 & 48 & Yes \\
 % & $10^{-4}$ & \texttt{vN} & 1.4 & 49 & Yes \\
 % & $10^{-5}$ & \texttt{vN} & 1.4 & 51 & Yes \\
 & $10^{-6}$ & \texttt{vN} & 1.5 & 52 & Yes \\
 & $10^{-3}$ & \texttt{Q} & 10.7 & 321 & Yes \\
 % & $10^{-4}$ & \texttt{Q} & 10.8 & 322 & Yes \\
 % & $10^{-5}$ & \texttt{Q} & 10.8 & 323 & Yes \\
 & $10^{-6}$ & \texttt{Q} & 10.8 & 324 & Yes \\
\bottomrule
\end{tabular}
\end{minipage}\hspace{0.02\textwidth}
\begin{minipage}[t]{0.31\textwidth}\centering
\begin{tabular}{@{} L T Y N I A @{}}
\multirow{4}{*}{\textbf{$10^{-6}$}} & $10^{-3}$ & \texttt{vN} & 3.7 & 132 & Yes \\
 % & $10^{-4}$ & \texttt{vN} & 3.7 & 133 & Yes \\
 % & $10^{-5}$ & \texttt{vN} & 3.7 & 134 & Yes \\
 & $10^{-6}$ & \texttt{vN} & 3.8 & 136 & Yes \\
 & $10^{-3}$ & \texttt{Q} & 233.6 & 7065 & Yes \\
 % & $10^{-4}$ & \texttt{Q} & 233.6 & 7066 & Yes \\
 % & $10^{-5}$ & \texttt{Q} & 233.7 & 7068 & Yes \\
 & $10^{-6}$ & \texttt{Q} & 233.7 & 7068 & Yes \\
\bottomrule
\end{tabular}
\end{minipage}\hspace{0.02\textwidth}
\begin{minipage}[t]{0.31\textwidth}\centering
\begin{tabular}{@{} L T Y N I A @{}}
\multirow{4}{*}{\textbf{$10^{-9}$}} & $10^{-3}$ & \texttt{vN} & 5.7 & 209 & Yes \\
 % & $10^{-4}$ & \texttt{vN} & 5.8 & 210 & Yes \\
 % & $10^{-5}$ & \texttt{vN} & 5.8 & 211 & Yes \\
 & $10^{-6}$ & \texttt{vN} & 5.9 & 213 & Yes \\
 & $10^{-3}$ & \texttt{Q} & 319.2 & 10000 & No \\
 % & $10^{-4}$ & \texttt{Q} & 319.2 & 10000 & No \\
 % & $10^{-5}$ & \texttt{Q} & 319.2 & 10000 & No \\
 & $10^{-6}$ & \texttt{Q} & 319.2 & 10000 & No \\
\bottomrule
\end{tabular}
\end{minipage}

% \caption{
% Performance table for the \textbf{QT3} problem instance.
% \textbf{Reg.} - regularization parameter; \textbf{Type} - regularization used. \texttt{vN} (entropy) or \texttt{Q} (quadratic);
% \textbf{Time} - seconds; \textbf{Ach.} - if reached tolerance.
% As the regularization parameter decreases from $10^{4}$ to $10^{-9}$, both runtime and iteration count grow substantially for all tolerances.
% Large~$\varepsilon$ yields fast and stable convergence (on the order of a few iterations), while small~$\varepsilon$ dramatically slows the algorithm.
% In particular, quadratic regularization becomes progressively harder to optimize for $\varepsilon \leq 10^{-6}$ and fails to achieve the target tolerance within the $10^{4}$-iteration budget, whereas von Neumann entropy remains reliably convergent across all tested regimes.
% }

% \caption{
% Table reports the performance results for the \textbf{QT3} problem instance. Each block corresponds to a fixed value of the regularization parameter (\textbf{Reg.}), while the rows within each block compare the two choices of regularizer (\textbf{Type}) -- von Neumann entropy (\texttt{vN}, $\varphi(z) = z \log z$) and quadratic regularization (\texttt{Q}, $\varphi(z) = \tfrac{1}{2} z^2$) -- under increasingly accuracy tolerances. For each configuration, the table reports the average runtime (\textbf{Time}, in seconds), the number of iterations required for convergence (\textbf{Iter.}), and whether the prescribed tolerance was met (\textbf{Ach.}). 
% }
\caption{
Performance results for the \textbf{QT3} problem instance. Each block corresponds to a fixed regularization parameter (\textbf{Reg.}) and compares von Neumann entropy (\texttt{vN}, $\varphi(z)=z\log z$) and quadratic (\texttt{Q}, $\varphi(z)=\tfrac12 z^2$) regularizers across decreasing tolerances. Reported are runtime (\textbf{Time}, s), iteration count (\textbf{Iter.}), and tolerance attainment (\textbf{Ach.}).
}
\label{tab:q_tom_2}
\end{table}

All the three experiments demonstrate identical behaviour.
As the regularization parameter decreases from $10^{4}$ to $10^{-9}$, both runtime and iteration count grow substantially for all tolerances.
Large~$\varepsilon$ yields fast and stable convergence (on the order of a few iterations), while small~$\varepsilon$ dramatically slows the algorithm. 
In particular, quadratic regularization becomes
progressively harder to optimize for $\varepsilon \leq 10^{-6}$ and fails to achieve the target tolerance within the $10000$-iteration budget, whereas von Neumann entropy remains reliably convergent across all tested regimes.

\FloatBarrier

\subsection{Quantum Optimal Transport}
\label{subsec:qot}

Let $\cH$ be a $n$-dimensional Hilbert space, and let $\rho, \sigma$ be density matrices on $\cH$, i.e. positive semidefinite matrices with unit trace.
Next, let $H$ denote the system Hamiltonian, which is self-adjoint matrix on $\cH \otimes \cH$. 
The Quantum Optimal Transport (QOT) problem~\cite{CalGolPau18,ColeEckFriZyc-MPAG-23,DPaTre19,DPaTre2023} is defined by
\begin{equation}
    \label{eq:qot_primal}
    \min \left\{
        \Tr[H \pi] \ : \ \pi \in \rmH_{\geq}(\cH \otimes\cH), \,  
        \Tr_2[\pi] = \rho \text{ and }
        \Tr_1[\pi] = \sigma
    \right\},
\end{equation}
where $\Tr_1$ ($\Tr_2$) %{\color{blue}[LP: I thought of swapping $\Tr_1$ and $\Tr_2$ above, for coherence with our paper -- but incoherently with chemistry/quantum mechanics, where the index denotes which variable is left out. Not sure, for now i leave it like this]}, 
 denotes partial trace with respect to the first (second) subsystem, i.e., defined by the relation
\begin{align}
    \Tr[(\Id \otimes U) \pi] 
        =
    \Tr[U \Tr_1[\pi]]
        \, , \qquad 
    \forall U \in \tH(\cH)
        \, ,
\end{align}
and analogously for $\Tr_2$ replacing $\Id \otimes U$ with $U \otimes \Id$. 

For more details and a historical overview on this topic, we refer to Section \ref{sec:intro}.

It turns out that QOT problem~\eqref{eq:qot_primal} is an instance of semidefinite programming~\cite{ColeEckFriZyc-MPAG-23}, as it is demonstrated below.
Let $\{\ket{i}\}_{i=1}^n$ be an orthonormal basis on $\cH$, and define the following operators
\begin{equation*}
\begin{gathered}
    E_{ij} = \ket{i}\bra{j}, \quad
    G_{ij} = \frac{1}{2}(E_{ij} + E_{ji}), \quad
    H_{ij} = \frac{\imag}{2}(E_{ij} - E_{ji}), \quad
    \forall \, 1\leq i \leq j \leq n,
\end{gathered}
\end{equation*}
where $\imag = \sqrt{-1}$. 
Note, that $G_{ij}$ and $H_{i j}$ are Hermitian and form a basis of the $n \times n$-dimensional space of all Hermitian matrices over $\cH$, so we can use them to reformulate the constraints.
We use the notation $\Tr_2[\pi] = [p^1_{ij}]_{i, j =1}^n$, and $\Tr_1[\pi] = [p^2_{ij}]_{i, j=1}^n$ to indicate the matrix representation of $\Tr_2[\pi]$ and $\Tr_1[\pi]$
with respect to the standard basis $\{E_{ij}\}_{i,j=1}^n$. 
Then, for all $1 \leq i \leq j \leq n$, we have
\begin{align*}
    \Tr[(G_{ij} \otimes \Id) \pi]
    &= \Tr[G_{ij} \Tr_2[\pi]]
    = \frac{p^1_{ij} + p^1_{ji}}{2}
    \quad 1 \leq i \leq j \leq n,
    \\
    \Tr[(H_{ij} \otimes \Id) \pi]
    &= \Tr[H_{ij} \Tr_2[\pi]]
    = \imag\frac{p^1_{ij} - p^1_{ji}}{2}
    \quad 1 \leq i < j \leq n.
\end{align*}
Equivalently, for the second partial trace
\begin{align*}
    \Tr[(\Id \otimes G_{ij}) \pi]
    &= \Tr[G_{ij} \Tr_1[\pi]]
    = \frac{p^2_{ij} + p^2_{ji}}{2}
    \quad 1 \leq i \leq j \leq n,
    \\
    \Tr[(\Id \otimes H_{ij}) \pi]
    &= \Tr[H_{ij} \Tr_1[\pi]]
    = \imag\frac{p^2_{ij} - p^2_{ji}}{2}
    \quad 1 \leq i < j \leq n.
\end{align*}
In other words, the partial trace constraints in the variational problem in~\eqref{eq:qot_primal} can be equivalently written as $M = 2 n^2$ equations having the following form of the semi-definite programming~\eqref{intro:main}
\begin{align*}
    \Tr[(G_{ij} \otimes \Id) \pi] = \Re(\rho_{ij}),\
    \Tr[(\Id \otimes G_{ij}) \pi] = \Re(\sigma_{ij})
    \quad 1 \leq i \leq j \leq n,
    \\
    \Tr[(H_{ij} \otimes \Id) \pi] = \Im(\rho_{ij}), \
    \Tr[(\Id \otimes H_{ij}) \pi] = \Im(\sigma_{ij})
    \quad 1 \leq i < j \leq n.
\end{align*}

The dual problem of~\eqref{eq:qot_primal} is given by
\begin{align}
    \label{eq:qot_dual}
    \sup \left\{
        \Tr[U \rho] + \Tr[V \sigma] \ : \ H - U \oplus V \geq 0
    \right\},
\end{align}
where $U \oplus V = U \otimes \Id + \Id \otimes V$ denotes the Kronecker sum.
Under the condition that $\rho, \sigma > 0$, Theorem~\ref{thm:duality_zero} guarantees strong duality and the existence of the maximizer for~\eqref{eq:qot_dual}.

Below, we illustrate the impact of regularization on several important theoretical QOT instances.

\noindent
\textbf{Numerical results.} Below, we present two Quantum Optimal Transport instances to demonstrate the regularization approach.

\paragraph{Quantum Wasserstein Distance (QWD).}
\label{sec:QC}
The first instance arises in~\cite{CalGolPau18}.
Initially, the setting is given in the infinite-dimensional space, and further we give a finite-dimensional discretization, which reduces the problem to the form suitable for the numerical tests.

% Consider $N = 2$ identical Hilbert spaces $\cH=L^2(\mathbb{R}^m)$, where $m \in \mathbb{N}$. Let $\{ X_j \}_{j \in [m]}$ and $\{ P_j \}_{j \in [m]}$ be the position and momentum operators, respectively. These act on functions $\psi \in L^2(\mathbb{R}^m)$ as
% \begin{align*}
%     (X_j \psi)(x) = x_j \psi(x), \, \text{ and } \, 
%     (P_j \psi)(x) = -\imag \frac{\partial}{\partial x_j} \psi(x), \quad \forall\, j \in [m],\,  x = (x_1, \dots, x_m) \in \mathbb{R}^m.
% \end{align*}
% The Hamiltonian $\textbf{H}$ is the bounded from below operator defined through the quadrature operators as
% \begin{align}\label{eq:qwass_hamilt}
%     \textbf{H} 
%     = \sum_{i = 1}^m \left(
%         X_i \otimes \Id - \Id \otimes X_i
%     \right)^2
%     + \sum_{i = 1}^m \left(
%         P_i \otimes \Id - \Id \otimes P_i
%     \right)^2.
% \end{align}

We start by considering the Hilbert space $L^2(\mathbb{R})$. Let $X$ and $P$ be the position and momentum operators, respectively. These are unbounded self-adjoint operators on $L^2(\mathbb{R})$ defined (on smooth functions) as
\begin{align*}
    (X \psi)(x) := x \psi(x), \, \text{ and } \, 
    (P \psi)(x) := -\imag \frac{\partial}{\partial x} \psi(x), \quad \forall x \in \mathbb{R}.
\end{align*}
The Hamiltonian $\textbf{H}$ is the nonnegative (hence, self-adjoint) unbounded operator on $L^2(\mathbb{R}) \otimes L^2(\mathbb{R})$ defined through the quadrature operators $X$ and $P$ as
\begin{align}\label{eq:qwass_hamilt}
    \textbf{H} 
    = \left(
        X \otimes \Id - \Id \otimes X
    \right)^2
    + \left(
        P \otimes \Id - \Id \otimes P
    \right)^2.
\end{align}

Consider $n$-dimensional Hilbert space $\cH$ with orthonormal basis $\ket{1}, \dots, \ket{n}$.
As a basis of $\cH \otimes \cH$, we use vectors formed by tensor products of single Hilbert space $\cH$, which are $\ket{v_{i j}} = \ket{i} \otimes \ket{j},\ 1\leq i, j \leq n$.
Using this basis, we build a finite-dimensional Hamiltonian $H \in \rmH(\cH \otimes \cH)$ which is given by
% The finite-dimensional representation of the Hamiltonian is given by
\begin{align}
\label{eq:qwass_ham}
    H_{(ij)(kl)}
    = \bra{v_{i j}} \textbf{H} \ket{v_{k l}}, \quad 1 \leq i, j, k, l \leq n.
\end{align}
% \nat{so what was $H_{ho}$ for? Only to define the special basis?}YES
% resulting in a Hermitian matrix $H \in \rmH(\cH \otimes \cH)$ we further use to define the problem~\eqref{eq:qot_primal}.
Further, we use it to define the problem~\eqref{eq:qot_primal}.

Next, we need to define the density matrices $\rho$ and $\sigma$ on $\cH$.
We use \emph{Gaussian states}, which is an important class of states in continuous-variable quantum systems.
They are fully characterized by the first and second statistical moments of the quadrature operators $X$ and $P$, 
and their Wigner function has a Gaussian profile in phase space.
For a reference on Gaussian states and related theory, one may see~\cite{brask-arxiv-2022}.
% {\color{blue} [LP: add a reference to the topic, just for completeness?]}

A Gaussian state $\rho_{m,V}$ on $L^2(\R)$ is a quantum state that can be fully described by a mean vector
\[
\bm m = \begin{pmatrix} \langle R_1 \rangle_{\rho_{m,V}} \\ \langle R_2 \rangle_{\rho_{m,V}} \end{pmatrix} \in \R^2,
\]
and a covariance matrix 
\[
V_{ij} = \tfrac{1}{2}\Tr\!\bigl[\rho_{m,V} \{R_i - \langle R_i \rangle_{\rho_{m,V}},\, R_j - \langle R_j \rangle_{\rho_{m,V}}\}\bigr], 
\quad i, j =1,2.
\]
$V$ is a real, symmetric, positive-definite $2\times 2$ matrix satisfying the uncertainty principle $V + \tfrac{i}{2}\Omega \geq 0$, where $\Omega = \begin{psmallmatrix} 0 & 1 \\ -1 & 0 \end{psmallmatrix}$.
Above, $R_1 = X,\ R_2 = P$, $\{\cdot, \cdot\}$ denotes anticommutator, i.e. the operator defined as $\{A, B\} = AB + BA$ for some operators $A, B$ on $L^2(\R)$, and $\langle R_i \rangle_{\rho_{m,V}} = \operatorname{Tr}[R_i \rho_{m,V}],\ i =1,2$.
In order to obtain quantum state $\pi$ on $\cH$ from the Gaussian state $\rho_{m, V}$, the quantum state on $L^2(\R)$, we discretize $\rho_{m, V}$ using the basis of $\cH$ analogously to the Hamiltonian $H$.
That is, $\pi$ is given by
\begin{align}
    \pi
    = \frac{P}{\operatorname{Tr}[P]}, \quad
    P_{i j} = \bra{i}\rho_{m, V} \ket{j},\ i, j = 1, \dots, n.
\end{align}

For the following test, we fix $n = 50$.
To construct marginals $\rho$ and $\sigma$, we discretize as above two Gaussian states with following parameters
\begin{align*}
    \bm m_\rho = \begin{pmatrix}
        0 \\ 0
    \end{pmatrix}, \
    V_\rho = \begin{pmatrix}
        3 & 1 \\ 1 & 3
    \end{pmatrix}, \
    \bm m_\sigma = \begin{pmatrix}
        1 \\ 1
    \end{pmatrix}, \
    V_\sigma = \begin{pmatrix}
        10 & 2 \\ 2 & 10
    \end{pmatrix}.
\end{align*}
Finally, the Hamiltonian $H$ is constructed as in the~\eqref{eq:qwass_ham}.

\begin{figure}[!htb]
  \centering

  % --- Row ---
  \begin{subfigure}{0.27\textwidth}
    \centering
    \includegraphics[width=\linewidth,height=0.26\textheight,keepaspectratio]{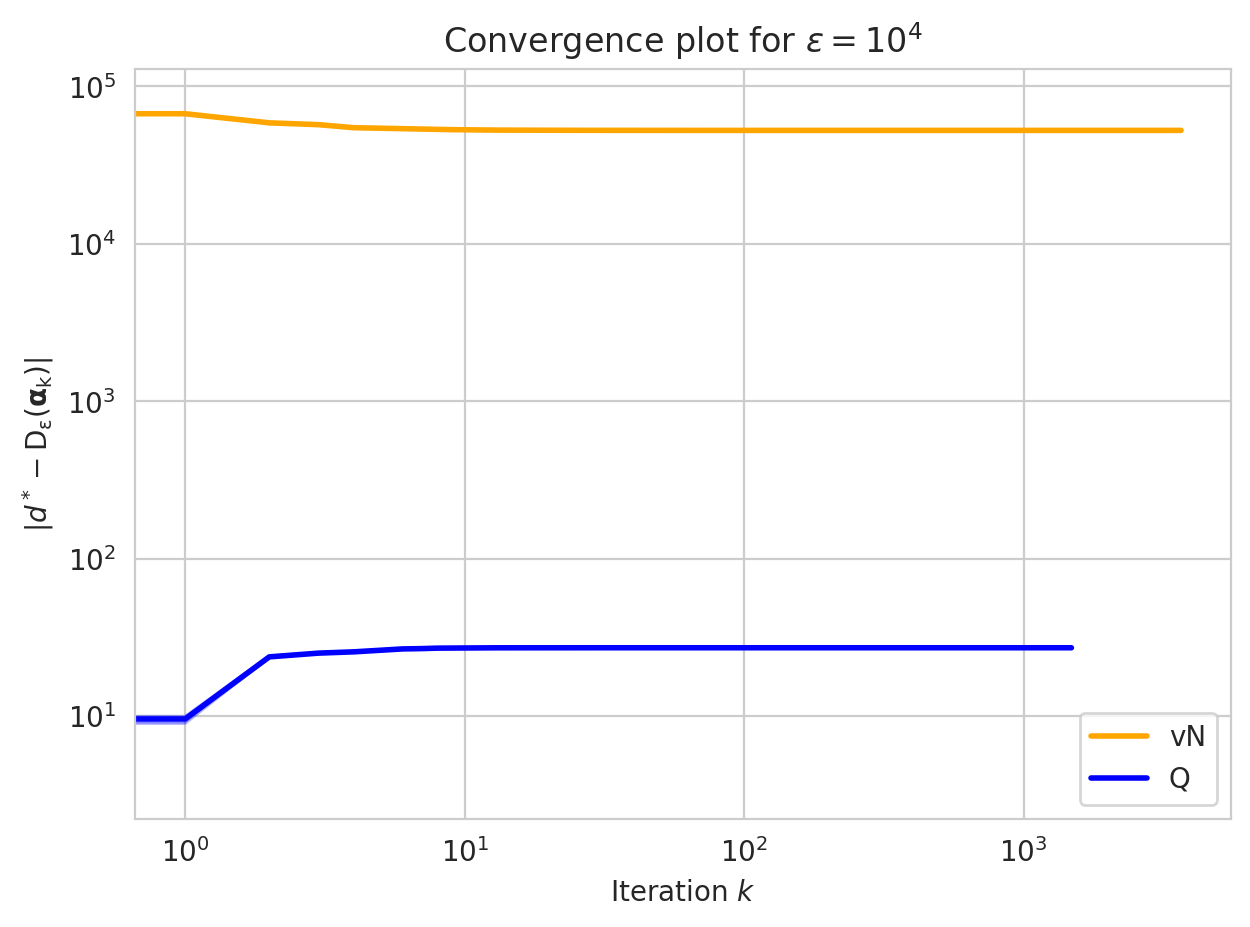}
    \caption{\textbf{QWD} for $\varepsilon = 10^{4}$}
  \end{subfigure}\hfill
  \begin{subfigure}{0.27\textwidth}
    \centering
    \includegraphics[width=\linewidth,height=0.26\textheight,keepaspectratio]{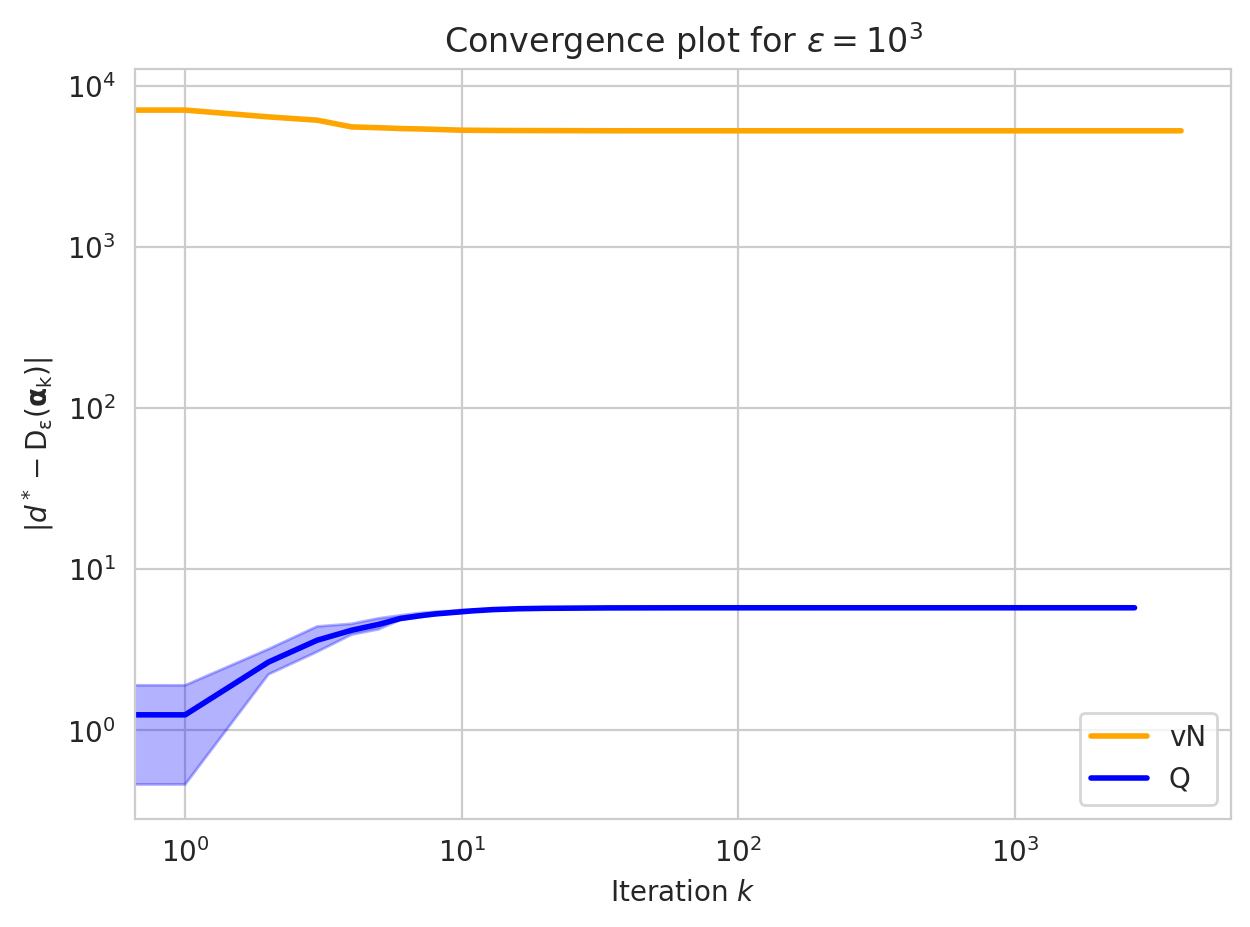}
    \caption{\textbf{QWD} for $\varepsilon = 10^{3}$}
  \end{subfigure}\hfill
  \begin{subfigure}{0.27\textwidth}
    \centering
    \includegraphics[width=\linewidth,height=0.26\textheight,keepaspectratio]{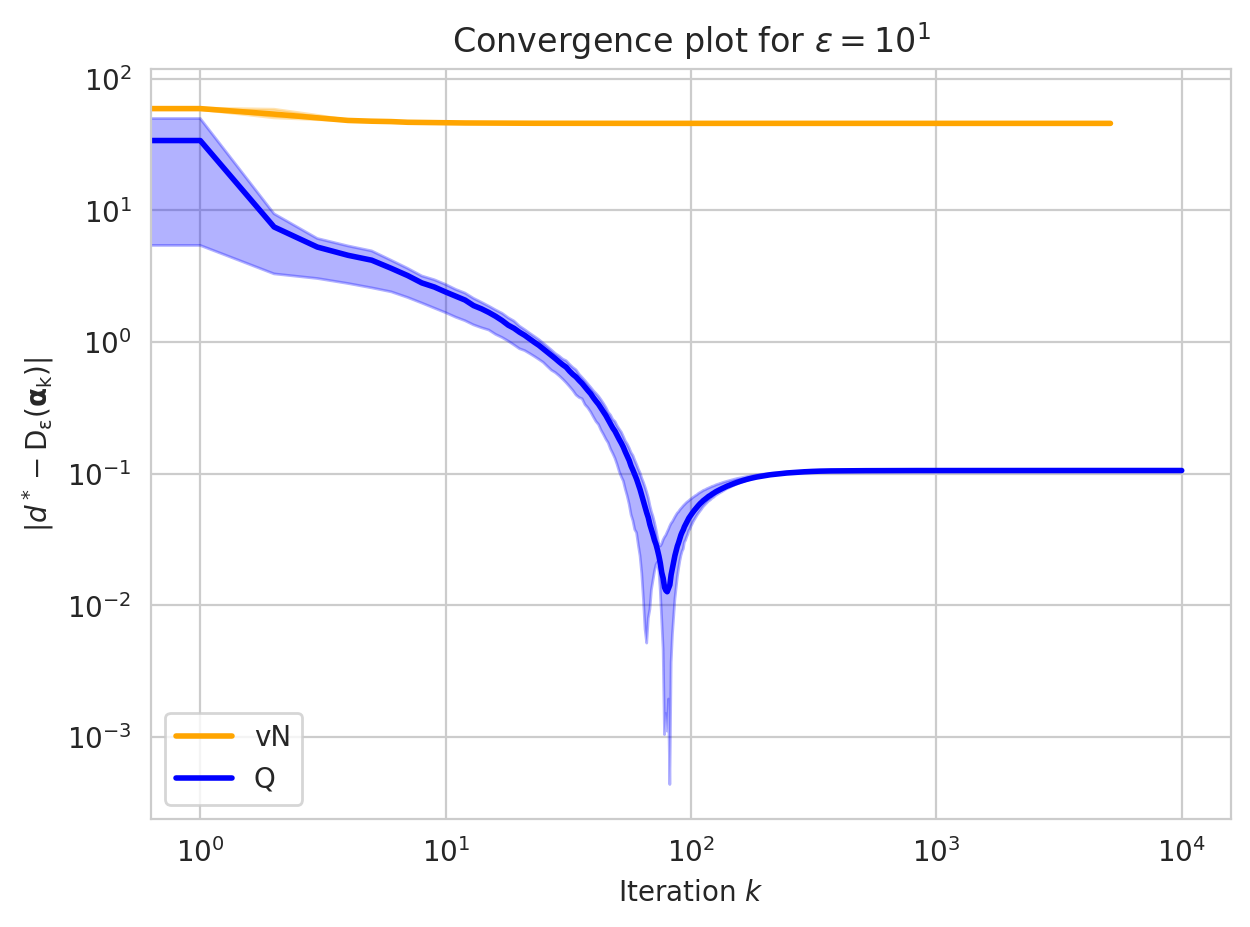}
    \caption{\textbf{QWD} for $\varepsilon = 10^{1}$}
  \end{subfigure}

  % \caption{Margin plots for the \textbf{QWD} problem instance with larger regularization values.
  % See the beginning of section~\ref{sec:numerics} for detailed explanation.}

  % \caption{\textbf{Quantum Optimal Transport problem instance \textbf{QWD}.} Graph of the absolute difference between the dual functional $\dualf$ in Eq.~\eqref{intro:maindual} and true optimal value. Panels~(a)--(f) display the iteration trajectories for different values of the regularization parameter $\varepsilon \in \{10^{4}, 10^{3}, 10^{1}\}$ for both von Neumann (orange) and quadratic regularization (blue). Increasing the regularization improves convergence of the algorithm at the price of higher bias in objective due to the regularization. In all cases, quaratic regularization gives less biased objective than von Neumann entropy.}
  \caption{\textbf{Quantum Optimal Transport problem instance \textbf{QWD}.} Graph of the absolute difference between the dual functional at iteration $k$ $\dualf(\bm \alpha_k)$ in Eq.~\eqref{intro:maindual} and true optimal value $d^*$. Panels~(a)--(c) display the iteration trajectories for different values of the regularization parameter $\varepsilon 
  \in \{10^{4}, 10^{3}, 10^{1}\}$ 
  for both von Neumann (orange) and quadratic regularization (blue).}
  \label{fig:qot_0a}
\end{figure}

\begin{figure}[!htb]
  \centering

  % --- Row ---
  \begin{subfigure}{0.27\textwidth}
    \centering
    \includegraphics[width=\linewidth,height=0.26\textheight,keepaspectratio]{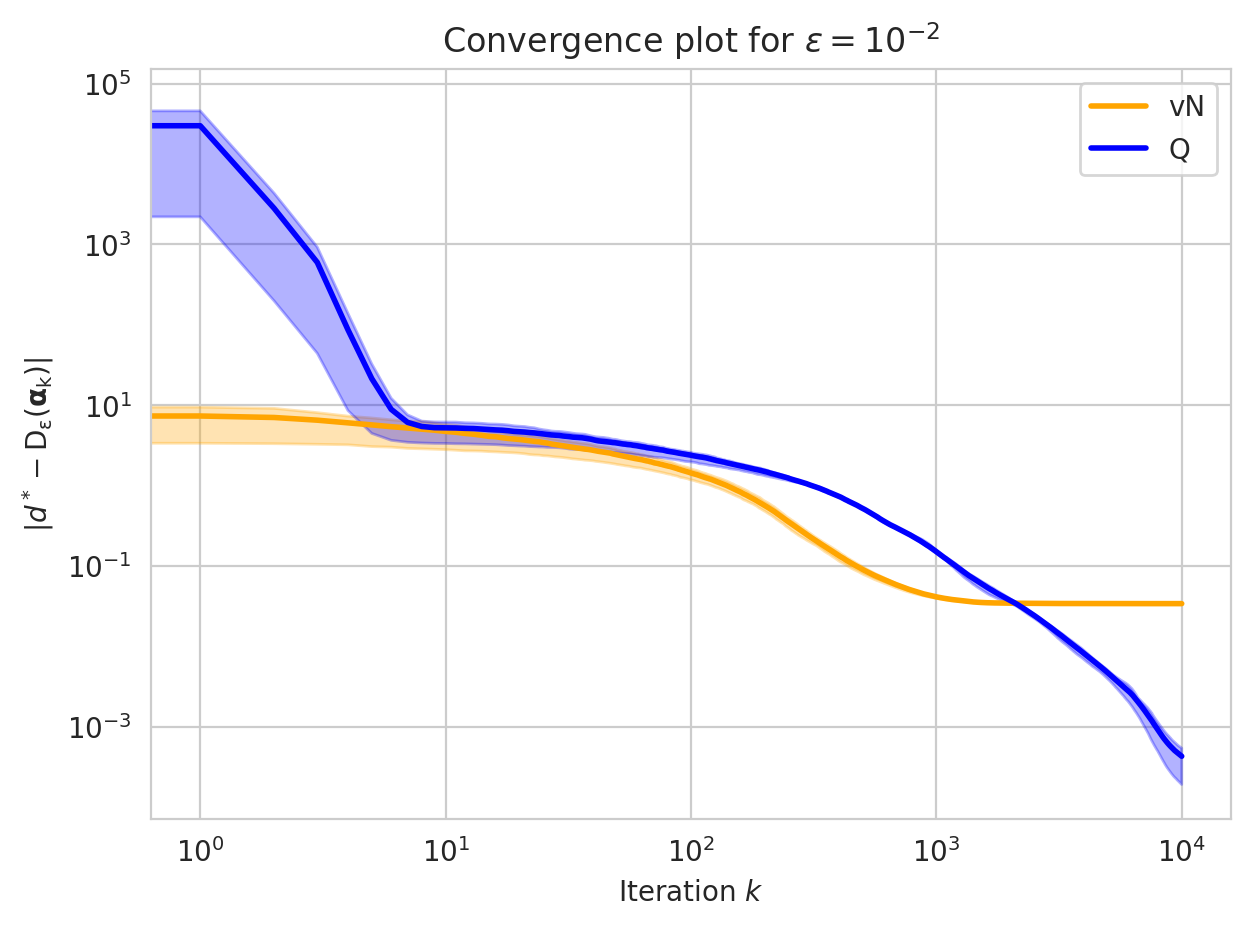}
    \caption{\textbf{QWD} for $\varepsilon = 10^{-2}$}
  \end{subfigure}\hfill
  \begin{subfigure}{0.27\textwidth}
    \centering
    \includegraphics[width=\linewidth,height=0.26\textheight,keepaspectratio]{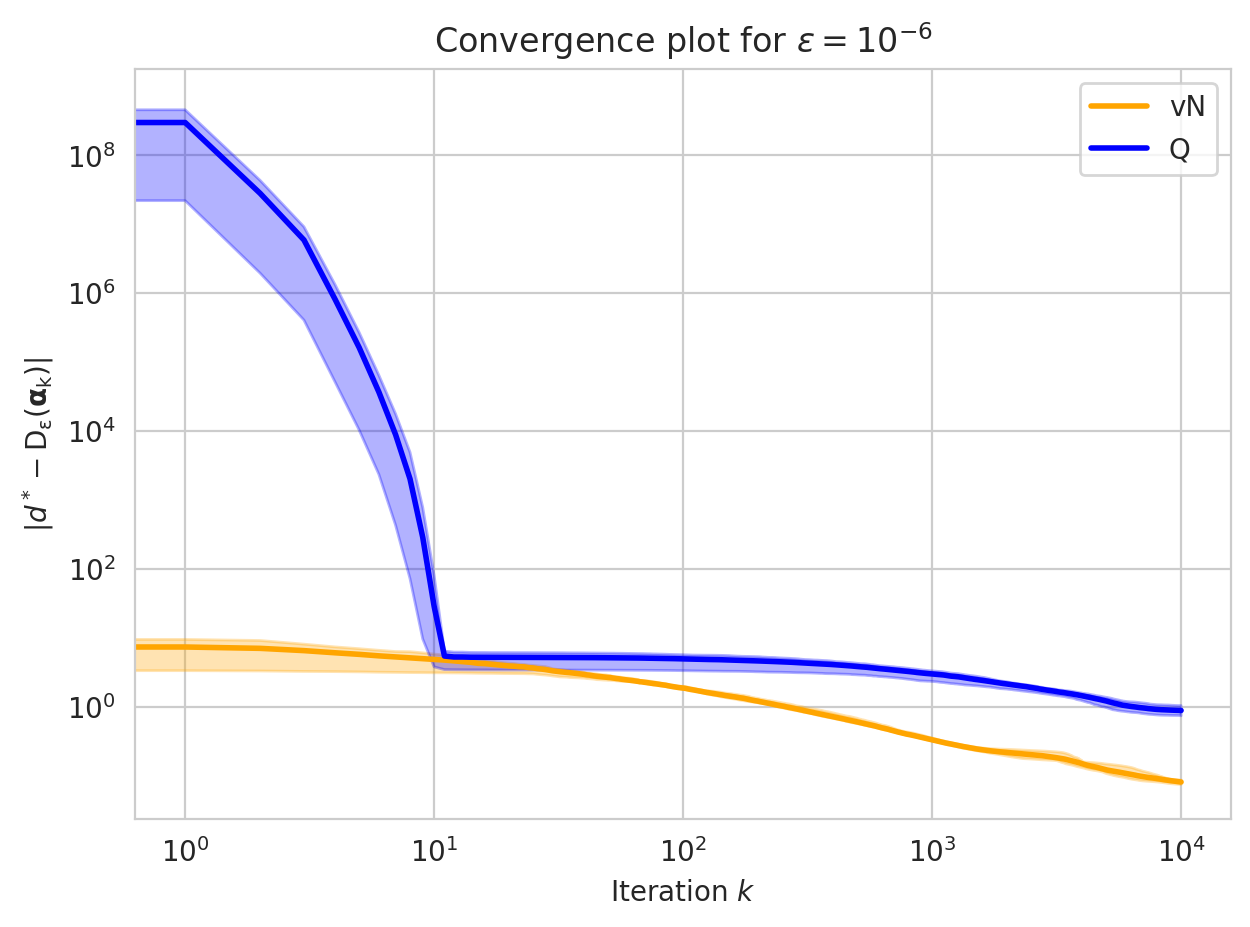}
    \caption{\textbf{QWD} for $\varepsilon = 10^{-6}$}
  \end{subfigure}\hfill
  \begin{subfigure}{0.27\textwidth}
    \centering
    \includegraphics[width=\linewidth,height=0.26\textheight,keepaspectratio]{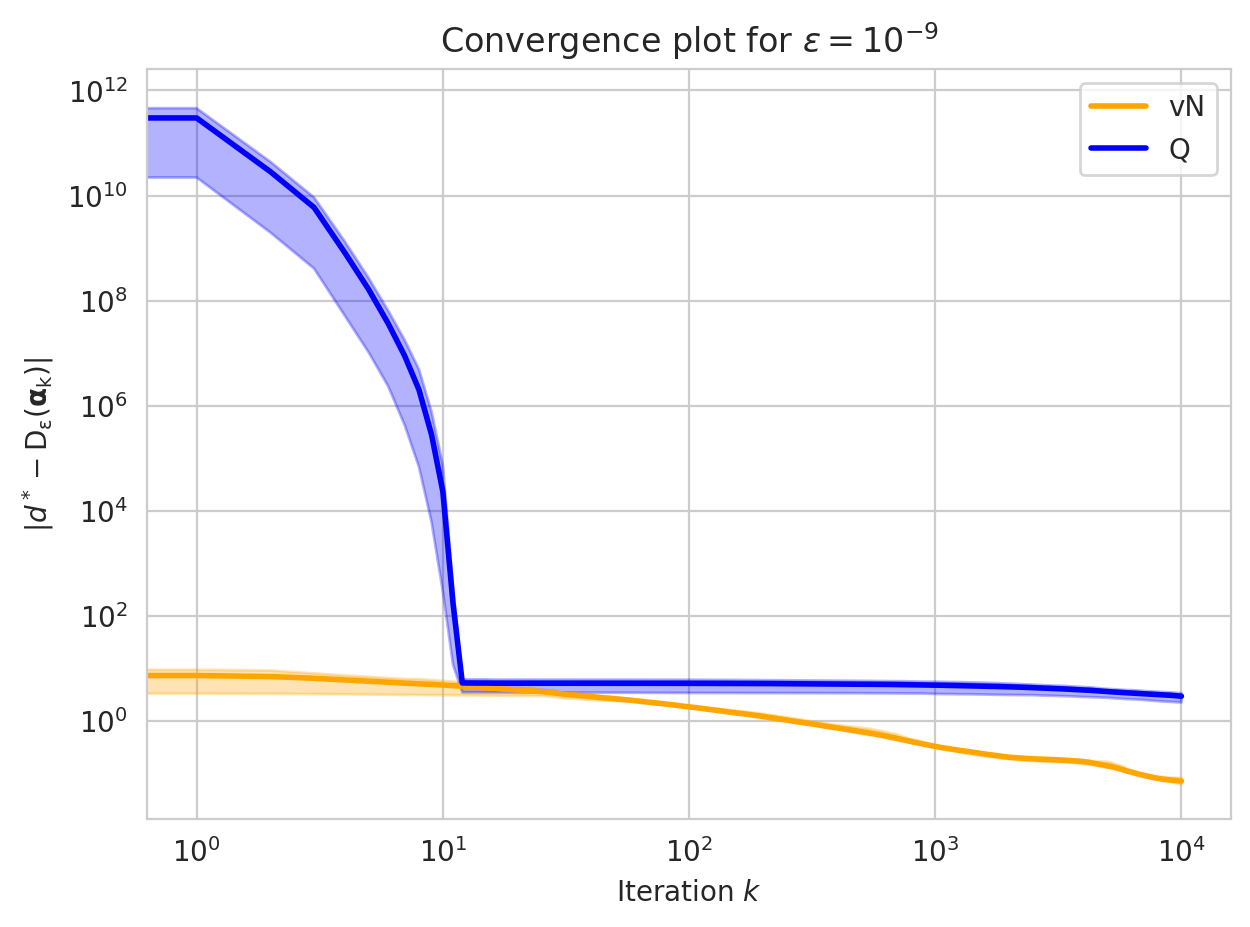}
    \caption{\textbf{QWD} for $\varepsilon = 10^{-9}$}
  \end{subfigure}

  % \caption{Margin plots for the \textbf{QWD} problem instance with smaller regularization values.
  % See the beginning of section~\ref{sec:numerics} for detailed explanation.}
  % \caption{\textbf{Quantum Optimal Transport problem instance \textbf{QWD}.} Graph of the absolute difference between the dual functional $\dualf$ in Eq.~\eqref{intro:maindual} and true optimal value. Panels~(a)--(f) display the iteration trajectories for different values of the regularization parameter $\varepsilon \in \{10^{-2}, 10^{-6}, 10^{-9}\}$ for both von Neumann (orange) and quadratic regularization (blue). Small regularization returns more accurate solution at the price of harder numerical problem. For both regularization types L-BFGS terminated due to iteration limit.}
  \caption{\textbf{Quantum Optimal Transport problem instance \textbf{QWD}.} Graph of the absolute difference between the dual functional at iteration $k$ $\dualf(\bm \alpha_k)$ in Eq.~\eqref{intro:maindual} and true optimal value $d^*$. Panels~(a)--(c) display the iteration trajectories for different values of the regularization parameter $\varepsilon 
  \in \{10^{-2}, 10^{-6}, 10^{-9}\}$ 
  for both von Neumann (orange) and quadratic regularization (blue).}
  \label{fig:qot_0b}
\end{figure}

\begin{table}[ht]
\centering
\footnotesize
\setlength{\tabcolsep}{2pt}

% ---------- Shared header row (only once) ----------
\begin{minipage}[t]{0.31\textwidth}\centering
\begin{tabular}{@{} L T Y N I A @{}}
\toprule Reg. & Tol. & Type & Time & Iters & Ach. \\ \midrule
\end{tabular}
\end{minipage}\hspace{0.02\textwidth}
\begin{minipage}[t]{0.31\textwidth}\centering
\begin{tabular}{@{} L T Y N I A @{}}
\toprule Reg. & Tol. & Type & Time & Iters & Ach. \\ \midrule
\end{tabular}
\end{minipage}\hspace{0.02\textwidth}
\begin{minipage}[t]{0.31\textwidth}\centering
\begin{tabular}{@{} L T Y N I A @{}}
\toprule Reg. & Tol. & Type & Time & Iters & Ach. \\ \midrule
\end{tabular}
\end{minipage}

\vspace{0.1em}

% ---------- Row 1 ----------
\begin{minipage}[t]{0.31\textwidth}\centering
\begin{tabular}{@{} L T Y N I A @{}}
\multirow{4}{*}{\textbf{$10^{4}$}} & $10^{-3}$ & \texttt{vN} & 12.0 & 74 & Yes \\
 % & $10^{-4}$ & \texttt{vN} & 50.3 & 309 & Yes \\
 % & $10^{-5}$ & \texttt{vN} & 192.0 & 1180 & Yes \\
 & $10^{-6}$ & \texttt{vN} & 510.7 & 3139 & Yes \\
 & $10^{-3}$ & \texttt{Q} & 3.1 & 20 & Yes \\
 % & $10^{-4}$ & \texttt{Q} & 8.5 & 55 & Yes \\
 % & $10^{-5}$ & \texttt{Q} & 53.1 & 342 & Yes \\
 & $10^{-6}$ & \texttt{Q} & 230.3 & 1484 & Yes \\
\bottomrule
\end{tabular}
\end{minipage}\hspace{0.02\textwidth}
\begin{minipage}[t]{0.31\textwidth}\centering
\begin{tabular}{@{} L T Y N I A @{}}
\multirow{4}{*}{\textbf{$10^{3}$}} & $10^{-3}$ & \texttt{vN} & 12.2 & 76 & Yes \\
 % & $10^{-4}$ & \texttt{vN} & 46.2 & 288 & Yes \\
 % & $10^{-5}$ & \texttt{vN} & 177.7 & 1109 & Yes \\
 & $10^{-6}$ & \texttt{vN} & 639.8 & 3992 & Yes \\
 & $10^{-3}$ & \texttt{Q} & 10.2 & 65 & Yes \\
 % & $10^{-4}$ & \texttt{Q} & 27.6 & 176 & Yes \\
 % & $10^{-5}$ & \texttt{Q} & 104.6 & 667 & Yes \\
 & $10^{-6}$ & \texttt{Q} & 383.7 & 2446 & Yes \\
\bottomrule
\end{tabular}
\end{minipage}\hspace{0.02\textwidth}
\begin{minipage}[t]{0.31\textwidth}\centering
\begin{tabular}{@{} L T Y N I A @{}}
\multirow{4}{*}{\textbf{$10^{1}$}} & $10^{-3}$ & \texttt{vN} & 14.0 & 86 & Yes \\
 % & $10^{-4}$ & \texttt{vN} & 45.0 & 277 & Yes \\
 % & $10^{-5}$ & \texttt{vN} & 194.9 & 1200 & Yes \\
 & $10^{-6}$ & \texttt{vN} & 656.1 & 4040 & Yes \\
 & $10^{-3}$ & \texttt{Q} & 63.2 & 406 & Yes \\
 % & $10^{-4}$ & \texttt{Q} & 211.6 & 1359 & Yes \\
 % & $10^{-5}$ & \texttt{Q} & 795.2 & 5108 & Yes \\
 & $10^{-6}$ & \texttt{Q} & 1556.8 & 10000 & No \\
\bottomrule
\end{tabular}
\end{minipage}

\vspace{0.3em}

% ---------- Row 2 ----------
\begin{minipage}[t]{0.31\textwidth}\centering
\begin{tabular}{@{} L T Y N I A @{}}
\multirow{4}{*}{\textbf{$10^{-2}$}} & $10^{-3}$ & \texttt{vN} & 364.8 & 2296 & Yes \\
 % & $10^{-4}$ & \texttt{vN} & 1343.6 & 8457 & Yes \\
 % & $10^{-5}$ & \texttt{vN} & 1588.8 & 10000 & No \\
 & $10^{-6}$ & \texttt{vN} & 1588.8 & 10000 & No \\
 & $10^{-3}$ & \texttt{Q} & 1572.0 & 10000 & No \\
 % & $10^{-4}$ & \texttt{Q} & 1572.0 & 10000 & No \\
 % & $10^{-5}$ & \texttt{Q} & 1572.0 & 10000 & No \\
 & $10^{-6}$ & \texttt{Q} & 1572.0 & 10000 & No \\
\bottomrule
\end{tabular}
\end{minipage}\hspace{0.02\textwidth}
\begin{minipage}[t]{0.31\textwidth}\centering
\begin{tabular}{@{} L T Y N I A @{}}
\multirow{4}{*}{\textbf{$10^{-6}$}} & $10^{-3}$ & \texttt{vN} & 1623.3 & 10000 & No \\
 % & $10^{-4}$ & \texttt{vN} & 1623.3 & 10000 & No \\
 % & $10^{-5}$ & \texttt{vN} & 1623.3 & 10000 & No \\
 & $10^{-6}$ & \texttt{vN} & 1623.3 & 10000 & No \\
 & $10^{-3}$ & \texttt{Q} & 2100.2 & 10000 & No \\
 % & $10^{-4}$ & \texttt{Q} & 2100.2 & 10000 & No \\
 % & $10^{-5}$ & \texttt{Q} & 2100.2 & 10000 & No \\
 & $10^{-6}$ & \texttt{Q} & 2100.2 & 10000 & No \\
\bottomrule
\end{tabular}
\end{minipage}\hspace{0.02\textwidth}
\begin{minipage}[t]{0.31\textwidth}\centering
\begin{tabular}{@{} L T Y N I A @{}}
\multirow{4}{*}{\textbf{$10^{-9}$}} & $10^{-3}$ & \texttt{vN} & 1679.5 & 10000 & No \\
 % & $10^{-4}$ & \texttt{vN} & 1679.5 & 10000 & No \\
 % & $10^{-5}$ & \texttt{vN} & 1679.5 & 10000 & No \\
 & $10^{-6}$ & \texttt{vN} & 1679.5 & 10000 & No \\
 & $10^{-3}$ & \texttt{Q} & 2033.7 & 10000 & No \\
 % & $10^{-4}$ & \texttt{Q} & 2033.7 & 10000 & No \\
 % & $10^{-5}$ & \texttt{Q} & 2033.7 & 10000 & No \\
 & $10^{-6}$ & \texttt{Q} & 2033.7 & 10000 & No \\
\bottomrule
\end{tabular}
\end{minipage}

% \caption{Performance table for the \textbf{QWD} problem instance. \textbf{Reg.} - regularization parameter; \textbf{Type} - regularization used. \texttt{vN} (entropy) or \texttt{Q} (quadratic); \textbf{Time} - seconds; \textbf{Ach.} - if reached tolerance.}
% \caption{
% Performance table for the \textbf{QWD} problem instance.
% \textbf{Reg.} - regularization parameter; \textbf{Type} - regularization used (\texttt{vN} entropy or \texttt{Q} quadratic);
% \textbf{Time} - seconds; \textbf{Ach.} - whether the target tolerance was reached.
% Decreasing~$\varepsilon$ substantially increases the numerical difficulty of the dual optimization: both runtime and iteration count grow quickly, and for $\varepsilon \leq 10^{-2}$ the algorithm frequently reaches the iteration limit of $10^{4}$ without achieving the prescribed tolerance.
% Quadratic regularization becomes noticeably harder to optimize than von Neumann entropy as~$\varepsilon$ decreases, and small regularization parameter prevents both methods from converging within the available iteration budget.
% }
\caption{
Performance results for the \textbf{QWD} problem instance. Each block corresponds to a fixed regularization parameter (\textbf{Reg.}) and compares von Neumann entropy (\texttt{vN}, $\varphi(z)=z\log z$) and quadratic (\texttt{Q}, $\varphi(z)=\tfrac12 z^2$) regularizers across decreasing tolerances. Reported are runtime (\textbf{Time}, s), iteration count (\textbf{Iter.}), and tolerance attainment (\textbf{Ach.}).
}

\label{tab:qot_0}
\end{table}

On the Figure~\ref{fig:qot_0a}~-~\ref{fig:qot_0b}, we plot absolute difference between the dual value at iteration $k$ $\dualf(\bm \alpha_k)$ and the value $d^*$ of~\eqref{eq:qot_dual}.
Both axes are shown on a logarithmic scale.
These plots aim to demonstrate the bias that two different regularizations introduce.
On the Figure~\ref{fig:qot_0a} we can see the early termination due to convergence, whereas Figure~\ref{fig:qot_0b} demonstrates insufficient computational budget to reach optimal value yet giving more accurate estimation of $d^*$.
This suggests the trade-off: increasing the regularization accelerates convergence to the optimum, but introduces a larger bias relative to the unregularized problem.
Unlike to quantum tomography, QOT problem is not invariant w.r.t. regularization strength, and $\varepsilon$ plays crucial role both theoretically and practically. 
In particular, weaker regularization makes the optimization problem significantly more difficult to solve due to the diminishing strict concavity of the objective.

The data shown in the Table~\ref{tab:qot_0} describes the time and number of iterations needed for L-BFGS to reach fixed tolerance $\{10^{-3}, 10^{-6}\}$.
The Table confirms the conclusion suggesting that reaching the same tolerance becomes more costly for smaller $\varepsilon$.

\paragraph{Quantum Optimal Transport with Ising Hamiltonian (IM).}
Another fundamental class of physically motivated test instances arises from Ising-type Hamiltonians, which model spin--spin interactions in multi-qubit systems. 
Let each local subsystem be a qubit with a Hilbert space $\mathcal{H}_i = \mathbb{C}^2$, and let the total Hilbert space be a tensor product
$
\mathcal{H} = \bigotimes_{i=1}^{2 N} \mathcal{H}_i,
$
where $N$ denotes the numbers of qubits in each of the two interacting subsystems. 
The Hamiltonian of the two-dimensional Ising model is given by
\begin{align*}
    H_{\mathrm{Ising}}
    = - \sum_{i = 1}^{2 N - 1} J_{i, i + 1}\, Z_i Z_{i + 1}
      - \sum_{i = 1}^{2 N} h_i X_i,
\end{align*}
where the first term describes nearest-neighbor spin coupling with interaction strengths $J_{i, i + 1} \in \mathbb{R}$, and the second term encodes local transverse magnetic fields $h_i$. 
Here, $Z_i$ and $X_i$ denote Pauli matrices acting nontrivially on the $i$-th qubit and as the identity elsewhere.

For simplicity, we fix uniform couplings $J_{i, i + 1} = 1$ and homogeneous magnetic field $h_i = h$, considering a total of $2 N$ spins. 
The corresponding Hamiltonian matrix $H_{\mathrm{Ising}} \in \mathrm{H}(\mathbb{C}^{2^{2 N}})$ defines a Hermitian cost operator for the Quantum Optimal Transport (QOT) problem~\eqref{eq:qot_primal}. 
We compute its eigendecomposition
\[
H_{\mathrm{Ising}} = \sum_{k=1}^{2^{2 N}} \lambda_k \ket{\psi_k}\bra{\psi_k},
\]
where $\lambda_1 \leq \dots, \leq \lambda_{2^{2 N}}$,
and select the eigenstate corresponding to the first excited energy level, $\ket{\psi_1}$, to serve as a representative pure state of the system.

Given a partition of the system into two subsystems, 
% with Hilbert spaces $\mathcal{H}_1 = (\mathbb{C}^2)^{\otimes N}$ and $\mathcal{H}_2 = (\mathbb{C}^2)^{\otimes N_2}$, 
the corresponding marginals are defined as the reduced density matrices obtained by partial tracing
\begin{align*}
    \rho &= \Tr_2 \big[\ket{\psi_1}\bra{\psi_1}\big], &
    \sigma &= \Tr_1 \big[\ket{\psi_1}\bra{\psi_1}\big],
\end{align*}
$\rho, \sigma \in \rmH_\geq(\mathbb{C}^{2^N})$.
The pair $(\rho, \sigma)$ then defines the input data for the QOT problem~\eqref{eq:qot_primal} associated with the cost matrix $H_{\mathrm{Ising}}$.

In our tests, we fix $N = 5$, yielding two subsystems of five qubits each, and set $h = 0.5$. 

On the Figure~\ref{fig:qot_1a}~-~\ref{fig:qot_1b}, we plot absolute difference between the dual value at iteration $k$ $\dualf(\bm \alpha_k)$ and the value $d^*$ of~\eqref{eq:qot_dual}.
Both axes are shown on a logarithmic scale.
These plots aim to demonstrate the bias that two different regularizations introduce.
Similar to \textbf{QWD} instance, larger regularization leads to faster convergence, as we can see on the Figure~\ref{fig:qot_1a}.
And, Figure~\ref{fig:qot_1b} shows that small regularization improves accuracy but fails to reach the tolerance in the gradient.

Table~\ref{tab:qot_1} reports the runtime and iteration counts required by L-BFGS to reach the prescribed tolerances $10^{-3}$ and $10^{-6}$. The results confirm that attaining a fixed tolerance becomes increasingly costly as $\varepsilon$ decreases.

\begin{figure}[!htb]
  \centering

  % --- Row ---
  \begin{subfigure}{0.27\textwidth}
    \centering
    \includegraphics[width=\linewidth,height=0.26\textheight,keepaspectratio]{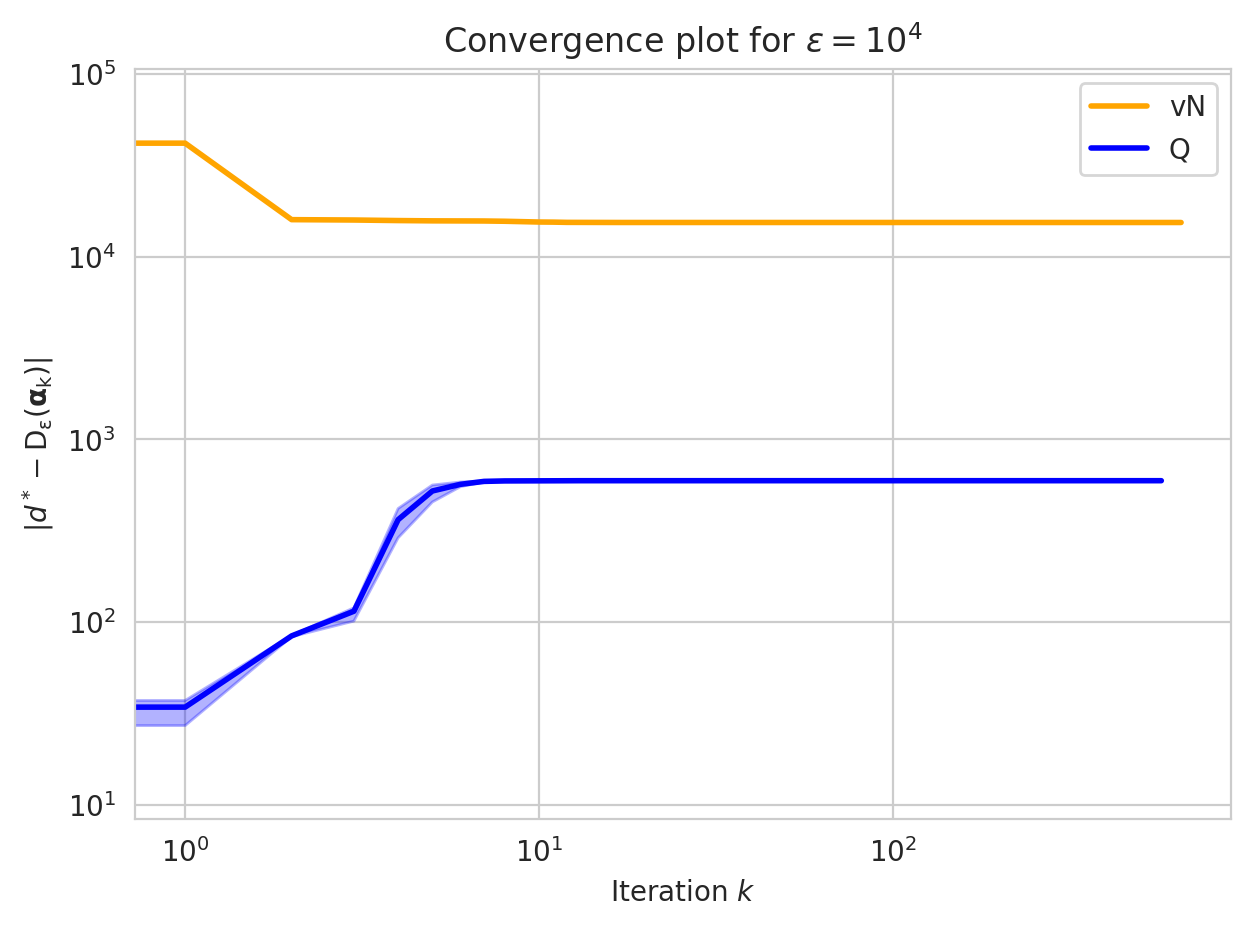}
    \caption{\textbf{IM} for $\varepsilon = 10^{4}$}
  \end{subfigure}\hfill
  \begin{subfigure}{0.27\textwidth}
    \centering
    \includegraphics[width=\linewidth,height=0.26\textheight,keepaspectratio]{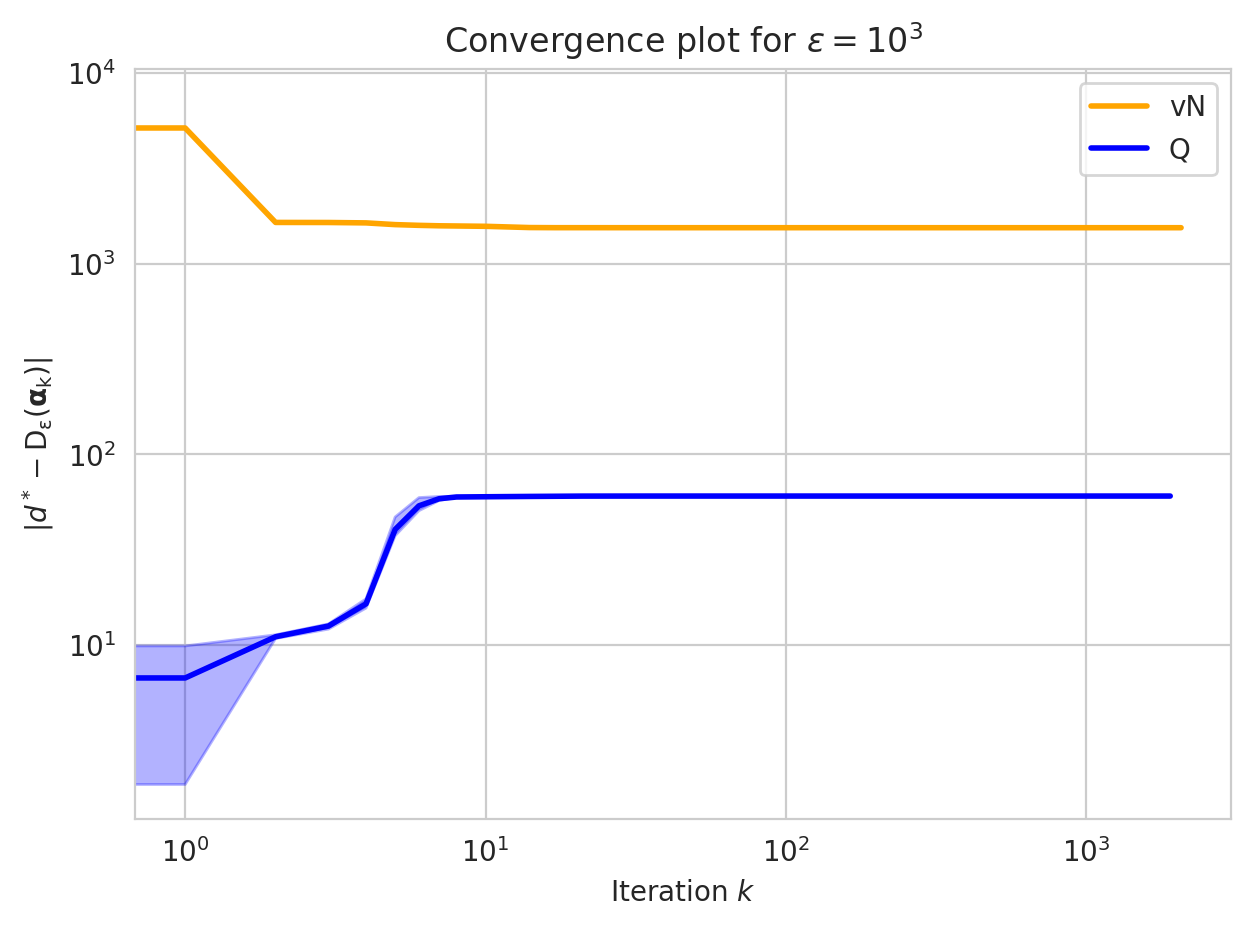}
    \caption{\textbf{IM} for $\varepsilon = 10^{3}$}
  \end{subfigure}\hfill
  \begin{subfigure}{0.27\textwidth}
    \centering
    \includegraphics[width=\linewidth,height=0.26\textheight,keepaspectratio]{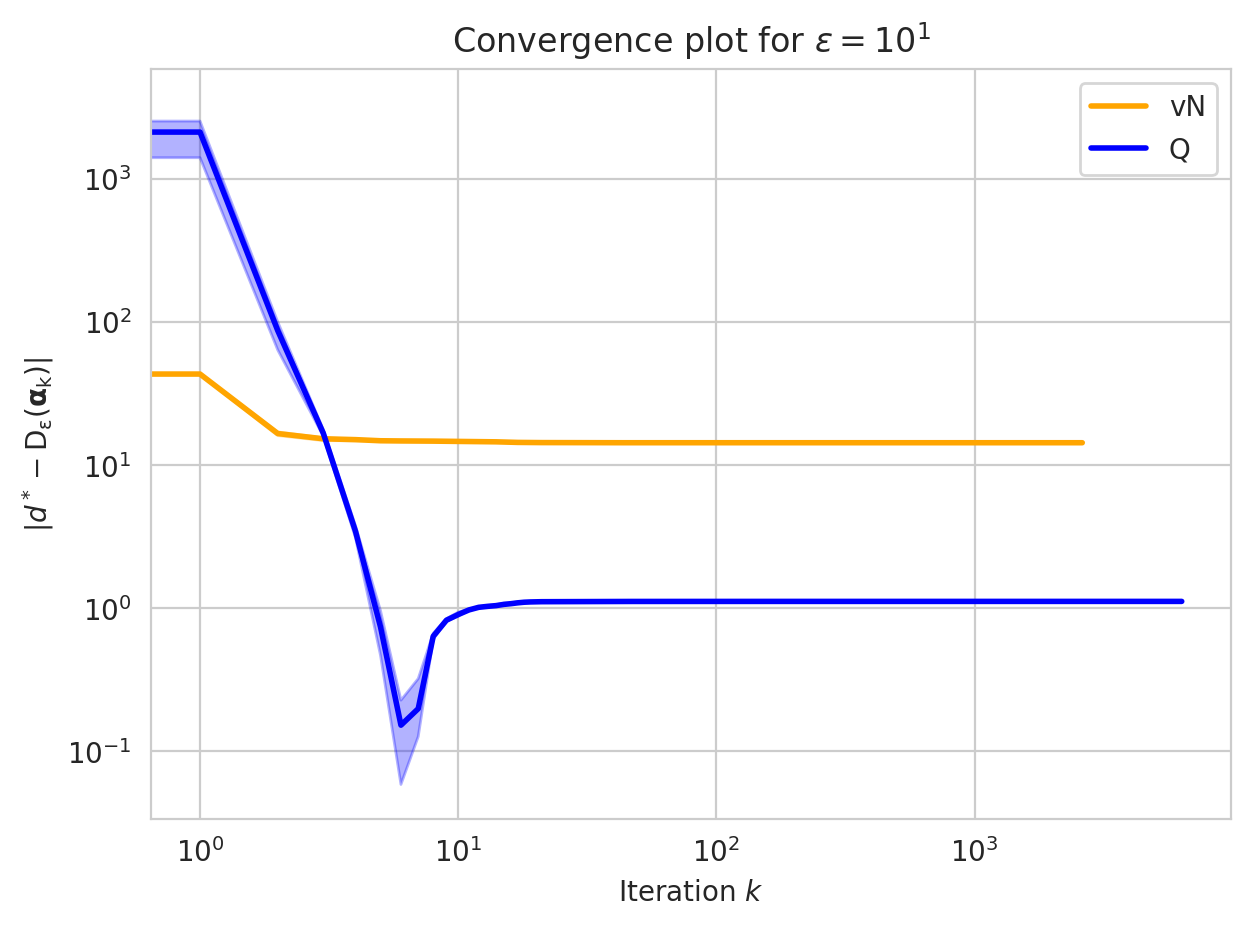}
    \caption{\textbf{IM} for $\varepsilon = 10^{1}$}
  \end{subfigure}

  % \caption{\textbf{Quantum Optimal Transport problem instance \textbf{IM}.} Graph of the absolute difference between the dual functional $\dualf$ in Eq.~\eqref{intro:maindual} and true optimal value. Panels~(a)--(f) display the iteration trajectories for different values of the regularization parameter $\varepsilon \in \{10^{4}, 10^{3}, 10^{1}\}$ for both von Neumann (orange) and quadratic regularization (blue). Increasing the regularization improves convergence of the algorithm at the price of higher bias in objective due to the regularization. In all cases, quaratic regularization gives less biased objective than von Neumann entropy.}
  \caption{\textbf{Quantum Optimal Transport problem instance \textbf{IM}.} Graph of the absolute difference between the dual functional at iteration $k$ $\dualf(\bm \alpha_k)$ in Eq.~\eqref{intro:maindual} and true optimal value $d^*$. Panels~(a)--(c) display the iteration trajectories for different values of the regularization parameter $\varepsilon 
  \in \{10^{4}, 10^{3}, 10^{1}\}$ 
  for both von Neumann (orange) and quadratic regularization (blue).}
  \label{fig:qot_1a}
\end{figure}

\begin{figure}[!htb]
  \centering

  % --- Row ---
  \begin{subfigure}{0.27\textwidth}
    \centering
    \includegraphics[width=\linewidth,height=0.26\textheight,keepaspectratio]{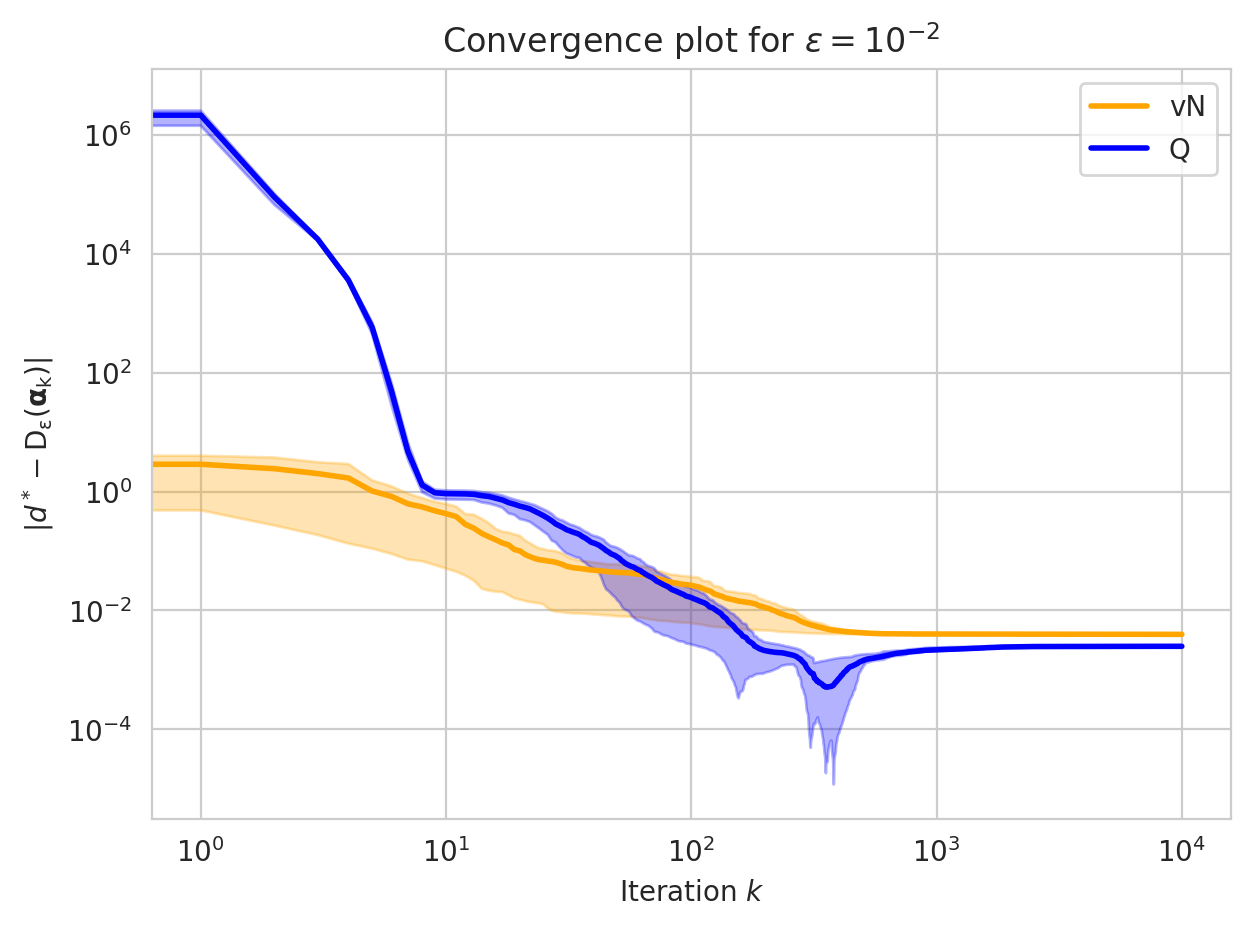}
    \caption{\textbf{IM} for $\varepsilon = 10^{-2}$}
  \end{subfigure}\hfill
  \begin{subfigure}{0.27\textwidth}
    \centering
    \includegraphics[width=\linewidth,height=0.26\textheight,keepaspectratio]{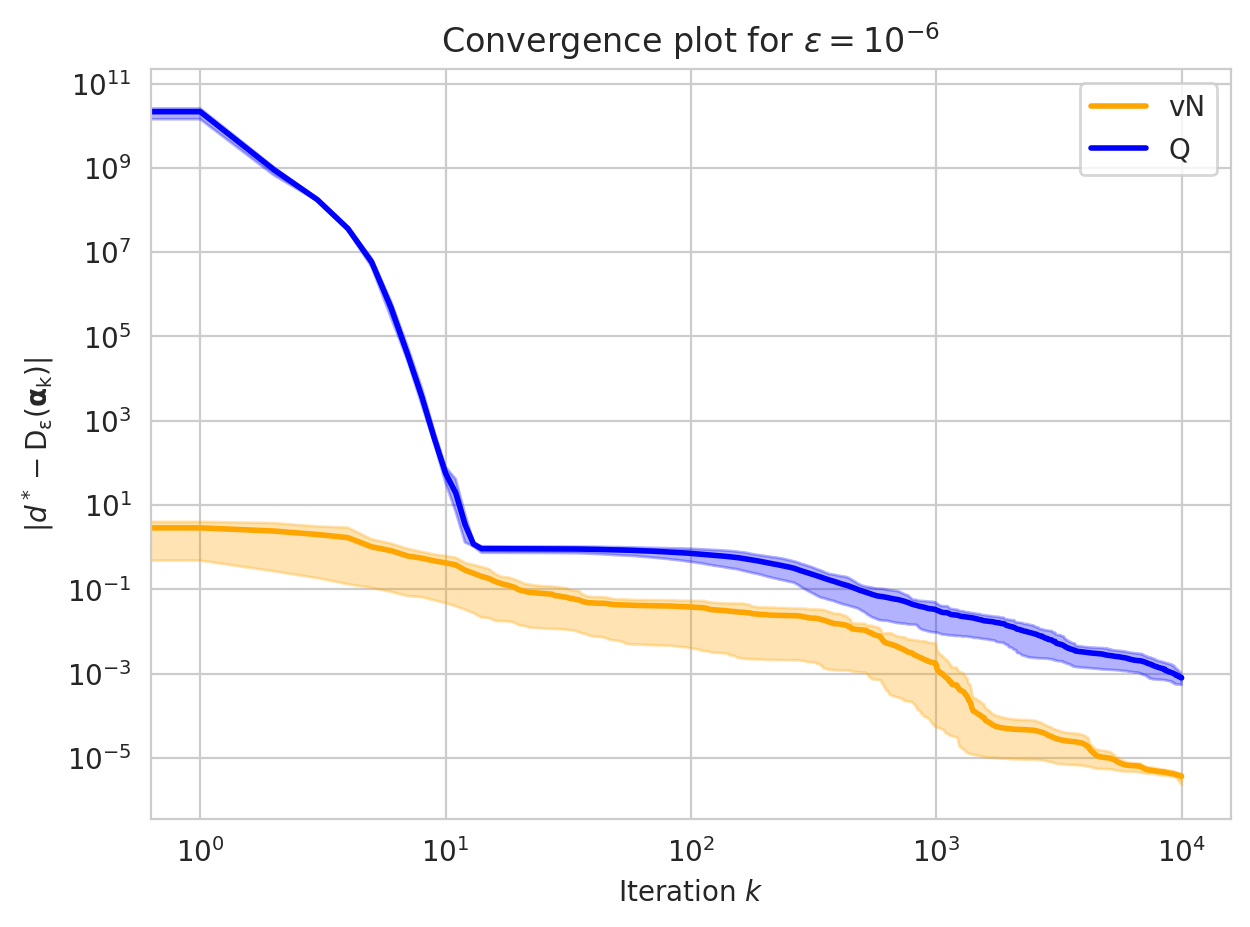}
    \caption{\textbf{IM} for $\varepsilon = 10^{-6}$}
  \end{subfigure}\hfill
  \begin{subfigure}{0.27\textwidth}
    \centering
    \includegraphics[width=\linewidth,height=0.26\textheight,keepaspectratio]{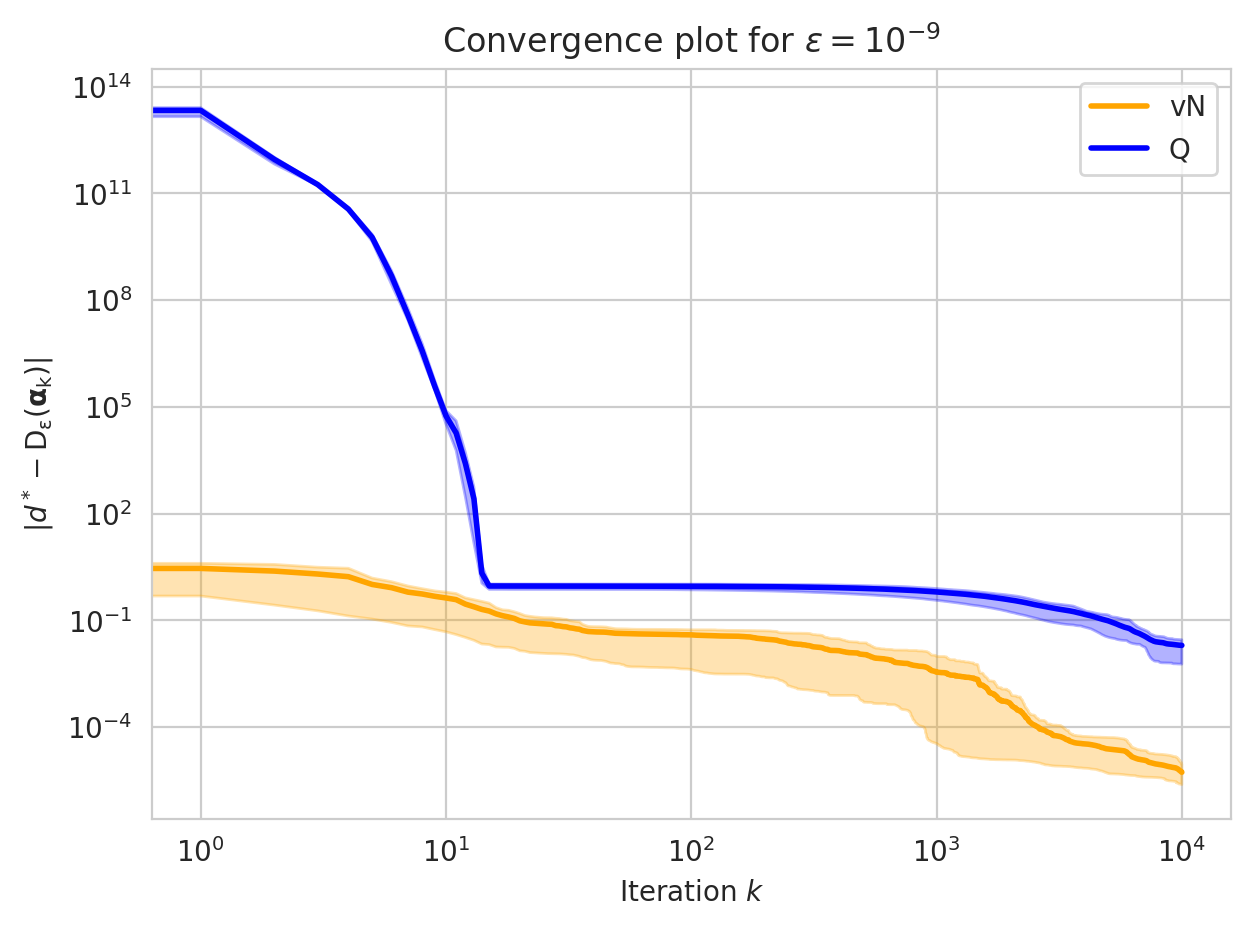}
    \caption{\textbf{IM} for $\varepsilon = 10^{-9}$}
  \end{subfigure}

  % \caption{Margin plots for the \textbf{IM} problem instance with smaller regularization values.
  % See the beginning of section~\ref{sec:numerics} for detailed explanation.}
  % \caption{\textbf{Quantum Optimal Transport problem instance \textbf{IM}.} Graph of the absolute difference between the dual functional $\dualf$ in Eq.~\eqref{intro:maindual} and true optimal value. Panels~(a)--(f) display the iteration trajectories for different values of the regularization parameter $\varepsilon \in \{10^{-2}, 10^{-6}, 10^{-9}\}$ for both von Neumann (orange) and quadratic regularization (blue). Small regularization returns more accurate solution at the price of harder numerical problem. For both regularization types L-BFGS terminated due to iteration limit.}
  \caption{\textbf{Quantum Optimal Transport problem instance \textbf{IM}.} Graph of the absolute difference between the dual functional at iteration $k$ $\dualf(\bm \alpha_k)$ in Eq.~\eqref{intro:maindual} and true optimal value $d^*$. Panels~(a)--(c) display the iteration trajectories for different values of the regularization parameter $\varepsilon 
  \in \{10^{-2}, 10^{-6}, 10^{-9}\}$ 
  for both von Neumann (orange) and quadratic regularization (blue).}
  \label{fig:qot_1b}
\end{figure}

\begin{table}[ht]
\centering
\footnotesize
\setlength{\tabcolsep}{2pt}

% ---------- Shared header row (only once) ----------
\begin{minipage}[t]{0.31\textwidth}\centering
\begin{tabular}{@{} L T Y N I A @{}}
\toprule Reg. & Tol. & Type & Time & Iters & Ach. \\ \midrule
\end{tabular}
\end{minipage}\hspace{0.02\textwidth}
\begin{minipage}[t]{0.31\textwidth}\centering
\begin{tabular}{@{} L T Y N I A @{}}
\toprule Reg. & Tol. & Type & Time & Iters & Ach. \\ \midrule
\end{tabular}
\end{minipage}\hspace{0.02\textwidth}
\begin{minipage}[t]{0.31\textwidth}\centering
\begin{tabular}{@{} L T Y N I A @{}}
\toprule Reg. & Tol. & Type & Time & Iters & Ach. \\ \midrule
\end{tabular}
\end{minipage}

\vspace{0.1em}

% ---------- Row 1 ----------
\begin{minipage}[t]{0.31\textwidth}\centering
\begin{tabular}{@{} L T Y N I A @{}}
\multirow{4}{*}{\textbf{$10^{4}$}} & $10^{-3}$ & \texttt{vN} & 0.6 & 20 & Yes \\
 % & $10^{-4}$ & \texttt{vN} & 0.7 & 23 & Yes \\
 % & $10^{-5}$ & \texttt{vN} & 4.8 & 152 & Yes \\
 & $10^{-6}$ & \texttt{vN} & 8.9 & 281 & Yes \\
 & $10^{-3}$ & \texttt{Q} & 0.6 & 19 & Yes \\
 % & $10^{-4}$ & \texttt{Q} & 0.7 & 24 & Yes \\
 % & $10^{-5}$ & \texttt{Q} & 2.3 & 76 & Yes \\
 & $10^{-6}$ & \texttt{Q} & 17.8 & 576 & Yes \\
\bottomrule
\end{tabular}
\end{minipage}\hspace{0.02\textwidth}
\begin{minipage}[t]{0.31\textwidth}\centering
\begin{tabular}{@{} L T Y N I A @{}}
\multirow{4}{*}{\textbf{$10^{3}$}} & $10^{-3}$ & \texttt{vN} & 0.7 & 24 & Yes \\
 % & $10^{-4}$ & \texttt{vN} & 0.9 & 30 & Yes \\
 % & $10^{-5}$ & \texttt{vN} & 5.3 & 170 & Yes \\
 & $10^{-6}$ & \texttt{vN} & 41.4 & 1336 & Yes \\
 & $10^{-3}$ & \texttt{Q} & 1.0 & 34 & Yes \\
 % & $10^{-4}$ & \texttt{Q} & 1.5 & 49 & Yes \\
 % & $10^{-5}$ & \texttt{Q} & 6.5 & 211 & Yes \\
 & $10^{-6}$ & \texttt{Q} & 23.5 & 765 & Yes \\
\bottomrule
\end{tabular}
\end{minipage}\hspace{0.02\textwidth}
\begin{minipage}[t]{0.31\textwidth}\centering
\begin{tabular}{@{} L T Y N I A @{}}
\multirow{4}{*}{\textbf{$10^{1}$}} & $10^{-3}$ & \texttt{vN} & 1.2 & 38 & Yes \\
 % & $10^{-4}$ & \texttt{vN} & 2.1 & 68 & Yes \\
 % & $10^{-5}$ & \texttt{vN} & 13.5 & 440 & Yes \\
 & $10^{-6}$ & \texttt{vN} & 77.3 & 2518 & Yes \\
 & $10^{-3}$ & \texttt{Q} & 1.9 & 62 & Yes \\
 % & $10^{-4}$ & \texttt{Q} & 6.6 & 215 & Yes \\
 % & $10^{-5}$ & \texttt{Q} & 27.3 & 891 & Yes \\
 & $10^{-6}$ & \texttt{Q} & 96.1 & 3136 & Yes \\
\bottomrule
\end{tabular}
\end{minipage}

\vspace{0.3em}

% ---------- Row 2 ----------
\begin{minipage}[t]{0.31\textwidth}\centering
\begin{tabular}{@{} L T Y N I A @{}}
\multirow{4}{*}{\textbf{$10^{-2}$}} & $10^{-3}$ & \texttt{vN} & 18.8 & 554 & Yes \\
 % & $10^{-4}$ & \texttt{vN} & 32.1 & 944 & Yes \\
 % & $10^{-5}$ & \texttt{vN} & 219.1 & 6442 & Yes \\
 & $10^{-6}$ & \texttt{vN} & 340.1 & 10000 & No \\
 & $10^{-3}$ & \texttt{Q} & 37.8 & 1171 & Yes \\
 % & $10^{-4}$ & \texttt{Q} & 92.6 & 2869 & Yes \\
 % & $10^{-5}$ & \texttt{Q} & 322.8 & 10000 & No \\
 & $10^{-6}$ & \texttt{Q} & 322.8 & 10000 & No \\
\bottomrule
\end{tabular}
\end{minipage}\hspace{0.02\textwidth}
\begin{minipage}[t]{0.31\textwidth}\centering
\begin{tabular}{@{} L T Y N I A @{}}
\multirow{4}{*}{\textbf{$10^{-6}$}} & $10^{-3}$ & \texttt{vN} & 51.2 & 1428 & Yes \\
 % & $10^{-4}$ & \texttt{vN} & 76.1 & 2123 & Yes \\
 % & $10^{-5}$ & \texttt{vN} & 276.2 & 7711 & Yes \\
 & $10^{-6}$ & \texttt{vN} & 358.2 & 10000 & No \\
 & $10^{-3}$ & \texttt{Q} & 384.7 & 10000 & No \\
 % & $10^{-4}$ & \texttt{Q} & 384.7 & 10000 & No \\
 % & $10^{-5}$ & \texttt{Q} & 384.7 & 10000 & No \\
 & $10^{-6}$ & \texttt{Q} & 384.7 & 10000 & No \\
\bottomrule
\end{tabular}
\end{minipage}\hspace{0.02\textwidth}
\begin{minipage}[t]{0.31\textwidth}\centering
\begin{tabular}{@{} L T Y N I A @{}}
\multirow{4}{*}{\textbf{$10^{-9}$}} & $10^{-3}$ & \texttt{vN} & 100.2 & 2571 & Yes \\
 % & $10^{-4}$ & \texttt{vN} & 139.3 & 3575 & Yes \\
 % & $10^{-5}$ & \texttt{vN} & 389.7 & 10000 & No \\
 & $10^{-6}$ & \texttt{vN} & 389.7 & 10000 & No \\
 & $10^{-3}$ & \texttt{Q} & 397.1 & 10000 & No \\
 % & $10^{-4}$ & \texttt{Q} & 397.1 & 10000 & No \\
 % & $10^{-5}$ & \texttt{Q} & 397.1 & 10000 & No \\
 & $10^{-6}$ & \texttt{Q} & 397.1 & 10000 & No \\
\bottomrule
\end{tabular}
\end{minipage}

% \caption{Performance table for the \textbf{IM} problem instance.
% \textbf{Reg.} - regularization parameter; \textbf{Type} - regularization used. \texttt{vN} (entropy) or \texttt{Q} (quadratic);
% \textbf{Time} - seconds; \textbf{Ach.} - if reached tolerance.}
% \caption{
% Performance table for the \textbf{IM} problem instance.
% \textbf{Reg.} - regularization parameter; \textbf{Type} - regularization used (\texttt{vN} entropy or \texttt{Q} quadratic);
% \textbf{Time} - seconds; \textbf{Ach.} - whether the target tolerance was reached.
% Decreasing~$\varepsilon$ substantially increases the numerical difficulty of the dual optimization: both runtime and iteration count grow quickly, and for $\varepsilon \leq 10^{-2}$ the algorithm frequently reaches the iteration limit of $10^{4}$ without achieving the prescribed tolerance.
% Quadratic regularization becomes noticeably harder to optimize than von Neumann entropy as~$\varepsilon$ decreases, and small regularization parameter prevents both methods from converging within the available iteration budget.
% }
\caption{
Performance results for the \textbf{IM} problem instance. Each block corresponds to a fixed regularization parameter (\textbf{Reg.}) and compares von Neumann entropy (\texttt{vN}, $\varphi(z)=z\log z$) and quadratic (\texttt{Q}, $\varphi(z)=\tfrac12 z^2$) regularizers across decreasing tolerances. Reported are runtime (\textbf{Time}, s), iteration count (\textbf{Iter.}), and tolerance attainment (\textbf{Ach.}).
}

\label{tab:qot_1}
\end{table}

\FloatBarrier

\subsection*{Acknowledgements}

A.G., N.M. and P.P. thank the Natural Sciences and Engineering Research Council of Canada (NSERC) for financial support under the Multi-Marginal Optimal Transport grant (Grant No. RGPIN-2022-05207) and the Alliance grant (Grant No. ALLRP/592521-2023), and the Canada Research Chairs Program (Grant No. CRC-2021-00234). A.G and N.M. received support from the National Research Council of Canada (Grant No. AQC-208-1). A.G. and P.P. received support from the MITACS Accelerate Program. Part of this research was conducted while the authors were visiting the Institute for Pure and Applied Mathematics (IPAM), which is supported by the National Science Foundation (Grant No. DMS-1925919). E.C. is supported by the European Union's Horizon 2020 research
and innovation programme (Grant agreement No. 948021).
LP gratefully acknowledges funding from the Deutsche Forschungsgemeinschaft (DFG -- German Research Foundation) under Germany's Excellence Strategy - GZ 2047/1, Projekt-ID 390685813. Financial support by the Deutsche Forschungsgemeinschaft (DFG) within the CRC 1060, at University of Bonn project number 211504053, is also gratefully acknowledged. 
LP is also thankful for the support under the U-GOV project, identification number PSR$\_$LINEA8A$\_$25SMANT$\_$05.

\bibliographystyle{abbrv}
\bibliography{biblio.bib}

\end{document}